\documentclass{theoretics}

\title{Protocols for Quantum Weak Coin Flipping}
\ThCSauthor[CQST,QuICS,IQIM,ULB]{Atul Singh Arora}{atul.singh.arora@gmail.com}[0000-0003-4696-1125]
\ThCSauthor[ULB]{J\'er\'emie Roland}{Jeremie.Roland@ulb.be}[0000-0003-0556-0376]
\ThCSauthor[IdTc,MIS]{Chrysoula Vlachou}{chrysoula.vlachou@tecnico.ulisboa.pt}[0000-0002-4277-4800]
\ThCSauthor[CZ]{Stephan Weis}{stephan.weis@fel.cvut.cz}[0000-0003-1316-9115]

\ThCSaffil[CQST]{CQST, IIIT Hyderabad, Telangana, India}
\ThCSaffil[QuICS]{QuICS, UMIACS, University of Maryland, College Park, Maryland, USA}
\ThCSaffil[IQIM]{IQIM, CMS, Caltech, Pasadena, California, USA} 
\ThCSaffil[ULB]{Universit\'e libre de Bruxelles, Brussels, Belgium}
\ThCSaffil[IdTc]{Instituto de Telecomunica\c c\~oes, Lisbon, Portugal}
\ThCSaffil[MIS]{Departamento de Matem\'atica,
  Instituto Superior T\'ecnico, Universidade de Lisboa, Lisbon, Portugal}
\ThCSaffil[CZ]{Faculty of Electrical Engineering, Czech Technical University in Prague, Czech Republic}
 
\ThCSkeywords{Quantum cryptography, secure two-party computation, information theoretic security}

\ThCSthanks{%
    Parts of this work were presented at the \href{https://dl.acm.org/doi/10.1145/3313276.3316306}{STOC '19} and \href{https://epubs.siam.org/doi/10.1137/1.9781611976465.58}{SODA '21} conferences (see \cref{subsec:RelToPrior}).}
\ThCSshortnames{A.~S.~Arora, J.~Roland, C.~Vlachou, S.~Weis}  
\ThCSshorttitle{Protocols for Quantum Weak Coin Flipping}

\ThCSyear{2025}
\ThCSarticlenum{26}
\ThCSreceived{Mar 8, 2024}
\ThCSrevised{Jul 12, 2025}
\ThCSaccepted{Oct 16, 2025}
\ThCSpublished{Nov 26, 2025}
\ThCSdoicreatedtrue

\usepackage{csquotes}
 \usepackage{subfig}
\usepackage{color}
\usepackage{cancel}
\usepackage{float}

\newcommand{\atul}[1]{#1} %

\newcommand{\dorigin}[0]{\delta_{\textrm{origin}}} %

\AtBeginDocument{\providecommand\cref[1]{\ref{#1}}} %
\usepackage[capitalise,noabbrev,nameinlink]{cleveref}

\usepackage{tikz}
\usetikzlibrary{calc}

\newcommand{\supp}{\text{supp}}
\definecolor{darkblue}{RGB}{0,0,158}

\usepackage[super]{nth}

\usepackage{braket}

\newcommand{\overbar}[1]{\mkern 1.75mu\overline{\mkern-1.75mu#1\mkern-1.75mu}\mkern 1.75mu}

\begin{document}

\maketitle

\begin{abstract}
  Weak coin flipping is an important cryptographic primitive---it is the strongest known secure two-party computation primitive that classically becomes secure only under certain assumptions (e.g. computational hardness), while quantumly there exist protocols that achieve arbitrarily close to perfect security. This breakthrough result was established by Mochon in 2007 [arXiv:0711.4114]. However, his proof relied on the existence of certain unitary operators which was established by a non-constructive argument. Consequently, explicit protocols have remained elusive. In this work, we give exact constructions of related unitary operators. These, together with a new formalism, yield a family of protocols approaching perfect security thereby also simplifying Mochon's proof of existence. We illustrate the construction of explicit weak coin flipping protocols by considering concrete examples (from the aforementioned family of protocols) that are more secure than all previously known protocols.

\end{abstract}

\global\long\def\diag{\text{diag}}%

\global\long\def\tr{\text{tr}}%

\global\long\def\hc{\text{h.c.}}%

\global\long\def\prob{\text{Prob}}%

\global\long\def\Prob{\text{Prob}}%

\global\long\def\ol#1{\overbar{#1}}%

\global\long\def\ob#1{\overbar{#1}}%

\global\long\def\ket#1{\left|#1\right\rangle }%

\global\long\def\bra#1{\left\langle #1\right|}%

\global\long\def\minmax#1#2{\left(#1,#2\right)}%

\section{Introduction\label{sec:intro}}
The problem we study in this paper is easy to state. Suppose
there are two parties, conventionally called Alice and Bob, who are
placed in physically remote locations and can communicate with each
other using a communication channel. They wish to exchange messages
over this channel in order to agree on a random bit, while having \emph{a priori} known opposite preferred outcomes.
This is easy to do\textemdash Alice flips a coin and sends a message
with the outcome to Bob. However, this requires Bob to trust Alice. Can Bob modify
the scheme to be sure that Alice did not cheat? More generally, can
one construct a protocol, which involves an exchange of messages over
a communication channel, to decide on a random bit while ensuring
that an honest party, i.e. one that follows the protocol, cannot
be deceived? It turns out that if one communicates over a classical
communication channel,\footnote{as opposed to a quantum communication channel} then
a cheating party can always force their desired outcome on the honest
party (unless one makes further assumptions, such as computational hardness).
On the other hand, if Alice and Bob use a quantum
communication channel, then protocols solving this problem up to
vanishing errors have been shown to \emph{exist} \cite{Mochon07}. This seminal result was proved in
2007. However, there is a non-constructive part in its analysis, which means that while we know such protocols exist, the protocols themselves remain unknown.  %
In this paper, we build upon the previous
pioneering works to %
construct protocols for
\emph{quantum weak coin flipping}, as this problem is referred to in the literature.

The coin flipping problem was introduced by Blum in 1983 \cite{Blum:1983:CFT:1008908.1008911}. It has since occupied an interesting place in the overall landscape
of cryptography. To overcome the severe limitations of key distribution, public key cryptography was invented \cite{DH,Merkel}. In 1994 it was shown that the widely used---even today---public key cryptosystem RSA \cite{RSA} can be broken using a quantum computer \cite{Shora}.
Interestingly, a decade earlier, a method for performing key distribution using quantum channels
\cite{BB84} was proposed whose security, in principle, relied only
on the validity of the laws of physics. It was thus thought that quantum mechanics could also revolutionise
\emph{secure two-party computation}. This is another branch of cryptography comprising protocols in which two distrustful parties wish to jointly compute a function on their inputs without having to reveal these inputs to each other. Success here, was marred by a cascade of impossibility results. In a central result of (classical) cryptography, it was shown that a primitive called \emph{oblivious transfer} is universal for secure two-party computation \cite{Kilian}. However, there exists no (classical) protocol that offers perfect security for oblivious transfer without relying on further assumptions,
such as computational hardness---classical secure two-party computation with perfect security is thus impossible \cite{Colbeck07}. In fact, it was shown that even if one allows quantum communication, oblivious transfer cannot be implemented with perfect security \cite{Lo97,CKS13}, extinguishing any lingering hope that quantum mechanics could serve as a panacea for cryptography.
\emph{Bit commitment}, a secure two-party computation primitive weaker than oblivious transfer was subsequently targeted, but it too turned out to be impossible---in the same sense---even in the quantum setting \cite{CK11}. This brings us to \emph{coin flipping}, an even weaker secure two-party computation primitive, which comes in two variants: \emph{strong} and \emph{weak} coin flipping. In a coin flipping protocol the two distrustful parties need to establish a shared random bit. For strong coin flipping\footnote{``Strong coin flipping'' is often referred to simply as ``coin flipping'' in the literature.} the preferences of the parties are unknown to each other, whereas in weak coin flipping, the parties have a priori known opposite preferences (as stated earlier).
While strong coin flipping suffered the same fate as that of oblivious transfer and bit commitment \cite{CK09}, weak coin flipping was poised
for fame---it is the strongest known primitive in the two-party
setting which admits no secure classical protocol, but can be implemented
over a quantum channel with near perfect security~\cite{Mochon07}.

More precisely, in a quantum strong coin flipping protocol a dishonest party can successfully cheat with probability at least $\frac{1}{\sqrt{2}}$ \cite{Kitaev03}, and
the best known explicit protocol %
has a cheating probability of $\frac{1}{2}+\frac{1}{4}$
\cite{Ambainis04b}. As for weak coin flipping, the existence of
protocols with arbitrarily-close-to-perfect security was
proved non-constructively, by elaborate successive reductions of the problem based on
the formalism introduced earlier by  Kitaev for the study of strong
coin flipping \cite{Kitaev03}. Consequently, the structure of the protocols whose existence
is proved was lost. A systematic verification led to a simplified proof of existence by  Aharonov, Chailloux, Ganz, Kerenidis and Magnin \cite{ACG+14}. Yet, over a decade later, an explicit, nearly perfectly secure weak coin flipping protocol was missing, despite various approaches ranging
from the distillation of a protocol using the proof of existence to
numerical search \cite{NST14,Nayak2015}.\footnote{%
  The known proof of existence for WCF implies that an exhaustive search, given enough time, will find an explicit WCF protocol. However, the search space is so large that this approach seems infeasible and has, indeed, been unsuccessful so far.}
While an explicit weak coin flipping protocol has remained elusive, several connections have been discovered.
In particular, (nearly) perfect weak coin flipping
provides, via black-box reductions, (nearly) optimal protocols for strong coin flipping \cite{CK09}, bit commitment \cite{CK11} and a variant of oblivious transfer \cite{Chailloux2013a}. It is also used to implement other cryptographic tasks
such as leader election \cite{Ganz2009} and dice rolling \cite{Aharon2010}.

The most significant advance in the study of weak coin flipping (WCF) was the invention
of the so-called point games, attributed to Kitaev  by Mochon \cite{Mochon07}.
They introduced three equivalent formalisms that can be used to describe
WCF protocols and their security properties: (i) Explicit protocols given by pairs of dual semi-definite programs (SDPs),
(ii) Time Dependent Point Games (TDPGs) and (iii) Time Independent Point Games
(TIPGs). The existence of quantum WCF protocols with almost perfect security was established using  TIPGs \cite{Mochon07}. However, the proposal of explicit protocols was hindered by the fact that no constructive method was given for obtaining a
protocol from a TDPG (even though, as we said, protocols and TDPG are equivalent formalisms).

In this work,
we start by constructing a new framework that allows us to convert point games into protocols, granted that we can find unitaries satisfying certain constraints. We use perturbative methods in conjunction with this framework to obtain a protocol with cheating probability $\frac{1}{2}+\frac{1}{10}$, improving the former best known protocol which has cheating probability $\frac{1}{2}+\frac{1}{6}$ \cite{Mochon05}.\footnote{Strictly speaking, these are families of protocols whose cheating probability approaches
  the said value asymptotically.} We then introduce a more systematic method for converting the point games used by Mochon (including the ones approaching perfect security) into explicit unitaries, which, in turn, can be readily converted
into explicit WCF protocols. %
  {%
    Our approach is also simpler, in at least three ways. First, prior works relied on conic duality arguments to show the equivalence between the various formalisms which was crucial to the proof of existence. Since we give exact constructions directly in the SDP formalism, this conic duality argument can be circumvented. Second, even though we do not use this equivalence for our main result, our approach is also equivalent to the various formalisms as the conic duality argument continues to hold in our approach---and is arguably easier to apply as it avoids the subtleties involving closures of cones (as detailed in \cref{subsec:TEFfunctions} and \cref{lem:setequality}). %
    Finally, our approach produces protocols where the message register can be discarded/reset after each round, unlike prior works where the message register had to be held coherent through all rounds of the protocol (see before \cref{subsec:framework}). }
%


\section{Technical Overview}
Below, we briefly introduce the various aforementioned formalisms. We need them in \cref{subsec:Contributions} where we informally describe our contributions. Later, in \cref{sec:PriorArt}, we present these formalisms in more detail, as we subsequently build upon them. %

Let us start with two elementary remarks about WCF.
First, without loss of generality,\footnote{Since in a WCF protocol, the parties have opposite known preferences, this is just a matter of labeling.} one can say that, if
the (bit-valued) outcome of a WCF protocol is $0$ %
it means that Alice won, while Bob wins on outcome $1$. Second, there are four situations which can arise
in a WCF scenario, of which only three are relevant to our discussion. %
Begin with the situation where both Alice and Bob are honest (denoted by HH),
i.e. they both follow the protocol. We want the protocol to
be such that both Alice and Bob (a) win with equal probability and
(b) are in agreement with each other. In the situation where Alice
is honest and Bob is cheating (denoted by HC), the protocol must protect Alice from
a cheating Bob, who tries to convince
her that he has won. His probability of succeeding by
using his best cheating strategy is denoted by $P_{B}^{*}$, where the subscript denotes
the cheating party. The situation
where Bob is honest and Alice is cheating (denoted by CH) naturally points us to the
corresponding definition of $P_{A}^{*}$.  We do not study the CC case, as neither party follows the prescribed protocol.

As an illustration, recall the naïve (trivially insecure) WCF protocol where Alice
flips a coin and reveals the outcome to Bob over the telephone. A cheating
Alice can simply lie and always win against an honest Bob, viz. $P_{A}^{*}=1$. On the other hand, a cheating Bob cannot do anything
to convince Alice that he has won, unless it happens by random chance
on the coin flip. This corresponds to $P_{B}^{*}=\frac{1}{2}$. We say that a protocol has \emph{bias} $\epsilon$ if neither party can force their preferred outcome with probability greater than $1/2+\epsilon$, for $\epsilon\geq 0$. For the aforementioned naïve protocol, the bias is $\epsilon=\max[P_{A}^{*},P_{B}^{*}]-\frac{1}{2}$ which
amounts to $\epsilon=\frac{1}{2}$ (the worst possible). Evidently, protocols that protect one party can be trivially constructed. The real challenge is constructing protocols where neither party is able to cheat against an honest party.
\subsection{The three formalisms}
Given a WCF protocol, it is not a priori clear how the maximum success
probability of a cheating party, $P_{A/B}^{*}$, should
be computed as the strategy space can be dauntingly large. It turns
out that all quantum WCF protocols can be defined using the exchange
of a (quantum) message register interleaved with the parties applying the unitaries
$U_{i}$ locally (see  \cref{fig:General-structure-of}) until a final
measurement---say $\Pi_{A}$ denoting Alice won and $\Pi_{B}$ denoting
Bob won---is made in the end.
\begin{figure}[ht]
  \centering{}\includegraphics[scale=0.875]{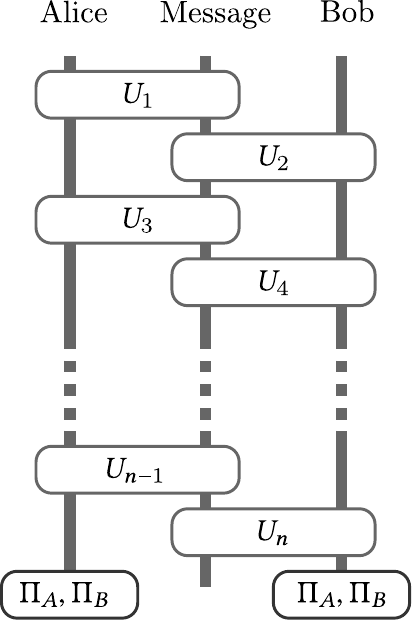}\bigskip{}
  \caption{General structure of a WCF protocol.\label{fig:General-structure-of}}
\end{figure}
Computing $P_{A}^{*}$ in this case reduces
to a semi-definite program (SDP) in $\rho$ (where $\rho$ is the state held by the honest party just before the measurement): maximise $P_{A}^{*}=\text{tr}(\Pi_{A}\rho)$ given the constraint
that the honest party (Bob in this case) follows the protocol. Similarly for computing
$P_{B}^{*}$ we can define another SDP. Using SDP duality one can
turn this maximisation problem over cheating strategies into a minimisation
problem over dual variables $Z_{A/B}$. Any dual feasible assignment (i.e. one that satisfies the constraints ``dual to'' those of the original SDP)
then provides an upper bound on the cheating probabilities $P_{A/B}^{*}$.
Handling SDPs is, in general, straightforward, but in this case, there are two SDPs,
and we must optimise both simultaneously.\footnote{%
  Furthermore, the size of the SDP scales with the dimension of the system, i.e. exponentially in the number of qubits. Therefore, optimising such SDPs in general is unlikely to be a tractable problem.}
Note that we assumed that the protocol is known and we are trying to
bound $P_{A}^{*}$ and $P_{B}^{*}$. However, our goal is
to find good protocols. Therefore, we would like a formalism which
allows us to do both, construct protocols \emph{and} find the associated
$P_{A}^{*}$ and $P_{B}^{*}$. Kitaev and Mochon, gave exactly such a formalism. %

They converted this problem about matrices ($Z$, $\rho$ and $U$)
into a problem about points on a plane, and Mochon called it Kitaev's
``Time Dependent Point Game formalism'' (TDPG). Therein, we are
concerned with a sequence of frames (also referred to as configurations). Each frame is a finite collection of points in the positive quadrant of the $xy$-plane with probability
weights assigned to them. This sequence must start with a fixed frame and end with
a frame that has only one point. The fixed starting frame consists
of two points at
$(0,1)$ and
$(1,0)$ with equal  weights $1/2$. The end frame must be a single point, say at $(\beta,\alpha)$,
with weight $1$. The objective of the protocol designer is to get
this end point as close to the point
$( \frac{1}{2},\frac{1}{2})$
as possible by transitioning through intermediate frames (see  \cref{fig:Point-game-corresponding}) following certain rules.
\begin{figure}[H]
  \centering{}\includegraphics[scale=0.825]{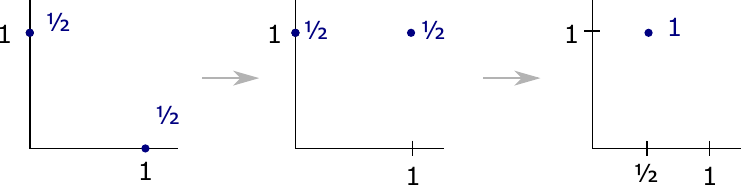}
  \caption{Point game.\label{fig:Point-game-corresponding}}
\end{figure}
\noindent The main theorem about this formalism, roughly stated, asserts that if one abides by these rules, then corresponding to every such sequence of frames, there
exists a WCF protocol with $P_{A}^{*}=\alpha$, $P_{B}^{*}=\beta$.

Let us now describe these rules. Consider a given frame and focus on a
set of points that fall along a vertical (or horizontal) line. Let
the $y$ (or $x$) coordinate  of the $i$th point be given
by $z_{g_{i}}$ and its weight by $p_{g_{i}}$, and let $z_{h_{i}}$
and $p_{h_{i}}$ denote the corresponding quantities for the points in the subsequent
frame. Then, the following conditions must hold:
\begin{enumerate}
  \item the probabilities are conserved, viz. $\sum_{i}p_{g_{i}}=\sum_{i}p_{h_{i}}$, and
  \item for all $\lambda>0$
        \begin{equation}
          \sum_{i}\frac{\lambda z_{g_{i}}}{\lambda+z_{g_{i}}}p_{g_{i}}\le\sum_{i}\frac{\lambda z_{h_{i}}}{\lambda+z_{h_{i}}}p_{h_{i}}.\label{eq:scalarCondition}
        \end{equation}
\end{enumerate}
From one frame to the next, we can either make a horizontal or a vertical transition. By combining these sequentially
we can obtain the desired form of the final frame, i.e. a single
point. The points in the frames and the rules of the transitions arise from
the variables $Z_{A/B}$ of the dual SDP and their constraints, respectively. Just as the state $\rho$ evolves through
the protocol, so do the dual variables $Z_{A/B}$. The points and
their weights in the TDPG are exactly the eigenvalue pairs of $Z_{A/B}$
with the probability weight assigned to them by the honest state $\left|\psi\right\rangle $
at a given step in the protocol. %
Given an explicit WCF protocol and a feasible
assignment for the dual variables witnessing a given bias, it is straightforward
to construct the TDPG. However, going backwards, constructing the
WCF dual from a TDPG is non-trivial and no general construction is
known.

As shall become evident shortly, it is useful to
encode the points on a line and their weights into a function from the
interval $[0,\infty)$ to itself. Let
\begin{equation}
  \left\llbracket a\right\rrbracket (z)=\delta_{a,z},\label{eq:MochonBracketPointGame}
\end{equation}
i.e. $\left\llbracket a\right\rrbracket (z)$
is zero when $z\neq a$ and one when $z=a$.
The \emph{transition} from a given frame to the next is written
as $\sum_{i}p_{g_{i}}\left\llbracket z_{g_{i}}\right\rrbracket \to\sum_{i}p_{h_{i}}\left\llbracket z_{h_{i}}\right\rrbracket $.
The corresponding
\emph{function} is written as $t=\sum_{i}p_{h_{i}}\left\llbracket z_{h_{i}}\right\rrbracket -\sum_{i}p_{g_{i}}\left\llbracket z_{g_{i}}\right\rrbracket $.
If the transition (function) satisfies the conditions (1) and (2)
above, it
is termed as a \emph{valid transition (function)} (see \cref{prop:ebmvalid}). %
If we restrict ourselves to transitions involving only one initial
and one final point, i.e. $\left\llbracket z_{g}\right\rrbracket \to\left\llbracket z_{h}\right\rrbracket $,
the second condition reduces to $z_{g}\le z_{h}$. This is called a \emph{raise}, and it means that we can
increase (but not decrease) the coordinate of a \emph{single} point. %
What about going from one
initial point to many final points, i.e. $\left\llbracket z_{g}\right\rrbracket \to\sum_{i}p_{h_{i}}\left\llbracket z_{h_{i}}\right\rrbracket $?
Note that the points before and after must lie along either a horizontal
or a vertical line. The second condition in this case becomes $1/z_{g}\ge\left\langle 1/z_{h}\right\rangle $,
which means that the harmonic mean
of the final points must be greater than or equal to that of the initial
point, where $\left\langle f(z_{h})\right\rangle:=\Big(\sum_{i}f(z_{h_{i}})p_{h_{i}}\Big)/ \left(\sum_{j}p_{h_j}\right)$.
This is called a \emph{split}. Finally, we can ask what happens upon
merging many points into a single point, i.e. $\sum_{i}p_{g_{i}}\left\llbracket z_{g_{i}}\right\rrbracket \to\left\llbracket z_{h}\right\rrbracket $.
The second condition becomes $\left\langle z_{g}\right\rangle \le z_{h}$,
which means that the final position must not be smaller than the average initial
position. This is called a \emph{merge}. While these three valid transitions
do not exhaust the set of possible valid moves, they are enough to construct games
approaching bias~$1/6$.

Let us consider a simple game
as an example (see  \cref{fig:Point-game-corresponding}). We start with the initial frame and raise the point
$(1,0)$ vertically to
$(1,1)$; this is a raise,
an allowed move. Next we merge
the points
$( 0,1)$ and $(1,1)$ using a horizontal merge. The $x$-coordinate
of the resulting point can at best be $\frac{1}{2}.0+\frac{1}{2}.1=\frac{1}{2}$
where we used the fact that both points have weight $1/2$. Thus, we
end up with a single point having all the weight at
$(\frac{1}{2},1)$.
This formalism tells us that there must exist a protocol which
yields $P_{A}^{*}=1$ while $P_{B}^{*}=\frac{1}{2}$, which is exactly the naïve telephone protocol that we presented earlier. It is a
neat consistency check but it yields the worst possible bias. This is because
we did not use the split move. If we use a split once, we can, by appropriately
matching the weights, already obtain a game with $P_{A}^{*}=P_{B}^{*}=\frac{1}{\sqrt{2}}$.
Various protocols corresponding to this bias were found
\cite{SR02,NS03,KerenidisNayak04} before the point game formalism was
known. In fact, this bias,
$\epsilon=\frac{1}{\sqrt{2}}-\frac{1}{2}$, is exactly the lower
bound for the bias of \emph{strong} coin flipping protocols. It was an exciting time---we
imagine---as the technique used to obtain the bound for strong coin flipping fails to apply to WCF. The matter was not resolved for some time, and this
protocol remained the best known implementation of WCF. Then, in 2005, Mochon showed that using multiple
splits at the beginning followed by a raise, and thereafter simply using merges, one can obtain a game with bias approaching $1/6$ \cite{Mochon05}. %
Obtaining lower biases, however, is not a straightforward extension of the above, and we need other moves which cannot be decomposed into the three basic ones:  splits, merges and raises.

\subsection{Contributions\label{subsec:Contributions}}

\subsubsection{TEF and bias 1/10 protocol}\label{subsubsec:contrtef}

In \cref{sec:TEF}, we provide a framework for converting a TDPG into an explicit WCF
protocol. We start by defining a ``canonical form'' for any given
frame of a TDPG, which
allows us to write the WCF dual variables,
$Z$s, and the honest state $\left|\psi\right\rangle $ associated
with each frame of the TDPG. We then define a sequence of quantum operations,
unitaries and projections, which describe how Alice and Bob transition
from the initial to the final frame. It turns out that there
is only one non-trivial quantum operation, $U$, in the sequence. Using the SDP
formalism we write the constraints at each step of the sequence on
the $Z$s and show that they are indeed satisfied. The aforementioned constraints can be summarised as in \cref{thm:TEFconstraint-inf} below. In \cref{sec:TEF}, one can find the full version, \cref{thm:TEFconstraint}, together with its proof and a detailed description of the framework. { Notice that compared to Mochon's Lemma 18, the key difference in our approach is the introduction of projectors and the treatment of message registers. We defer the details to \cref{sec:TEF}.}

\begin{theorem}[TEF constraint (simplified)]
  \begin{sloppy}If a unitary matrix $U$ acting on the space $\text{span}\{\left|g_{1}\right\rangle ,\left|g_{2}\right\rangle \dots,\left|h_{1}\right\rangle ,\left|h_{2}\right\rangle \dots\}$
    satisfying the constraints\footnote{We use $A\ge B$ to mean that $A-B$ has non-negative
      eigenvalues; we implicitly assume that $A$ and $B$ are Hermitian.}
    \begin{align}
      U\left|v\right\rangle =\left|w\right\rangle \ \ \text{ and }\ \
      \sum_{i}x_{h_{i}}\left|h_{i}\right\rangle \left\langle h_{i}\right|-\sum_{i}x_{g_{i}}E_{h}U\left|g_{i}\right\rangle \left\langle g_{i}\right|U^{\dagger}E_{h}\ge0,\label{eq:constraint}
    \end{align}
    can be found for every transition (see \cref{def:transition} and \cref{def:EBMlineTransition})
    of a TDPG, then an explicit protocol with the corresponding bias can
    be obtained using the TDPG\textendash to\textendash Explicit\textendash protocol
    Framework (TEF). Here, $\{\{\left|g_{i}\right\rangle \},\{\left|h_{i}\right\rangle \}\}$
    are orthonormal vectors. If the transition is horizontal, then
    \begin{itemize}
      \item the initial points have $x_{g_{i}}$ as their $x$-coordinate and
            $p_{g_{i}}$ as their corresponding probability weight,
      \item the final points have $x_{h_{i}}$ as their $x$-coordinate
            and $p_{h_{i}}$ as their corresponding probability weight,
      \item $E_{h}$ is a projection onto the $\text{span}\left\{ \left|h_{i}\right\rangle \right\} $
            space,
      \item $\left|v\right\rangle =\sum_{i}\sqrt{p_{g_{i}}}\left|g_{i}\right\rangle /\sqrt{\sum p_{g_{i}}},$
            $\left|w\right\rangle =\sum_{i}\sqrt{p_{h_{i}}}\left|h_{i}\right\rangle /\sqrt{\sum p_{h_{i}}}$.
    \end{itemize}
    If the transition is vertical, the $x_{g_{i}}$ and $x_{h_{i}}$
    become the $y$-coordinates $y_{g_{i}}$ and $y_{h_{i}}$ with everything
    else unchanged.\label{thm:TEFconstraint-inf}\end{sloppy}
\end{theorem}

The TDPG already specifies the coordinates $x_{h_{i}},x_{g_{i}}$
and the probabilities $p_{h_{i}},p_{g_{i}}$ satisfying the scalar condition Equation~\eqref{eq:scalarCondition}, therefore our task reduces to finding the correct
$U$ which satisfies the matrix constraints Equation~\eqref{eq:constraint}.
Given such a unitary $U$ we show in detail how we can progressively build the sequence of unitaries corresponding to the complete WCF protocol. In fact, we need to reverse the order of the operations in the sequence we get in order to obtain the final protocol.
We continue by introducing what we call the \emph{blinkered
  unitary}, that satisfies the required constraints (as in Equation~\eqref{eq:constraint}) for split and
merge moves. In particular, any
valid transition from $m$ initial to $n$ final points that can be implemented by means of the blinkered unitary, can be seen as a combination of an $m\rightarrow 1$ merge and an $1\rightarrow n$ split (see \cref{subsec:BlinkeredUnitary} and \cref{sec:Blinkered-transition}).
With these the former best known explicit protocol with bias $1/6$ \cite{Mochon05} can already be derived
from its TDPG.  We finally study
the family of TDPGs with bias $1/10$ and isolate the precise moves required
to implement it. These cannot be produced by a combination of merges and splits, therefore, we need to go beyond blinkered unitaries. We give analytic expressions for the required
unitaries and show that they satisfy the corresponding constraints.
This allows us to convert Mochon's family of games with bias $1/10$
into explicit protocols, thus breaking the bias $1/6$ barrier.
However, we essentially guessed the form that the blinkered unitary and the unitaries of the $1/10$ game should have in these cases, and then showed that
they
indeed satisfy the required constraints.
Games achieving
lower biases, though, correspond to larger unitary matrices, therefore
this approach becomes untenable. We overcome this issue in  \cref{sec:1by4k+2}, where we find a way to systematically construct the unitaries for the whole family of Mochon's games achieving bias $\epsilon(k)=1/(4k+2)$ for arbitrary integers $k>0$.

\subsubsection{Exact Unitaries for approaching zero bias using Mochon's assignments\label{subsec:contralgebraic}}

As we saw, TEF allows us to convert any TDPG into an explicit protocol, granted that
the unitaries satisfying Equation~\eqref{eq:constraint} can be found corresponding
to each valid transition
used in the game (see \cref{thm:TEFconstraint-inf}). Using Kitaev's and Mochon's
formalism \cite{Mochon07}, we have that the following---an even weaker
requirement---is enough (see \cref{subsec:fassignmentequivmonomial}): Suppose that a valid function (see the discussion after Equation~\eqref{eq:scalarCondition}), $t$, can be written
as a sum of valid functions. Then, in order to obtain the \emph{effective solution} for $t$ (see \cref{def:solvingassignment}), it suffices to find unitaries corresponding
to the valid functions appearing in the sum.
We consider the class of valid functions that Mochon uses in his family
of point games approaching bias $\epsilon(k)=\frac{1}{4k+2}$ for an arbitrary integer $k>0$.
These are of the form (see \cref{def:f_assignment-f_0_assignment-balanced-m_kmonomial})
\begin{displaymath}
  t=\sum_{i=1}^{n}\frac{-f(x_{i})}{\prod_{j\neq i}(x_{j}-x_{i})}\left\llbracket x_{i}\right\rrbracket,
\end{displaymath}
where $0\le x_{1}<x_{2}\dots<x_{n}\in \mathbb{R}$, $f(x)$ is a polynomial,\footnote{with some restrictions which we suppress for brevity} and the notation is as in Equation~\eqref{eq:MochonBracketPointGame}.
We refer to these as \emph{$f$-assignments} and in particular, when
$f$ is a monomial, we call them \emph{monomial assignments}. We observe that the $f$-assignments
can be expressed as a sum of monomial assignments, and we give formulas for the unitaries corresponding to these monomial assignments. There
are four types of monomial assignments\textemdash which we call balanced or unbalanced (depending on whether the number of points with negative weights in the point game is equal to the number of points with positive weight or not) and
aligned or misaligned (depending on whether the power of the polynomial $f(x)$ is even or odd). The formulas for their \emph{solutions} (see \cref{def:solvingassignment}) and their
proofs of correctness comprise most of \cref{sec:1by4k+2} whose central result is summarised in the following theorem.
\begin{theorem}[informal---we suppressed some constraints on $f$ for brevity.]
  Let $t$ be an $f$-assignment (see \cref{def:f_assignment-f_0_assignment-balanced-m_kmonomial}).
  Then, $t$ can be expressed as $t=\sum_{i}\alpha_{i}t_{i}'$ where
  $\alpha_{i}>0$ and $t_{i}'$ are monomial assignments (see \cref{def:f_assignment-f_0_assignment-balanced-m_kmonomial}).
  Each $t_{i}'$ admits a solution (see \cref{def:solvingassignment}) given in either  \cref{prop:ExactSolnBalancedMonomialAligned},
  \cref{prop:ExactSolnBalancedMonomialMisaligned}, \cref{prop:ExactSolnUnbalancedMonomialAligned}
  or \cref{prop:ExactSolnUnbalancedMonomialMisaligned}, depending on
  the form of $t'_{i}$. \label{thm:Main}
\end{theorem}

In \cref{subsec:1over14} we illustrate, as an example, the construction of a WCF protocol with bias $1/14$ from the corresponding point game by means of the TEF and the analytical solutions to the monomial assignments.

Having found these unitaries, we have effectively solved our problem, since TEF allows the conversion of point games---including the ones with arbitrarily small bias---into WCF protocols with the respective bias { as illustrated in  \cref{fig:argument_outline} below}.
\begin{figure}[ht]
  \begin{centering}
    \includegraphics[scale=1.25]{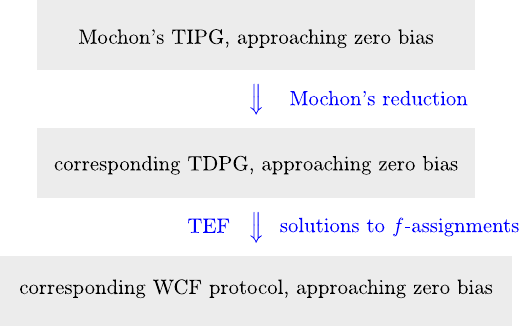}
    \par\end{centering}
  \caption{Mochon constructed a Time Independent Point Game approaching zero bias which, in combination with prior results and the ones in this manuscript, results in the corresponding WCF protocol approaching zero bias. \label{fig:argument_outline}}
\end{figure}

\subsection{Relation to existing pre-prints}\label{subsec:RelToPrior}
{%
  This work is a self-contained and (presently) the most \emph{concise} version of the main result---construction of WCF protocols with vanishing bias---in \texttt{arXiv:1811.02984} \cite{Arora2018} (presented at STOC '19 \cite{Arora2019}) and \texttt{arXiv:1911.13283v2} \cite{Arora2019a} (presented at SODA '21 \cite{Analytic_ARV_2020}). \\
  On the other hand, the \texttt{Cryptology ePrint 2022/1101} \cite{cryptoeprint_ARVW_2022} is a self-contained, \emph{comprehensive} version that contains all the results in \texttt{arXiv:1811.02984} \cite{Arora2018} and \texttt{arXiv:1911.13283v1} \cite{Analytic_ARV_2020} and v2 (v1 gave a geometric construction while v2 was algebraic).}

\section{Preliminaries: Existence of Almost Perfect Quantum WCF Protocols\label{sec:PriorArt}}
The contents of this section are based on two works:
the first is by  Mochon \cite{Mochon07}---part of which is attributed to  Kitaev---and the second is by  Aharonov,
Chailloux, Ganz, Kerenidis and Magnin \cite{ACG+14},
who simplified and verified the former. Here, we only state specific notation and statements (without proofs) from these works that we need to present our work.

\subsection{WCF protocol as an SDP and its dual\label{subsec:sdp}}
Any WCF protocol can be expressed
in the following general form (see
\cite{Ambainis04b} and page 9 of \cite{Mochon07}): %

\begin{definition}[WCF protocol with bias $\epsilon$]
  \label{def:WCFprotocol}For $n$ even, an $n$-message WCF protocol
  between two parties, Alice and Bob, is described by:
  \begin{itemize}
    \item Three Hilbert spaces: $A$ and $B$ corresponding
          to Alice's and Bob's private work-spaces (Bob does not have any access
          to $A$ and, similarly, Alice to $B$) and a message space
          $M$.
    \item An initial product state $\left|\psi_{0}\right\rangle =\left|\psi_{A,0}\right\rangle \otimes\left|\psi_{M,0}\right\rangle \otimes\left|\psi_{B,0}\right\rangle \in A\otimes M\otimes B$.
    \item A set of $n$ unitaries $\{U_{1},\dots U_{n}\}$ acting on $A\otimes M\otimes B$
          with $U_{i}=U_{A,i}\otimes\mathbb{I}_{B}$ for $i$ odd
          and $U_{i}=\mathbb{I}_{A}\otimes U_{B,i}$ for $i$ even.
    \item A set of honest states $\{\left|\psi_{i}\right\rangle :i\in[n]\}$
          defined as $\left|\psi_{i}\right\rangle =U_{i}U_{i-1}\dots U_{1}\left|\psi_{0}\right\rangle $.
    \item A set of $n$ projectors $\{E_{1},\dots E_{n}\}$ acting on $A\otimes M\otimes B$
          with $E_{i}=E_{A,i}\otimes\mathbb{I}_{B}$ for $i$ odd,
          and $E_{i}=\mathbb{I}_{A}\otimes E_{B,i}$ for $i$ even,
          such that $E_{i}\left|\psi_{i}\right\rangle =\left|\psi_{i}\right\rangle $.
    \item Two positive operator valued measures (POVMs) $\{\Pi_{A}^{(0)},\Pi_{A}^{(1)}\}$
          acting on $A$ and $\{\Pi_{B}^{(0)},\Pi_{B}^{(1)}\}$ acting
          on $B$.
  \end{itemize}
  The WCF protocol proceeds as follows:
  \begin{itemize}
    \item In the beginning, Alice holds $\left|\psi_{A,0}\right\rangle \left|\psi_{M,0}\right\rangle $
          and Bob $\left|\psi_{B,0}\right\rangle $.
    \item For $i=1$ to $n$:
          \begin{itemize}
            \item If $i$ is odd, Alice applies $U_{i}$ and measures the resulting
                  state with the POVM $\{E_{i},\mathbb{I}-E_{i}\}$. On the first outcome,
                  she sends the message qubits to Bob; on the second outcome, she
                  ends the protocol by outputting ``0'', i.e, she declares herself
                  the winner.
            \item If $i$ is even, Bob applies $U_{i}$ and measures the resulting state
                  with the POVM $\{E_{i},\mathbb{I}-E_{i}\}$. On the first outcome,
                  he sends the message qubits to Alice; on the second outcome, he ends
                  the protocol by outputting ``1'', i.e., he declares himself the winner.
            \item Alice and Bob measure their part of the state with the final POVM
                  and output the outcome of their measurements. Alice wins on outcome
                  ``0'' and Bob on outcome ``1''.
          \end{itemize}
  \end{itemize}
  The WCF protocol has the following properties:
  \begin{itemize}
    \item \emph{Correctness:} When both parties are honest, their outcomes
          are always the same: \\
          $\Pi_{A}^{(0)}\otimes\mathbb{I}_{M}\otimes\Pi_{B}^{(1)}\left|\psi_{n}\right\rangle =\Pi_{A}^{(1)}\otimes\mathbb{I}_{M}\otimes\Pi_{B}^{(0)}\left|\psi_{n}\right\rangle =0$.
    \item \emph{Balanced:} When both parties are honest, they win with probability
          $1/2$: \\
          $P_{A}=\left|\Pi_{A}^{(0)}\otimes\mathbb{I}_{M}\otimes\Pi_{B}^{(0)}\left|\psi_{n}\right\rangle \right|^{2}=\frac{1}{2}$
          and $P_{B}=\left|\Pi_{A}^{(1)}\otimes\mathbb{I}_{M}\otimes\Pi_{B}^{(1)}\left|\psi_{n}\right\rangle \right|^{2}=\frac{1}{2}.$
    \item \emph{$\epsilon$-biased:} When Alice is honest, the probability that both
          parties agree on Bob winning is $P_{B}^{*}\le\frac{1}{2}+\epsilon$.
          Conversely, when Bob is honest, the probability that both parties
          agree on Alice winning is $P_{A}^{*}\le\frac{1}{2}+\epsilon$.
  \end{itemize}
\end{definition}
For a depiction of the protocol see  \cref{fig:General-protocol-parametrised}.

\begin{figure}[H]
  \begin{centering}
    \includegraphics[scale=0.53]{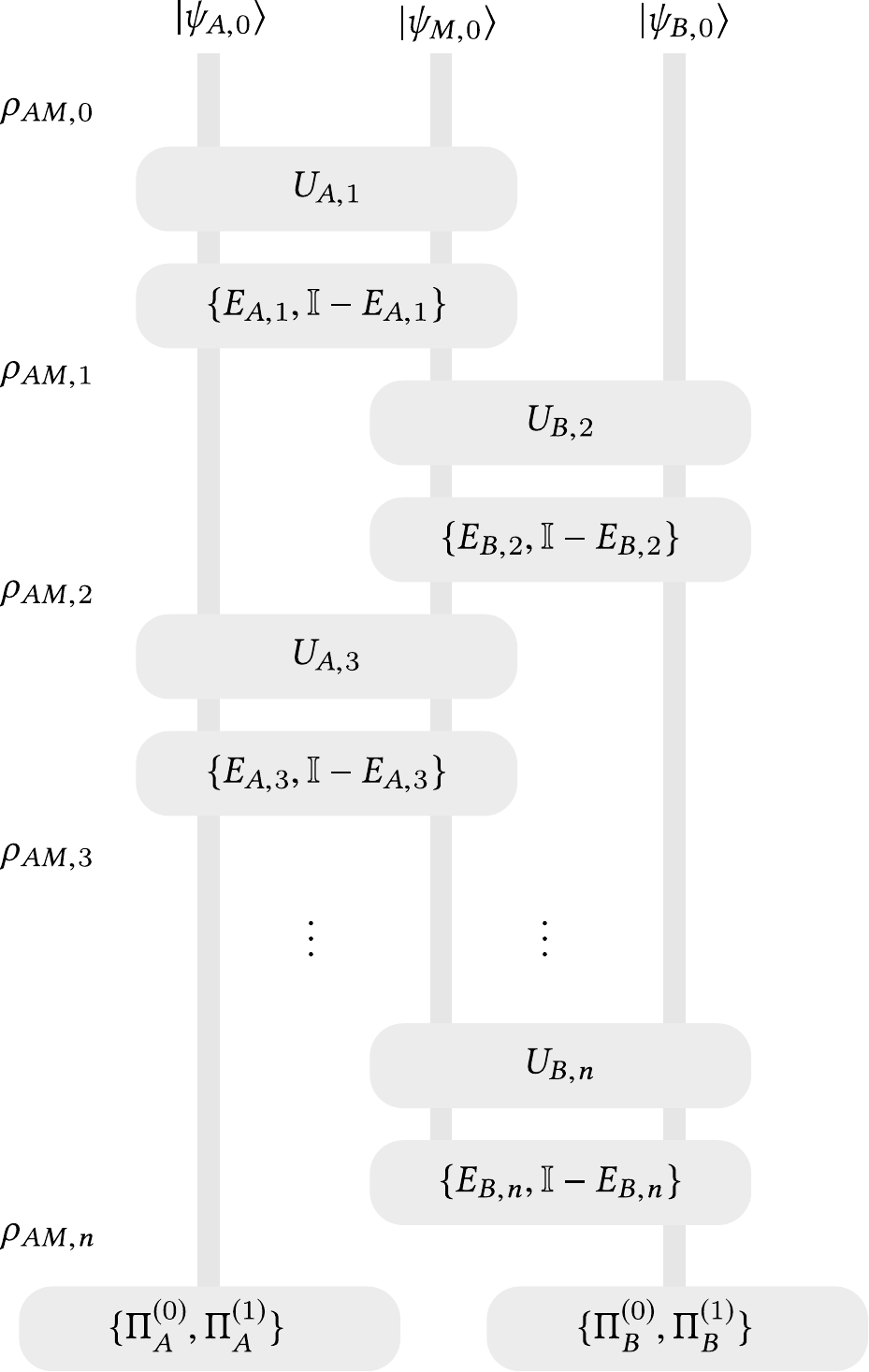}
    \par\end{centering}
  \caption{Every quantum WCF protocol can be cast into this general
    form. \label{fig:General-protocol-parametrised}}
\end{figure}
To define the bias of the protocol, we need to know
$P_{A}^{*}$ and $P_{B}^{*}$ corresponding to the best possible
cheating strategy of the opponent. This is formalised by the following (primal) semi-definite program:

\begin{theorem}[Primal] \label{thm:primal} Using the notation in \cref{def:WCFprotocol}, it holds that %
  ~\\
  $P_{B}^{*}=\max\text{Tr}((\Pi_{A}^{(1)}\otimes\mathbb{I}_{M})\rho_{AM,n})$
  over all $\rho_{AM,i}$ satisfying the constraints
  \begin{itemize}
    \item $\text{Tr}_{M}(\rho_{AM,0})=\text{Tr}_{MB}(\left|\psi_{0}\right\rangle \left\langle \psi_{0}\right|)=\left|\psi_{A,0}\right\rangle \left\langle \psi_{A,0}\right|$,
    \item for $i$ odd, $\text{Tr}_{M}(\rho_{AM,i})=\text{Tr}_{M}(E_{i}U_{i}\rho_{AM,i-1}U_{i}^{\dagger}E_{i})$, and
    \item for $i$ even, $\text{Tr}_{M}(\rho_{AM,i})=\text{Tr}_{M}(\rho_{AM,i-1}).$
  \end{itemize}
  $P_{A}^{*}=\max\text{Tr}((\mathbb{I}_{M}\otimes\Pi_{B}^{(0)})\rho_{MB,n})$
  over all $\rho_{BM,i}$ satisfying the constraints
  \begin{itemize}
    \item $\text{Tr}_{M}(\rho_{MB,0})=\text{Tr}_{AM}(\left|\psi_{0}\right\rangle \left\langle \psi_{0}\right|)=\left|\psi_{B,0}\right\rangle \left\langle \psi_{B,0}\right|$,
    \item for $i$ even, $\text{Tr}_{M}(\rho_{MB,i})=\text{Tr}_{M}(E_{i}U_{i}\rho_{MB,i-1}U_{i}^{\dagger}E_{i})$, and
    \item for $i$ odd, $\text{Tr}_{M}(\rho_{MB,i})=\text{Tr}_{M}(\rho_{MB,i-1})$.
  \end{itemize}
\end{theorem}

\begin{remark}
  In fact, one can restrict to unitaries without loss of generality
  (see page 9 of \cite{Mochon07}) by simulating the projections as
  coherent measurements and absorbing them into the final measurement.
  Generality is not
  lost because (a) the projections can only improve the bias
  and (b) a protocol with projections can be converted into one without projections. The use of projectors, though, simplifies the proofs, as we will see later.
  For instance,
  One could have, in addition to the measurement $\{E_{i},\mathbb{I}-E_{i}\}$,
  introduced a similar measurement, say $\{F_{i},\mathbb{I}-F_{i}\}$,
  before the unitary. This would yield $\tr_{M}(\rho_{AM,i})=\tr_{M}(E_{i}U_{i}F_{i}\rho_{AM,i-1}F_{i}U_{i}^{\dagger}E_{i})$
  for the SDP of $P_{B}^{*}$. %
  \label{rem:projBeforeAndAfter}
\end{remark}

Notice that $P_{B}^{*}$ depends on Alice's actions
specified in the protocol---as we optimise
over all possible actions of Bob---and thus involves variables such
as $\rho_{AM,i}$ and $U_{A,i}$. Analogously, $P_{A}^{*}$ depends
on Bob's actions.

A feasible solution to an optimisation problem is one that satisfies the constraints but is not necessarily optimal (viz. it does not necessarily achieve the highest/lowest value). Clearly, a feasible solution to the primal problems only yields a lower bound on $P_A^*$ and $P_B^*$. Using standard arguments, it is easily seen that feasible solutions to the dual problems (described below) yield \emph{upper} bounds on $P_A^*$ and $P_B^*$. In fact, in our case, it has been shown that \emph{strong duality} holds which means that the optimal values of the dual problems yield $P_A^*$ and $P_B^*$ exactly (and not just lower bounds). Physically, this entails that there exist cheating strategies corresponding to the optimal values of the dual problems.

\begin{theorem}[Dual] Using the notation in \cref{def:WCFprotocol}, it holds that
  \label{thm:dual}~\\
  $P_{B}^{*}=\min\text{Tr}(Z_{A,0}\left|\psi_{A,0}\right\rangle \left\langle \psi_{A,0}\right|)$
  over all $Z_{A,i}$ satisfying the constraints
  \begin{enumerate}
    \item $\forall i,$ $Z_{A,i}\ge0$,
    \item For $i$ odd, $Z_{A,i-1}\otimes\mathbb{I}_{M}\ge U_{A,i}^{\dagger}E_{A,i}(Z_{A,i}\otimes\mathbb{I}_{M})E_{A,i}U_{A,i}$,
    \item For $i$ even, $Z_{A,i-1}=Z_{A,i}$, and
    \item $Z_{A,n}=\Pi_{A}^{(1)}$.
  \end{enumerate}
  $P_{A}^{*}=\min\text{Tr}(Z_{B,0}\left|\psi_{B,0}\right\rangle \left\langle \psi_{B,0}\right|)$
  over all $Z_{B,i}$ satisfying the constraints
  \begin{enumerate}
    \item $\forall i,$ $Z_{B,i}\ge0$,
    \item For $i$ even, $\mathbb{I}_{M}\otimes Z_{B,i-1}\ge U_{B,i}^{\dagger}E_{B,i}(\mathbb{I}_{M}\otimes Z_{B,i})E_{B,i}U_{B,i}$,
    \item For $i$ odd, $Z_{B,i-1}=Z_{B,i}$, and
    \item $Z_{B,n}=\Pi_{B}^{(0)}$.
  \end{enumerate}
\end{theorem}

\begin{remark}
  As in \cref{rem:projBeforeAndAfter}, we note that the dual SDP corresponding to $P_{B}^{*}$
  would have yielded the constraint
  \begin{displaymath}
    Z_{A,i-1}\otimes\mathbb{I}_{M}\ge F_{A,i}U_{A,i}^{\dagger}E_{A,i}\left(Z_{A,i}\otimes\mathbb{I}_{M}\right)E_{A,i}U_{A,i}F_{A,i} \qquad \text{  for $i$ odd}.
  \end{displaymath}
  Similarly for $P_A^*$ and even $i$. \atul{Here, $F_{A,i}$ and $F_{B,i}$ are the same as $F_i$ in \cref{rem:projBeforeAndAfter} for odd and even $i$ respectively, with Alice and Bob explicitly labelled.}
  \label{rem:projBeforeAndAfterDual}
\end{remark}
Below, we formally define Time Dependent Point Games (TDPGs) which were briefly described earlier in \cref{sec:intro}. In fact, we define two variants---TDPGs with EBM functions and those with valid functions. %

\subsection{TDPGs with EBM transitions/functions\label{subsec:tdpgebm}}

Evidently, every protocol admits infinitely many representations as, in particular, there is freedom in the choice of basis. It is desirable to remove this redundancy to analyse the WCF problem.
Kitaev's solution was to define \emph{Time Dependent Point Games (TDPGs)}---a formulation equivalent to WCF protocols---that address exactly this issue. To define TDPGs, %
first consider, at a given step, the dual variables $Z_{A},Z_{B}$ as observables with
$\left|\psi\right\rangle $ governing the probability. This combines
the evolution of the certificates on cheating probabilities with the
evolution of the honest state\textemdash the state obtained when none of the parties is cheating.\footnote{Originally, using a similar
  maneuver, Kitaev settled the solvability of the quantum strong coin flipping
  problem by giving a lower bound on its bias \cite{Kitaev03}.}  This idea is formalised as follows.  %
\begin{definition}[Prob]
  Consider $Z\ge0$ and let $\Pi^{[z]}$ represent the projector on
  the eigenspace of eigenvalue $z\in\text{spectrum}(Z)$. We have $Z=\sum_{z}z\Pi^{[z]}$.
  Let $\left|\psi\right\rangle $ be a
  vector, not necessarily normalized. We define the function  $\text{Prob}[Z,\psi]:[0,\infty)\to[0,\infty)$
  as
  \begin{displaymath}
    \text{Prob}[Z,\psi](z)=\begin{cases}
      \left\langle \psi\right|\Pi^{[z]}\left|\psi\right\rangle & \text{if }z\in\text{sp}(Z) \\
      0                                                        & \text{else}.
    \end{cases}
  \end{displaymath}
  If $Z=Z_{A}\otimes\mathbb{I}_{M}\otimes Z_{B}$, using the
  same notation, we define the $2$-variate function
  $\text{Prob}[Z_{A},Z_{B},\psi]:[0,\infty)\times[0,\infty)\to[0,\infty)$, with finite support,
  as \begin{small}
    \begin{displaymath}
      \text{Prob}[Z_{A},Z_{B},\psi](z_{A},z_{B})=\begin{cases}
        \left\langle \psi\right|\Pi^{[z_{A}]}\otimes\mathbb{I}_{M}\otimes\Pi^{[z_{B}]}\left|\psi\right\rangle & \text{if }(z_{A},z_{B})\in\text{sp}(Z_{A})\times\text{sp}(Z_{B}), \\
        0                                                                                                     & \text{else}.
      \end{cases}
    \end{displaymath}
  \end{small}

  \label{def:prob}
\end{definition}

In this subsection, we consider TDPGs with EBM transitions. An \emph{Expressible by Matrices} (EBM) transition may be viewed as a distillation of each (non-trivial) step of a protocol. It is formalised as follows.

\begin{definition}[Line Transition]
  A line transition is an ordered pair of finitely supported functions $g,h:[0,\infty)\to[0,\infty)$, which
  we denote as $g\to h$.\label{def:transition}
\end{definition}

\begin{definition}[EBM line transition]
  Let $g,h:[0,\infty)\to[0,\infty)$ be two functions with finite
  supports. The line transition $g\to h$ is EBM if there exist two matrices $0\le G\le H$ and a vector $\left|\psi\right\rangle $, not necessarily
  normalized, such that $g=\prob\left[G,\left|\psi\right\rangle \right]$
  and $h=\prob\left[H,\left|\psi\right\rangle \right]$.\label{def:EBMlineTransition}
\end{definition}

\begin{definition}[EBM transition]
  Let $g,h:[0,\infty)\times[0,\infty)\to[0,\infty)$ be two functions
  with finite supports. The transition $g\to h$ is an
  \begin{itemize}
    \item EBM horizontal transition if $g(.,y)\to h(.,y)$
          is an EBM line transition for all $y\in[0,\infty)$, and
    \item EBM vertical transition if $g(x,.)\to h(x,.)$
          is an EBM line transition for all $x\in[0,\infty)$.
  \end{itemize}
\end{definition}

\begin{remark}
  When clear from the context, we refer to an EBM line transition simply
  as an EBM transition.
\end{remark}
We can now combine these two notions to define TDPGs with EBM transitions (also referred to as \emph{EBM point games}). We use the following 2-variate generalisation of Equation~\eqref{eq:MochonBracketPointGame}, in subsequent definitions: %

\begin{displaymath}
  \left\llbracket x_{g},y_{g}\right\rrbracket (x,y)=\begin{cases}
    1 & x_{g}=x\text{ and }y_{g}=y \\
    0 & \text{else.}
  \end{cases}
\end{displaymath}

\begin{definition}[TDPG with EBM transitions---EBM point game]
  \label{def:EBMpointGame} An EBM point game is a sequence of functions
  $\{g_{0},g_{1},\dots,g_{n}\}$ with finite support such that
  \begin{itemize}
    \item $g_{0}=1/2\left\llbracket 0,1\right\rrbracket +1/2\left\llbracket 1,0\right\rrbracket $;
    \item for all even $i$, $g_{i}\to g_{i+1}$ is an EBM vertical transition;
    \item for all odd $i$, $g_{i}\to g_{i+1}$ is an EBM horizontal transition;
    \item $g_{n}=1\left\llbracket \beta,\alpha\right\rrbracket $ for some $\alpha,\beta\in[0,1]$.
          We call $\left\llbracket \beta,\alpha\right\rrbracket $ the final
          point of the EBM point game.
  \end{itemize}
\end{definition}

In informal discussions, we often refer to transitions as \emph{moves} (of the corresponding point game).
As we alluded to, EBM point games may be viewed as a distillation of a WCF protocol and therefore the following should not come as a surprise.

\begin{proposition}[WCF $\implies$ EBM point game]
  \label{prop:WCFimpliesEBMPointGame}Given a WCF protocol with cheating
  probabilities $P_{A}^{*}$ and $P_{B}^{*}$, along with a positive
  real number $\delta>0$, there exists an EBM point game with final
  point $\left\llbracket P_{B}^{*}+\delta,P_{A}^{*}+\delta\right\rrbracket $.
\end{proposition}

The converse statement---given
an EBM point game the corresponding WCF protocol can be constructed---is not as easy to see, but it does indeed hold.
\begin{theorem}[EBM point game to protocol]
  \label{thm:EBMpointGameToWCFprotocol}Given an EBM point game with
  final point $\left\llbracket\beta,\alpha\right\rrbracket$, there exists a WCF protocol with $P_{A}^{*}\le\alpha$
  and $P_{B}^{*}\le\beta$.
\end{theorem}
These establish the equivalence between EBM point games and WCF protocols. We use it in \cref{sec:TEF}, to prove \cref{thm:TEFconstraint}. The proofs of all statements made here can be found in \cite{Mochon07,ACG+14}. %

\subsection{TDPGs with valid transitions/functions}
\label{subsec:tdpgvalid}
To check whether a given transition is EBM is not an easy task. Kitaev and Mochon \cite{Mochon07} introduced the following
alternate characterisation of EBM line
transitions to simplify the analysis.

\begin{proposition}(Relating EBM and strictly valid transitions ~\cite{Mochon07,ACG+14})
  \label{prop:ebmvalid}
  Let $g\to h$ where $g=\sum_{i=1}^{n_{g}}p_{g_{i}}\left\llbracket x_{g_{i}}\right\rrbracket $
  and $h=\sum_{i=1}^{n_{h}}p_{h_{i}}\left\llbracket x_{h_{i}}\right\rrbracket $
  with all $x_{g_{i}},x_{h_{i}}$ being non-negative and distinct ($x_{g_{i}}\neq x_{g_{j}}$
  and $x_{h_{i}}\neq x_{h_{j}}$ for every $i\neq j$), and $p_{g_{i}},p_{h_{i}}>0$. Then, the transition is EBM if it is \emph{strictly
    valid}, i.e. the following equality holds and the inequalities are
  \emph{strictly} satisfied:
  \begin{align*}
     & \sum_{i=1}^{n_{h}}p_{h_{i}}  =\sum_{i=1}^{n_{g}}p_{g_{i}}                                                                                                                                       \\
     & \sum_{i=1}^{n_{h}}p_{h_{i}}\frac{\lambda x_{h_{i}}}{\lambda+x_{h_{i}}}  \ge\sum_{i=1}^{n_{g}}p_{g_{i}}\frac{\lambda x_{g_{i}}}{\lambda+x_{g_{i}}} \quad\forall\lambda>0, \quad\text{ and }\quad
    \sum_{i=1}^{n_{h}}x_{h_{i}}p_{h_{i}}  \ge\sum_{i=1}^{n_{g}}x_{g_{i}}p_{g_{i}}.
  \end{align*}
  Conversely, a transition is \emph{valid}, i.e. satisfies these inequalities, if the transition $g\to h$ is EBM.
\end{proposition}

Using \cref{prop:ebmvalid}, one can consider a \emph{TDPG with valid transitions} (or briefly, a \emph{valid point game}), instead of looking at a TDPG with EBM transitions (or briefly, an EBM point game) as in \cref{def:EBMpointGame}. This is simply because a TDPG with valid transitions can be converted to a TDPG with strictly valid transitions, for any $\delta>0$ increase in the coordinates of the final point. Then, an application of \cref{prop:ebmvalid} immediately gives the corresponding TDPG with EBM transitions.

How do valid transitions help? Recall that EBM transitions involved ensuring certain matrix inequalities were satisfied. Valid transitions, instead, are characterised by scalar inequalities (albeit infinitely many, one for each $\lambda>0$) and this leads to significant simplification. For instance, one can check that the following transitions involving a single point are valid. These, as stated earlier, are already enough to construct TDPGs approaching bias $1/6$.

\begin{example}[Point raise]
  $p\left\llbracket x_{g}\right\rrbracket \to p\left\llbracket x_{h}\right\rrbracket $
  with $x_{h}\ge x_{g}$ is a valid transition.\label{exa:pointRaise}

\end{example}

\begin{example}[Point merge]
  \label{exa:merge} $p_{g_{1}}\left\llbracket x_{g_{1}}\right\rrbracket +p_{g_{2}}\left\llbracket x_{g_{2}}\right\rrbracket \to(p_{g_{1}}+p_{g_{2}})\left\llbracket x_{h}\right\rrbracket $
  with $x_{h}\ge\frac{p_{g_{1}}x_{g_{1}}+p_{g_{2}}x_{g_{2}}}{p_{g_{1}}+p_{g_{2}}}$
  is a valid transition, or generally $\sum_{i}p_{g_{i}}\left\llbracket x_{g_{i}}\right\rrbracket \to(\sum_{i}p_{g_{i}})\left\llbracket x_{h}\right\rrbracket $
  with $x_{h}\ge\left\langle x_{g}\right\rangle $ is a valid transition.
\end{example}

\begin{example}[Point split]
  $p_{g}\left\llbracket x_{g}\right\rrbracket \to p_{h_{1}}\left\llbracket x_{h_{1}}\right\rrbracket +p_{h_{2}}\left\llbracket x_{h_{2}}\right\rrbracket $
  with $p_{g}=p_{h_{1}}+p_{h_{2}}$ and $\frac{p_{g}}{x_{g}}\ge\frac{p_{h_{1}}}{x_{h_{1}}}+\frac{p_{h_{2}}}{x_{h_{2}}}$
  is a valid transition, or generally $\left(\sum_{i}p_{h_{i}}\right)\left\llbracket x_{g}\right\rrbracket \to\sum_{i}p_{h_{i}}\left\llbracket x_{h_{i}}\right\rrbracket $
  with $\frac{1}{x_{g}}\ge\left\langle \frac{1}{x_{h}}\right\rangle $
  is a valid transition.\label{exa:split}
\end{example}

We conclude this discussion by outlining the idea behind the proof of \cref{prop:ebmvalid}.\footnote{This result was first presented by Mochon and Kitaev, but it was proved using matrix perturbation theory \cite{Mochon07}. In \cite{ACG+14}, Aharonov, Chailloux,  Ganz,  Kerenidis and Magnin worked out a simpler proof, along the lines alluded to by  Mochon and  Kitaev, and this is the approach that we outline here.}
To this end, note that whenever $g$ and $h$ have disjoint support, one can equivalently consider the function $t=h-g$. Then, assuming the support is indeed disjoint, one can consider EBM (valid) functions instead of EBM (valid) transitions. The advantage of considering the set of functions (instead of transitions) is that such sets have better structure. In particular, the set of EBM functions is a convex cone, $K$. Interestingly, the dual of this cone, $K^*$, happens to be the set of \emph{operator monotone functions} (i.e. functions such that if $X\ge Y$, then $f(X) \ge f(Y)$ for all Hermitian matrices $X,Y$). %
This set, $K^*$, has been widely studied and shown to admit a surprisingly elegant and simple characterisation. Consequently, the bi-dual of EBM functions, i.e. $K^{**}$, also admits a simple characterisation---it is exactly the set of valid functions. A standard result in conic duality \cite{Boyd2004} states that $K^{**}=\rm{cl}(K)$ where $\rm{cl}$ denotes the closure. That is, the set of EBM functions and the set of valid functions are the same up to closures, which almost completes the proof. Crucially, this is exactly the step which is non-constructive in Mochon's analysis---given a valid function, there is no known general procedure for constructing the matrices which certify the function is EBM. To complete the proof, the subtlety about closures must be handled. In \cite{ACG+14} the authors handle it by considering strictly valid functions instead of valid functions. In our approach introduced in \cref{sec:TEF}, we show that the closure issue is naturally accounted for, by explicitly considering projectors (as in \cref{thm:primal}).

\subsection{Time-Independent Point Games (TIPGs)}
\label{subsec:tipg}

The point game formalism can be further simplified, and it is in this simplified formalism that Mochon constructed his family of point games achieving arbitrarily small bias. Instead of considering the entire \emph{sequence} of horizontal and vertical transitions,
he focused on just two functions (hence the name \emph{time-independent}), as described below.
\begin{definition}[TIPG]
  \label{def:TIPG}A \emph{time-independent point game (TIPG)} is
  a valid horizontal function, denoted by $a$, and a valid vertical function, denoted by $b$,
  such that
  \begin{displaymath}
    a+b=1\left\llbracket \beta,\alpha\right\rrbracket -\frac{1}{2}\left\llbracket 0,1\right\rrbracket -\frac{1}{2}\left\llbracket 1,0\right\rrbracket
  \end{displaymath}
  for some $\alpha,\beta>1/2$. Further
  \begin{itemize}
    \item we call the point $\left\llbracket \beta,\alpha\right\rrbracket $
          the final point of the game, and
    \item we call the set $\mathcal{S}=\left(\supp(a)\cup\supp(b)\right) \backslash\supp(a+b)$,
          the set of intermediate points.
  \end{itemize}
\end{definition}

\begin{remark}
  When clear from the context, we may use the word TIPG even when $a+b$
  is not necessarily $\left\llbracket \beta,\alpha\right\rrbracket -\frac{1}{2}\left(\left\llbracket 0,1\right\rrbracket +\left\llbracket 1,0\right\rrbracket \right)$
  but some other function, $c$, with finite support in
  $[0,\infty)\times [0,\infty)$
  satisfying $\sum_{x\in\supp(c)}c(x)=0$.
\end{remark}

It is straightforward to show that every valid point game (as defined above) corresponds to a TIPG with the same final point $(\beta,\alpha)$. %
Explicitly, if the valid point game with final point $\left\llbracket \beta,\alpha\right\rrbracket $
is specified by $a_{1},a_{2}\dots a_{n}$ valid horizontal and $b_{1},b_{2}\dots b_{n}$ valid vertical functions,
then the corresponding TIPG is specified by $a=\sum_{i=1}^{n}a_{i}$
and $b=\sum_{i=1}^{n}b_{i}$, which are horizontally and vertically
valid, respectively, and satisfy $a+b=\left\llbracket \beta,\alpha\right\rrbracket -\frac{1}{2}\left\llbracket 0,1\right\rrbracket -\frac{1}{2}\left\llbracket 1,0\right\rrbracket $. Surprisingly, the converse was also shown to hold.

\begin{theorem}[TIPG to valid point games~\cite{Mochon07,ACG+14}]
  Given a TIPG with a valid horizontal function $a$ and a valid vertical
  function $b$ such that $a+b=1\left\llbracket\beta,\alpha\right\rrbracket-\frac{1}{2}\left\llbracket0,1\right\rrbracket-\frac{1}{2}\left\llbracket1,0\right\rrbracket$,
  one can construct, for all $\epsilon>0$, a valid point game with its final
  point being $\left\llbracket\beta+\epsilon,\alpha+\epsilon\right\rrbracket$, where the number of transitions
  depends on $\epsilon$ \atul{
    and each transition is either}
  \begin{itemize}
    \item \atul{
            a point raise, a point merge, a point split, or}
    \item \atul{
            the horizontally valid function $a$, scaled down, or}
    \item \atul{
            the vertically valid function $b$, scaled down,}
  \end{itemize} %
  \atul{
    where the scaling factor also depends on $\epsilon$.}
  \label{thm:TIPG-to-valid-point-games}
\end{theorem}

In words, the theorem says that every TIPG can be converted to a valid TDPG with almost the same final point\atul{, using essentially the same non-trivial moves encoded in $a$ and $b$}. However, this seems counter-intuitive because it is not a priori clear how a time ordered sequence of transitions can be extracted from a time-independent point game. For instance, one might run into causal loops---we expect a point to be present to create another point which in turn is required to produce the first point. To overcome such issues, the key idea is to use a so-called \emph{catalyst state}: (i) Deposit a small amount of weight wherever $a$ assigns negative weight. (ii) Run a scaled down round of $a$ and $b$ (the scaling is proportional to the weight deposited in the beginning). (iii) Repeat (ii) until almost all the weight has been transferred to the final point. (iv) Absorb the catalyst state \atul{into the final point,} at a small cost to the bias.

Among these, performing step (iv), needs most care. The weight in step (i) determines the number of times step (ii) must be repeated. That, in turn, determines the number of rounds the protocol requires. While in this work, we do not focus on the resources required to implement WCF, we nonetheless state the following which, in particular, relates the bias to the round complexity (number of rounds of communication) of point games. The latter, (using our results in \cref{sec:TEF}) can be used to obtain protocols with (essentially) the same bias and round complexity.\footnote{However, this particular result is not a new contribution.}

\begin{corollary}[\cite{ACG+14}]
  Consider a TIPG with a valid horizontal function $a=a^{+}-a^{-}$
  and a valid vertical function $b=b^{+}-b^{-}$ such that $a+b=\left\llbracket \beta,\alpha\right\rrbracket -\frac{1}{2}\left\llbracket 0,1\right\rrbracket -\frac{1}{2}\left\llbracket 1,0\right\rrbracket $ where $a^+,a^-,b^+,b^-$ are finitely supported functions that take values in $[0,\infty)$ with disjoint support (i.e. $\supp(a^+)\cap \supp(a^-)=\emptyset$ and similarly for $b^+$ and $b^-$).
  Let $\Gamma$ be the largest coordinate of all the points that appear
  in the TIPG. Then, for all $\epsilon>0$, one can construct a point
  game with $\mathcal{O}\left(\frac{\left\Vert b\right\Vert \Gamma^{2}}{\epsilon^{2}}\right)$
  valid transitions and final point $\left\llbracket \beta+\epsilon,\alpha+\epsilon\right\rrbracket $.\label{cor:numberOfTransitionsForValidPointGames}
\end{corollary}

\subsection{Mochon's TIPG achieving bias \texorpdfstring{$\epsilon(k)=1/(4k+2)$}{epsilon(k)=1/(4k+2)}}\label{subsec:mochontipg}
We can now explain how Mochon~\cite{Mochon07} proved the existence of WCF protocols with arbitrarily small bias. He constructed a family of TIPGs, parametrised by an integer $k>0$, such that the final point is $\left\llbracket \frac{1}{2}+\epsilon(k),\frac{1}{2}+\epsilon(k)\right\rrbracket $, where $\epsilon(k)=1/(4k+2)$ (see  \cref{fig:IllusMochonGamek=00003D2}). %

\begin{figure}[ht]
  \begin{centering}
    \subfloat[Mochon's TIPG for~$k=2$.\label{fig:IllusMochonGamek=00003D2}]{\begin{centering}
        \hspace{2.2cm}\includegraphics[scale=1.1]{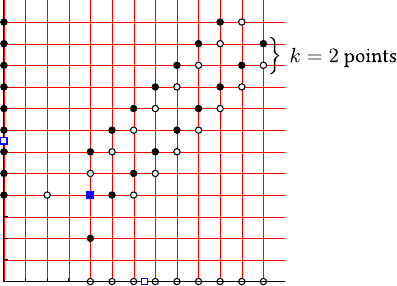}
        \par\end{centering}
    }
    \par\end{centering}
  \begin{centering}
    \subfloat[Mochon's TIPG in three stages, the initial \emph{splits},
      the \emph{ladder} and the \emph{raises}.\label{fig:MochonGameStages}]{\begin{centering}
        \includegraphics[scale=1.1]{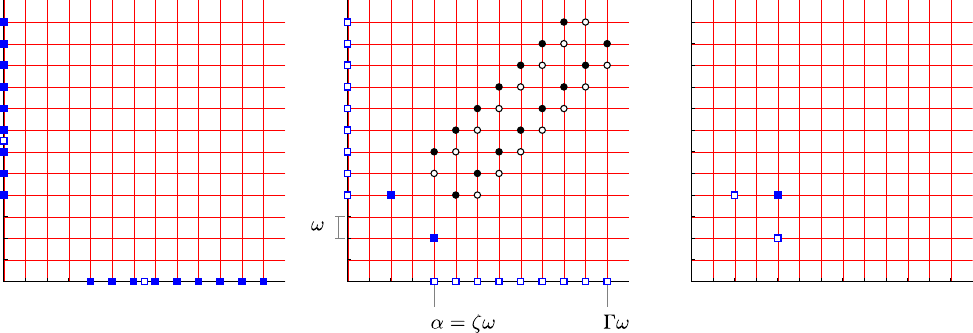}
        \par\end{centering}
    }
    \par\end{centering}
  \caption{Mochon's TIPG.  The unfilled squares represent initial points of a TIPG (i.e. points with negative weight in $a+b$) and the filled squares point represent final points (i.e. points with positive weight in $a+b$). %
    \atul{Filled circles carry negative weight and unfilled circles carry positive weight for the horizontally valid function. For the vertically valid function, it is the other way around. Thus, when the functions are added, the points corresponding to these circles cancel---except \emph{for circles on the axes}: circles along the $y$-axis represent points present only in the horizontally valid function and circles along the $x$-axis represent those only in the vertically valid function. Note that in both cases, they have negative weight.}}
\end{figure}

The overall structure of these games is easy to describe. Apart from their initial points, $\left\llbracket 0,1\right\rrbracket $
and $\left\llbracket 1,0\right\rrbracket $, all the other points involved are placed on a regular lattice, i.e. at locations of the form $\left\llbracket a\omega,b\omega\right\rrbracket $
where $a,b\in\mathbb{N}$ and $\omega\in (0,\infty)$. The final point of the games is $\left\llbracket \alpha,\alpha\right\rrbracket $
for $\alpha=\zeta\omega=\frac{1}{2}+\mathcal{O}\left(\frac{1}{k}\right)$
where $\zeta\in\mathbb{N}$, and in general, they have the following three stages (see
 \cref{fig:MochonGameStages}):
\begin{enumerate}
  \item \emph{Split}. The point $\left\llbracket 0,1\right\rrbracket $ is
        vertically split into many points along the $y$-axis. The resulting
        points lie between $\zeta\omega$ and $\Gamma\omega$ with $\zeta,\Gamma\in \mathbb{N}$. Analogously, the point $\left\llbracket 1,0\right\rrbracket $
        is horizontally split into many points along the $x$-axis.
  \item \emph{Ladder}. This is the main non-trivial move of the games parametrised by an integer $k>0$, and it consists of points along the diagonal
        and along the axes (see the second image in  \cref{fig:MochonGameStages}). The points on the axes are transformed by the
        ladder into the final points $\left\llbracket \alpha-k\omega,\alpha\right\rrbracket $
        and $\left\llbracket \alpha,\alpha-k\omega\right\rrbracket $.
  \item \emph{Raise}. The two points $\left\llbracket \alpha-k\omega,\alpha\right\rrbracket $
        and $\left\llbracket \alpha,\alpha-k\omega\right\rrbracket $ are
        raised to the final point $\left\llbracket \alpha,\alpha\right\rrbracket $.
\end{enumerate}
For each integer $k>0$ there exist parameters
$\omega,\Gamma\in (0,\infty)$ such that
the two initial splits are valid,  the \emph{ladder}
corresponds to a horizontally and vertically valid function,
and $\alpha=\frac{1}{2}+\mathcal{O}\left(\frac{1}{k}\right)$. %

The key technical tool that Mochon introduced is the following: given a set
of point coordinates, he constructed a way of assigning non-trivial  weights
to them such that this assignment is valid while still retaining considerable
freedom. This weight assignment is parametrised
by a polynomial and works for essentially all polynomials up to a
certain degree. In other words, he simplified the validity condition
by restricting to a class of functions which are easy to manipulate
and are valid by construction.
\begin{lemma}[Mochon's assignment is valid\cite{Mochon07,ACG+14}]
  Let \label{lem:fAssignment}
  \begin{itemize}
    \item $x_{1},x_{2}\dots x_{n}$ be distinct, non-negative real numbers, and
    \item $f$ be a polynomial of degree at most $n-1$ satisfying $f(-\lambda)\ge0$
          for all $\lambda\ge0$.
  \end{itemize}
  Then,
  \begin{equation}
    a=\sum_{i=1}^{n}\frac{-f(x)}{\prod_{j\neq i}(x_{j}-x_{i})}\left\llbracket x_{i}\right\rrbracket 	\label{eq:fAssignmentInitial}
  \end{equation}
  is a valid function.
\end{lemma}

These functions, which are later referred to as $f$-assignments, play a crucial role in our systematic construction of WCF protocols corresponding to the TIPGs described above (see \cref{sec:1by4k+2}).

\section{TDPG-to-Explicit-protocol Framework (TEF) and Bias 1/10 Game and  Protocol\label{sec:TEF}}

In this section, we give a framework for converting a TDPG (with EBM or valid transitions) into an explicit protocol, approaching the same bias. In fact, we introduce a slightly different condition which is similar to the EBM condition but involves projectors. These conditions (valid, EBM and the one we introduce) are equivalent but we defer this discussion to the appendix. %
This is because, in the present and subsequent section, we explicitly construct the matrices to show the required conditions are satisfied for TDPGs of interest. In particular, we begin by constructing the appropriate matrices corresponding to the three \emph{basic} moves involving a single point---raise, split and merge (\cref{exa:pointRaise}, \cref{exa:split} and \cref{exa:merge} resp.). These already recover the bias 1/6 protocol from the bias 1/6 TDPG. To go below, we construct matrices for \emph{advanced} moves that take three points to two points (and also two points to two points), corresponding to Mochon's TDPG approaching bias $1/10$. Together with the three \emph{basic} moves, these allow us to construct protocols approaching bias $1/10$. The construction of \emph{advanced} moves is perturbative. Thus, going below $1/10$ requires more work and that is covered in the next section. %

\emph{Remark about prior work.} To establish the equivalence between TDPG and WCF protocols, prior works \cite{ACG+14} and \cite{Mochon07} also showed a way to convert a TDPG into a WCF protocol. However, one of the primary differences compared to our work is that, as we \atul{alluded to in the introduction,} %
the message register in our case decouples after each round as we suitably place projectors (which correspond to cheat detection). This leads to simplifications---both mathematical and practical.

\subsection{The framework}
\label{subsec:framework}

We want to construct a WCF protocol such that its dual (see \cref{thm:dual}) corresponds to a given TDPG. We therefore start with a frame of a TDPG, and \atul{
  sequentially build dual variables and rotations. These rotations specify a WCF protocol and the dual variables are such that they constitute a feasible solution to the dual of this WCF protocol. The key property of this dual feasible solution is that it certifies that the WCF protocol we obtain has the same bias as that encoded in the final frame of the TDPG we started with.}

Recall that TDPGs are formulated in terms of \emph{Prob} (see \cref{def:prob}). The most natural way to construct the matrices $Z$s and the vector $\ket{\psi}$ (which appear in the definition of \emph{Prob}) is the following: Given an arbitrary frame of a TDPG, construct an entangled state that encodes the weight and define $Z$s to contain the coordinates corresponding to these weights. %
We formalise these as the \emph{Canonical Form}.

\begin{definition}[Canonical Form]
  The tuple $(\left|\psi\right\rangle ,Z^{A},Z^{B})$ is said to be
  in the Canonical Form with respect to a set of points in a frame \atul{$\sum_i P_i \llbracket x_i,y_i \rrbracket$} of
  a TDPG\footnote{One could define the canonical form for any frame but we only use it for those arising from TDPGs.} if  $\left|\psi\right\rangle =\sum_{i}\sqrt{P_{i}}\left|ii\right\rangle _{AB}\otimes\left|\varphi\right\rangle _{M}$,
  $Z^{A}=\sum x_{i}\left|i\right\rangle \left\langle i\right|_{A}$
  and $Z^{B}=\sum y_{i}\left|i\right\rangle \left\langle i\right|_{B}$
  where $\left|\varphi\right\rangle _{M}$ represents the state of extra
  uncoupled registers which might be present.\label{def:canonicalform}
\end{definition}

The label $\left|ii\right\rangle$
corresponds to a point with coordinates $x_{i},y_{i}$ and weight $P_{i}$
in the frame (see also  \cref{fig:TDPGframe}). It is tempting to imagine that
we systematically construct, from each frame of a TDPG, a canonical
form of $\left|\psi\right\rangle s$ and $Z$s, and deduce the unitaries from the evolution of the state $\ket \psi$. This approach suffers from two issues: (a) the unitaries are not necessarily decomposable into moves by Alice and Bob who communicate only through the message register, and, (b) the constraints imposed on consecutive $Z$s (by, say, a TDPG with EBM transitions), that take the form $Z_{n-1}\otimes\mathbb{I}\ge U_{n}^{\dagger}\left(Z_{n}\otimes\mathbb{I}\right)U_{n}$, are not satisfied in general.

We design our framework to overcome these issues. Before we delve into the details, we clarify how the output of the framework relates to a WCF protocol. The framework outputs variables indexed as $\left|\psi_{(i)}\right\rangle $, $Z_{(i)}$,
$U_{(i)}$ (see \cref{def:EBMpointGame} and \cref{prop:WCFimpliesEBMPointGame}) and they are produced in the reverse time convention (relative to the WCF protocol). This means that the variables at the $i$th step of the protocol (which follows the forward time convention) are given by $\ket{\psi_{i}}=\ket{\psi_{(N-i)}}, Z_i=Z_{(N-i)}$ and $U_{i}=U^\dagger_{(N-i)}$. In fact, this extends naturally to the case where one additionally has projectors, e.g. $U_{i}E_{i}=E_{(N-i)}U_{(N-i)}^{\dagger}$.

\begin{figure}[ht]
  \begin{centering}
    \subfloat[Frame of a TDPG\label{fig:TDPGframe}]{\begin{centering}
        \includegraphics[scale=0.7]{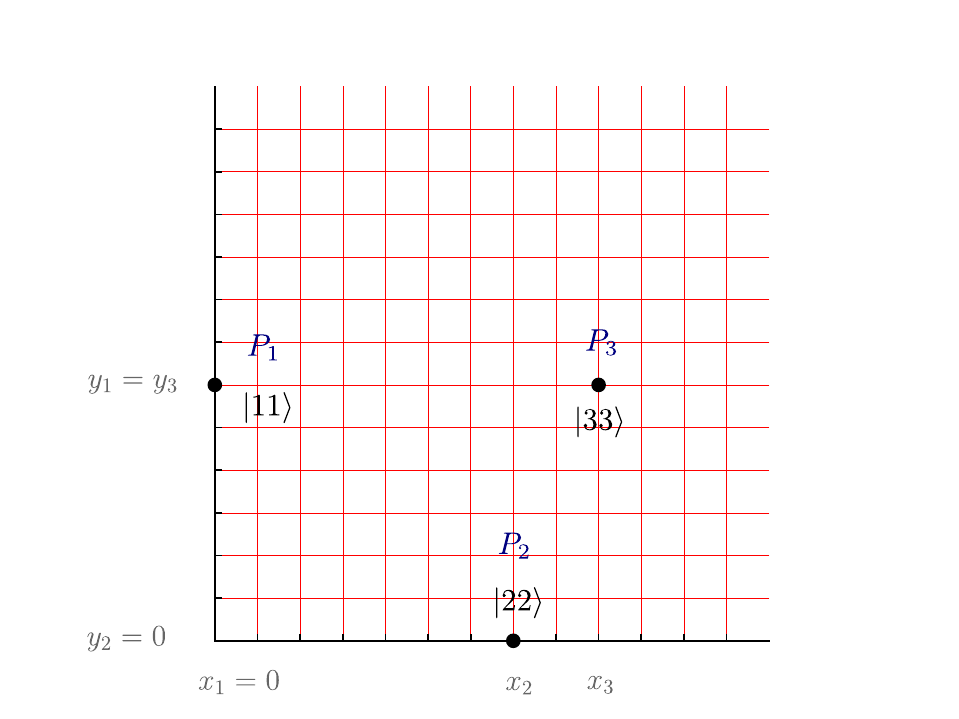}
        \par\end{centering}
    }
    \par\end{centering}
  \begin{centering}
    \subfloat[The points that are unchanged from one frame to another are labelled
      by $\{k_{i}\}$.\newline Among the points that change, the initial ones are
      labelled by $\{g_{i}\}$ and the final ones by $\{h_{i}\}$.\label{fig:TDPGillustrating_kgh}]{\begin{centering}
        \includegraphics[scale=0.7]{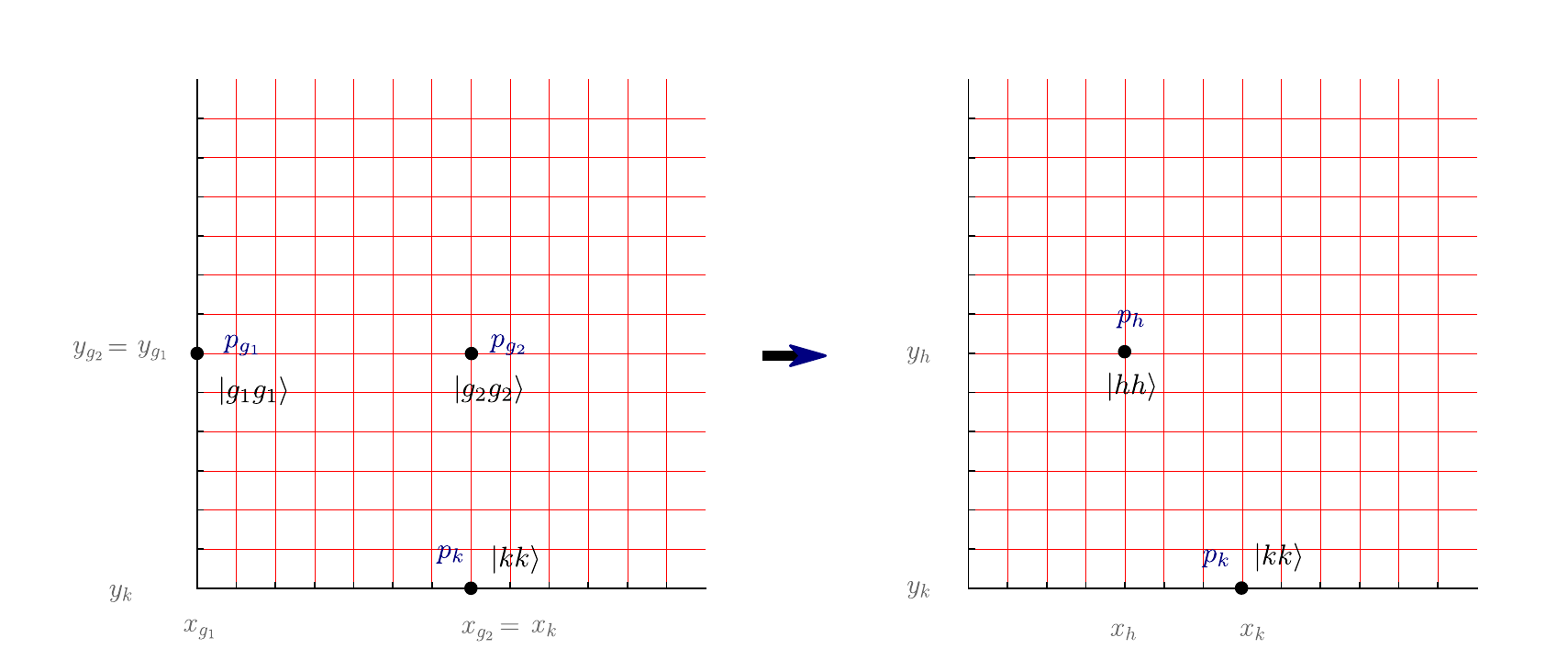}
        \par\end{centering}
    }
    \par\end{centering}
  \caption{Illustrations for the Canonical Form}
\end{figure}

Let us start with an informal outline of our framework.
Assume that a canonical description is given. Let the
labels on the points we want to transform be $\{g_{i}\}$,
and let us also assume that we wish to apply a horizontal transition, i.e. Alice performs the non-trivial step. Let the labels of the
points that will be left unchanged be $\{k_{i}\}$
(see  \cref{fig:TDPGillustrating_kgh}). We can write the state as
\begin{displaymath}
  \left|\psi_{(1)}\right\rangle =\left(\sum_{i}\sqrt{p_{g_{i}}}\left|g_{i}g_{i}\right\rangle _{AB}+\sum_{i}\sqrt{p_{k_{i}}}\left|k_{i}k_{i}\right\rangle _{AB}\right)\otimes\left|m\right\rangle _{M}.
\end{displaymath}
We\footnote{To be explicit, for $\mathcal{X} \in \{\mathcal{A},\mathcal{M},\mathcal{B}\}$, the Hilbert space $\mathcal{X}$ is the span of the orthonormal vectors $\{\{\ket{g_i}_X\}_i,\{\ket{k_i}_X\}_i,\{\ket{h_i}_X\}_i,\ket{m}\}$} want Bob to send his part of $\left|g_{i}\right\rangle $ states
to Alice through the message register. One way is to conditionally
swap to obtain
\begin{displaymath}
  \left|\psi_{(2)}\right\rangle =\sum_{i}\sqrt{p_{g_{i}}}\left|g_{i}g_{i}\right\rangle _{AM}\otimes\left|m\right\rangle _{B}+\sum_{i}\sqrt{p_{k_{i}}}\left|k_{i}k_{i}\right\rangle _{AB}\otimes\left|m\right\rangle _{M}.
\end{displaymath}
This way, all the points align along the $y$-axis, while the respective $x$-coordinates remain the same due to the fact that it is a horizontal transition.\footnote{\atul{Intuitively, this is because we relax the dual variable $Z^B$ and make all the points have the same $y_{\mathrm max}$ coordinate (where $y_{\mathrm max}$ is the highest $y$ coordinate among the points of interest), as explained in Step \emph{2. Bob sends to Alice} below.}}
Let $\{h_{i}\}$ be the labels of the new points after the transformation.
We assume that $h_{i}$, $g_{i}$ and $k_{i}$ index orthonormal vectors.
Alice can update the probabilities and labels by locally performing
a unitary to obtain
\begin{displaymath}
  \left|\psi_{(3)}\right\rangle =\sum_{i}\sqrt{p_{h_{i}}}\left|h_{i}h_{i}\right\rangle _{AM}\otimes\left|m\right\rangle _{B}+\sum_{i}\sqrt{p_{k_{i}}}\left|k_{i}k_{i}\right\rangle _{AB}\otimes\left|m\right\rangle _{M}.
\end{displaymath}
It is precisely this step that yields the non-trivial constraint.
Bob must now accept this by `unswapping' to get
\begin{displaymath}
  \left|\psi_{(4)}\right\rangle =\left(\sum_{i}\sqrt{p_{h_{i}}}\left|h_{i}h_{i}\right\rangle _{AB}+\sum_{i}\sqrt{p_{k_{i}}}\left|k_{i}k_{i}\right\rangle _{AB}\right)\otimes\left|m\right\rangle _{M}.
\end{displaymath}
As we mentioned, relative to the actual protocol, the sequence is in the reverse time convention. Note also that we add a few extra frames to the final
TDPG to go from a given frame to the next of the original TDPG. This
is irrelevant, when resource usage is not of interest, as the bias
does not change.

We now fill in the  details and show that at each step, one can ensure certain matrix inequalities hold. (For the non-trivial step, a matrix inequality is assumed to hold, instead.) These inequalities, \atul{are exactly the inequalities in the constraints of the dual SDP (see \cref{thm:dual})}, and, in turn, ensure that one directly obtains a dual of the WCF protocol corresponding to the TDPG of interest.

\begin{enumerate}
  \item \textbf{First frame.}
        \begin{align*}
          \left|\psi_{(1)}\right\rangle & =\left(\sum_{i}\sqrt{p_{g_{i}}}\left|g_{i}g_{i}\right\rangle _{AB}+\sum_{i}\sqrt{p_{k_{i}}}\left|k_{i}k_{i}\right\rangle _{AB}\right)\otimes\left|m\right\rangle _{M} \\
          Z_{(1)}^{A}                   & =\sum_{i}x_{g_{i}}\left|g_{i}\right\rangle \left\langle g_{i}\right|_{A}+\sum_{i}x_{k_{i}}\left|k_{i}\right\rangle \left\langle k_{i}\right|_{A}                      \\
          Z_{(1)}^{B}                   & =\sum_{i}y_{g_{i}}\left|g_{i}\right\rangle \left\langle g_{i}\right|_{B}+\sum_{i}y_{k_{i}}\left|k_{i}\right\rangle \left\langle k_{i}\right|_{B}.
        \end{align*}
        \begin{proof}
          Follows from the assumption of starting with a Canonical Form.
        \end{proof}
  \item \textbf{Bob sends to Alice.} With $y\ge\text{max}\{y_{g_{i}}\}$ the
        following
        \begin{align*}
          \left|\psi_{(2)}\right\rangle & =\sum_{i}\sqrt{p_{g_{i}}}\left|g_{i}g_{i}\right\rangle _{AM}\otimes\left|m\right\rangle _{B}+\sum_{i}\sqrt{p_{k_{i}}}\left|k_{i}k_{i}\right\rangle _{AB}\otimes\left|m\right\rangle _{M} \\
          U_{(1)}                       & =U_{BM}^{\text{SWP}\{\vec{g},m\}}                                                                                                                                                        \\
          Z_{(2)}^{A}                   & =Z_{(1)}^{A}\ \ \ \text{ and }\ \ \
          Z_{(2)}^{B}  =y\mathbb{I}_{B}^{\{\vec{g},m\}}+\sum_{i}y_{k_{i}}\left|k_{i}\right\rangle \left\langle k_{i}\right|_{B},
        \end{align*}
        is a viable choice, i.e. it satisfies the properties  $$(1)\quad \left|\psi_{(2)}\right\rangle =U_{(1)}\left|\psi_{(1)}\right\rangle, $$
        and $$(2)\quad U^\dagger_{(1)}\left(Z_{(2)}^{B}\otimes\mathbb{I}_{M}\right)U_{(1)}\ge\left(Z_{(1)}^{B}\otimes\mathbb{I}_{M}\right).$$
        \begin{proof}
          We have to prove that the above properties (1) and (2) are satisfied.
          (1) It follows trivially from the defining action of $U_{(1)}$.\\
          (2) For ease of notation, let $U=U_{(1)}$ and note that $U^{\dagger}=U$,
          so that we can write 
          \begin{align*}
            &U\left(Z_{(2)}^{B}\otimes\mathbb{I}_{M}\right)U\\
            & =y\left(U\left(\mathbb{I}_{B}^{\{\vec{g},m\}}\otimes\mathbb{I}_{M}^{\{\vec{g},m\}}\right)U+U\underbrace{\left(\mathbb{I}_{B}^{\{\vec{g},m\}}\otimes\mathbb{I}_{M}^{\{\vec{k},\vec{h}\}}\right)}_{\text{outside }U\text{'s action space}}U\right)+U\underbrace{\left(\sum y_{k_{i}}\left|k_{i}\right\rangle \left\langle k_{i}\right|\otimes\mathbb{I}\right)}_{\text{outside }U\text{'s action space}}U \\ &=Z_{(2)}\otimes\mathbb{I}_{M}\ge Z_{(1)}\otimes\mathbb{I}_{M}
          \end{align*}
          so long\footnote{By the action space of $U$ we mean the space where $U$ acts non-trivially.} as $y\ge y_{g_{i}}$, which is guaranteed by the choice of
          $y$.
        \end{proof}
  \item \textbf{Alice's non-trivial step. }Consider the following choice
        \begin{align*}
          \left|\psi_{(3)}\right\rangle & =\sum_{i}\sqrt{p_{h_{i}}}\left|h_{i}h_{i}\right\rangle _{AM}\otimes\left|m\right\rangle _{B}+\sum_{i}\sqrt{p_{k_{i}}}\left|k_{i}k_{i}\right\rangle _{AB}\otimes\left|m\right\rangle _{M} \\
          E_{(2)}U_{(2)}                & =E_{(2)}\left(\left|w\right\rangle \left\langle v\right|+\text{other terms acting on span\{}\left|h_{i}h_{i}\right\rangle ,\left|g_{i}g_{i}\right\rangle \}\right)_{AM}                  \\
          Z_{(3)}^{A}                   & =\sum_{i}x_{h_{i}}\left|h_{i}\right\rangle \left\langle h_{i}\right|+\sum_{i}x_{k_{i}}\left|k_{i}\right\rangle \left\langle k_{i}\right|\ \ \ \text{ and }\ \ \
          Z_{(3)}^{B}  =Z_{(2)}^{B}
        \end{align*}
        where \begin{footnotesize}
          \begin{displaymath}\left|v\right\rangle =\frac{\sum_{i}\sqrt{p_{g_{i}}}\left|g_{i}g_{i}\right\rangle }{\sqrt{\sum_{i}p_{g_{i}}}},\,\left|w\right\rangle =\frac{\sum_{i}\sqrt{p_{h_{i}}}\left|h_{i}h_{i}\right\rangle }{\sqrt{\sum_{i}p_{h_{i}}}},
            E_{(2)}=\left(\sum\left|h_{i}\right\rangle \left\langle h_{i}\right|_{A}+\sum\left|k_{i}\right\rangle \left\langle k_{i}\right|_{A}\right)\otimes\mathbb{I}_{M}
          \end{displaymath}

        \end{footnotesize}

        subject to the condition
        \begin{equation}
          \sum x_{h_{i}}\left|h_{i}h_{i}\right\rangle \left\langle h_{i}h_{i}\right|\ge\sum x_{g_{i}}E_{(2)}U_{(2)}\left|g_{i}g_{i}\right\rangle \left\langle g_{i}g_{i}\right|U_{(2)}^\dagger E_{(2)}\label{eq:MainConstraintInequality}
        \end{equation}
        and the conservation of probability, viz. $\sum p_{g_{i}}=\sum p_{h_{i}}$.
        We claim that this choice is viable, i.e. it satisfies the conditions
        (1) $E_{(2)}\left|\psi_{(3)}\right\rangle =U_{(2)}\left|\psi_{(2)}\right\rangle$, and (2) $Z_{(3)}^{A}\otimes\mathbb{I}_{M}\ge E_{(2)}U_{(2)}\left(Z_{(2)}^{A}\otimes\mathbb{I}_{M}\right)U^\dagger_{(2)}E_{(2)}.$

        \begin{proof}
          We must show that (1) and (2) as above hold. For (1)  we observe that $E_{(2)}\left|\psi_{(3)}\right\rangle =\left|\psi_{(3)}\right\rangle $ and
          the statement holds by construction of $U_{(2)}$.\\
          (2) Consider the space $\mathcal{H}=\text{span}\left\{ \left|g_{1}g_{1}\right\rangle ,\left|g_{2}g_{2}\right\rangle \dots,\left|h_{1}h_{1}\right\rangle ,\left|h_{2},h_{2}\right\rangle \dots\right\} $ which is a subspace of $\mathcal{A}\otimes\mathcal{M}$ (space of Alice and the message register). One can write $\mathcal{A}\otimes \mathcal{M}=\mathcal{H} \oplus \mathcal{H}^{\perp}$. %
          We separate all expressions which act on the $\mathcal{H}$
          space from the rest. We start with the RHS,
          excluding the $U_{(2)}$'s,
          \begin{small}
            \begin{displaymath}
              Z_{(2)}^{A}\otimes\mathbb{I}_{M}=\underbrace{\sum x_{g_{i}}\left|g_{i}g_{i}\right\rangle \left\langle g_{i}g_{i}\right|}_{\text{I}}+\sum x_{g_{i}}\left|g_{i}\right\rangle \left\langle g_{i}\right|\otimes\left(\mathbb{I}-\left|g_{i}\right\rangle \left\langle g_{i}\right|\right)+\sum x_{k_{i}}\left|k_{i}\right\rangle \left\langle k_{i}\right|\otimes\mathbb{I}.
            \end{displaymath}
          \end{small}
          Note that $Z_{(2)}^{A}\otimes\mathbb{I}_{M}$ is block diagonal with respect to $\mathcal{H} \oplus \mathcal{H}^{\perp}$, with term I making the first block (corresponding to $\mathcal{H}$), and the rest constituting the second block.
          Next consider the LHS,
          \begin{footnotesize}
            \begin{align*}
              Z_{(3)}^{A}\otimes\mathbb{I}_{M} & =\underbrace{\sum x_{h_{i}}\left|h_{i}h_{i}\right\rangle \left\langle h_{i}h_{i}\right|}_{\text{I}}+\sum x_{h_{i}}\left|h_{i}\right\rangle \left\langle h_{i}\right|\otimes\left(\mathbb{I}-\left|h_{i}\right\rangle \left\langle h_{i}\right|\right)+\sum x_{k_{i}}\left|k_{i}\right\rangle \left\langle k_{i}\right|\otimes\mathbb{I},
            \end{align*}
          \end{footnotesize}

          which is also block diagonal with respect to $\mathcal{H}\oplus \mathcal{H}^{\perp}$ and has only term I in the first block. Consequently,
          only on these will $U_{(2)}$ have a non-trivial action (as $U_{(2)}$ is of the form $\begin{bmatrix} U &0\\ 0 &\mathbb{I}_{\mathcal{H}^\perp} \end{bmatrix}$ wrt $\mathcal{H}\oplus\mathcal{H}^\perp$). Let us first evaluate
          the non-$\mathcal{H}$ part where we only need to apply the projector.
          The result after separating equations where possible is
          \begin{align*}
            \sum x_{h_{i}}\left|h_{i}\right\rangle \left\langle h_{i}\right|\otimes\left(\mathbb{I}-\left|h_{i}\right\rangle \left\langle h_{i}\right|\right)  \ge0, \text{ and }
            \sum(x_{k_{i}}-x_{k_{i}})\left|k_{i}\right\rangle \left\langle k_{i}\right|\otimes\mathbb{I}  \ge0,
          \end{align*}
          which imply
          $
            x_{h_{i}}\ge0.
          $
          The non-trivial part yields
          \begin{displaymath}
            \sum x_{h_{i}}\left|h_{i}h_{i}\right\rangle \left\langle h_{i}h_{i}\right|\ge\sum x_{g_{i}}E_{(2)}U_{(2)}\left|g_{i}g_{i}\right\rangle \left\langle g_{i}g_{i}\right|U_{(2)}^{\dagger}E_{(2)}
          \end{displaymath}
          completing the proof.
        \end{proof}
  \item \textbf{Bob accepts Alice's change.} The following holds:
        \begin{align*}
          \left|\psi_{(4)}\right\rangle & =\left(\sum_{i}\sqrt{p_{h_{i}}}\left|h_{i}h_{i}\right\rangle _{AB}+\sum_{i}\sqrt{p_{k_{i}}}\left|k_{i}k_{i}\right\rangle _{AB}\right)\otimes\left|m\right\rangle _{M} \\
          E_{(3)}U_{(3)}                & =E_{(3)}U_{BM}^{\text{SWP}\{\vec{h},m\}}                                                                                                                              \\
          Z_{(4)}^{A}                   & =Z_{(3)}^{A} \ \ \ \text{ and }\ \ \ \
          Z_{(4)}^{B}  =y\sum_{i}\left|h_{i}\right\rangle \left\langle h_{i}\right|+\sum_{i}y_{k_{i}}\left|k_{i}\right\rangle \left\langle k_{i}\right|_{B},
        \end{align*}
        where $E_{(3)}=\left(\sum\left|h_{i}\right\rangle \left\langle h_{i}\right|+\sum\left|k_{i}\right\rangle \left\langle k_{i}\right|\right)_{B}\otimes\mathbb{I}_{M}$.
        \begin{proof}
          We have to prove: $$ (1)\quad E_{(3)}\left|\psi_{(4)}\right\rangle =U_{(3)}\left|\psi_{(3)}\right\rangle $$
          and
          $$
          (2)\quad
            Z_{(4)}^{B}\otimes\mathbb{I}_{M}\ge E_{(3)}U_{(3)}\left(Z_{(3)}^{B}\otimes\mathbb{I}_{M}\right)U^\dagger_{(3)}E_{(3)}.
          $$ 
          Condition (1) can be shown by a direct application of $U^{\dagger}E$
          on $\left|\psi_{(4)}\right\rangle $, where $E,U$ denote
          $E_{(3)}$ and $U_{(3)}$, respectively, in this proof for ease of notation.

          Condition (2) can be established by first noting that \begin{footnotesize}
            \begin{align*}
               & EU\left(\mathbb{I}_{B}^{\{\vec{g},m\}}\otimes\mathbb{I}_{M}^{\{\vec{h},\vec{g},\vec{k},m\}}\right)U^{\dagger}E  =EU\left(\mathbb{I}_{B}^{\{m\}}\otimes\mathbb{I}_{M}^{\{\vec{h},\vec{g},\vec{k},m\}}\right)U^{\dagger}E+E\left(\mathbb{I}_{B}^{\{\vec{g}\}}\otimes\mathbb{I}_{M}^{\{\vec{h},\vec{g},\vec{k},m\}}\right)E \\
               & =EU\left(\mathbb{I}_{B}^{\{m\}}\otimes\mathbb{I}_{M}^{\{\vec{h},m\}}\right)U^{\dagger}E=\sum\left|h_{i}\right\rangle \left\langle h_{i}\right|\otimes\mathbb{I}_{M}^{\{m\}}.
            \end{align*}
          \end{footnotesize}
          Since the other term in $Z_{(3)}^{B}\otimes\mathbb{I}$ is not in the action space of $U$ it follows that
          \begin{displaymath}
            EU(Z_{(3)}^{B}\otimes\mathbb{I})U^{\dagger}E=y\sum\left|h_{i}\right\rangle \left\langle h_{i}\right|\otimes\mathbb{I}_{M}^{\{m\}}+\sum y_{k_{i}}\left|k_{i}\right\rangle \left\langle k_{i}\right|\otimes\mathbb{I}_{M}.
          \end{displaymath}
          It only remains to show that $Z_{(4)}^{B}\otimes\mathbb{I}_{M}\ge EU\left(Z_{(3)}^{B}\otimes\mathbb{I}_{M}\right)U^\dagger E$
          which holds as $y\sum\left|h_{i}\right\rangle \left\langle h_{i}\right|\otimes\mathbb{I}_{M}\ge y\sum\left|h_{i}\right\rangle \left\langle h_{i}\right|\otimes\mathbb{I}_{M}^{\{m\}}$
          and the $y_{k_{i}}$ term is common.
        \end{proof}
\end{enumerate}

Suppose that for each transition in the TDPG, the equation corresponding to Equation~\eqref{eq:MainConstraintInequality} can be satisfied. Then, as asserted, using the previous four steps for each transition, one directly obtains a \atul{
  WCF protocol together with a feasible solution to its corresponding dual (as in \cref{thm:dual} with projectors) certifying that the WCF protocol has the same bias as the TDPG. Formally (using the notation above) we have the following.}

\begin{definition}[TEF constraint]\label{def:TEFconstraint} A transition
  \begin{equation}
    \sum_{i=1}^{n_{k}}p_{k_{i}}\left\llbracket x_{k_{i}}\right\rrbracket +\sum_{i=1}^{n_{g}}p_{g_{i}}\left\llbracket x_{g_{i}}\right\rrbracket \to\sum_{i=1}^{n_{h}}p_{h_{i}}\left\llbracket x_{h_{i}}\right\rrbracket +\sum_{i=1}^{n_{k}}p_{k_{i}}\left\llbracket x_{k_{i}}\right\rrbracket \label{eq:transitionForTEF}
  \end{equation}
  satisfies the \emph{TEF constraint} if there is a unitary matrix $U_{(2)}$ that satisfies the inequality
  \begin{equation}
    \sum_{i=1}^{n_{h}}x_{h_{i}}\left|h_{i}h_{i}\right\rangle \left\langle h_{i}h_{i}\right|_{AM}\ge\sum_{i=1}^{n_{g}}x_{g_{i}}E_{(2)}^{h}U_{(2)}\left|g_{i}g_{i}\right\rangle \left\langle g_{i}g_{i}\right|_{AM}U^\dagger_{(2)}E_{(2)}^{h}\label{eq:ProtocolConstraintEquation}
  \end{equation}
  and the honest action constraint
  $
    U_{(2)}\left|v\right\rangle =\left|w\right\rangle
  $,
  where ${{\left|h_{i}\right\rangle},{\left|g_{i}\right\rangle}}$
  are orthonormal basis vectors,
  \begin{displaymath}
    \left|v\right\rangle =\mathcal{N}\left(\sum\sqrt{p_{g_{i}}}\left|g_{i}g_{i}\right\rangle _{AM}\right)\ \
    \text{ and }\ \
    \left|w\right\rangle =\mathcal{N}\left(\sum\sqrt{p_{h_{i}}}\left|h_{i}h_{i}\right\rangle _{AM}\right)
  \end{displaymath}
  for $\mathcal{N}(\left|\psi\right\rangle )=\left|\psi\right\rangle /\sqrt{\left\langle \psi|\psi\right\rangle }$,
  $E^{h}=\left(\sum_{i=1}^{n_{h}}\left|h_{i}\right\rangle \left\langle h_{i}\right|_{A}+\sum\left|k_{i}\right\rangle \left\langle k_{i}\right|_{A}\right)\otimes\mathbb{I}_{M}$
  with $U_{(2)}$'s non-trivial action restricted to $\text{span}\left\{ \{\left|g_{i}g_{i}\right\rangle _{AM}\},\{\left|h_{i}h_{i}\right\rangle _{AM}\}\right\} $,
  and $\left|k_{i}\right\rangle $ correspond to the points that
  are left unchanged in the transition.

\end{definition}

\begin{theorem}\label{thm:TEFconstraint} Suppose for each transition of a TDPG, the TEF constraint (see \cref{def:TEFconstraint}) can be satisfied. Then, there exists a WCF protocol that has the same TDPG (up to some repetition in frames\footnote{The new TDPG has some extra frames where nothing changes (from the point of view of the TDPG)}).
\end{theorem}

We implicitly used \cref{rem:projBeforeAndAfterDual} and  %
\cref{thm:dual}.

\subsection{TEF Functions/Transitions}\label{subsec:TEFfunctions}

It is evident that the TEF constraint (see \cref{def:TEFconstraint} above) can be simplified by neglecting the parts of the Hilbert space where $U_{(2)}$ behaves as identity. Thus, an equivalent formulation of \cref{def:TEFconstraint} is the following.

\begin{definition}[TEF constraint (simpler formulation), unitary solves a transition/function, TEF transitions/functions]\label{def:TEFconstraint_} Let $g\to h$ be a transition (see \cref{def:transition}), with the associated function $ t= h - g = \sum_{i=1}^{n_h} p_{h_i} \llbracket{x_{h_i}}\rrbracket - \sum_{i=1}^{n_g} p_{g_i} \llbracket{x_{g_i}}\rrbracket$,
  where all $p_{h_i}$ and $p_{g_i}$ are positive and let $\{\{\ket{g_i}\}_{i=1}^{n_g},\{\ket{h_i}\}_{i=1}^{n_h}\}$ constitute an orthonormal basis, spanning $\mathcal{H}$. We say $U$ (acting on $\mathcal{H}$) \emph{solves} the transition $t$ if $U$ satisfies the following \emph{TEF constraint},
  $$\sum_{i=1}^{n_h} x_{h_i} \ket{h_i}\bra{h_i} \ge \sum_{i=1}^{n_g} x_{g_i} E U \ket{g_i}\bra{g_i} U^{\dagger} E, \quad \text{and } \quad EU\underbrace{\sum_{i=1}^{n_g} \sqrt{p_{g_i}} \ket {g_i}}_{\ket{v}} = \underbrace{\sum_{i=1}^{n_h} \sqrt{p_{h_i}} \ket{h_i}}_{\ket{w}},$$
  where $E = \sum_{i=1}^{n_h} \ket{h_i}\bra{h_i}$.
  The transition (function) is a \emph{TEF transition (function)} if there is a unitary matrix that solves it. %

\end{definition}

As alluded to earlier, one may use TEF functions (instead of EBM or valid functions), without loss of generality.
\begin{lemma}[TEF = closure of EBM = valid]
  The set of the TEF functions, the set of valid functions
  and the closure of the set of the EBM functions
  are the same. \label{lem:setequality}
\end{lemma}

We defer the proof of \cref{lem:setequality} to Appendix \ref{sec:TEF-functions=00003DValid-functions=00003DclosureOfEBMfunctions} as we do not need it to prove our result.
We do note, however, that \cref{lem:setequality} above, allows one to circumvent the notion of strictly valid functions, (arguably) simplifying the analysis.

\subsection{Special case: the blinkered unitary\label{subsec:BlinkeredUnitary}}
In this subsection, we use the more explicit notation from \cref{def:TEFconstraint}, \cref{subsec:framework} to illustrate how TEF easily allows one to construct WCF protocols approaching bias $1/6$. To this end, we introduce an important class of unitaries we call \emph{Blinkered Unitaries}. For clarity, to describe the TEF constraint (as in \cref{def:TEFconstraint}), we use $U$ instead of $U_{(2)}$ and $E$ instead of $E^h_{(2)}$. Given a transition (as in Equation~\eqref{eq:transitionForTEF}), the associated Blinkered Unitary is defined as

\begin{displaymath}
  U=\left|w\right\rangle \left\langle v\right|+\left|v\right\rangle \left\langle w\right|+\sum_{i}\left|v_{i}\right\rangle \left\langle v_{i}\right|+\sum_{i}\left|w_{i}\right\rangle \left\langle w_{i}\right|+\mathbb{I}^{\text{outside }\text{\ensuremath{\mathcal{H}}}},
\end{displaymath}
where $\mathcal{H}=\text{span}\left\{ \left|g_{1}g_{1}\right\rangle ,\left|g_{2}g_{2}\right\rangle \dots,\left|h_{1}h_{1}\right\rangle ,\left|h_{2},h_{2}\right\rangle \dots\right\} $.
We can ignore the last term and restrict our analysis
to the $\mathcal{H}$-operator space, where $\left|v\right\rangle ,\{\left|v_{i}\right\rangle \}$
form a complete orthonormal basis with respect to $\text{span}\{\left|g_{i}g_{i}\right\rangle \}$, and so do $\left|w\right\rangle ,\{\left|w_{i}\right\rangle \}$
for $\text{span}\{\left|h_{i}h_{i}\right\rangle \}$.
What makes blinkered unitaries useful is that they satisfy the TEF constraint (as stated in \cref{def:TEFconstraint}), when the transition is a non-trivial \emph{basic} move, i.e. a merge (see \cref{exa:merge}) or a split (see \cref{exa:split}).
\begin{itemize}
  \item Merge: $g_{1},g_{2}\to h_{1}$\\
        Using the definitions, we have
        \begin{small}
          \begin{displaymath}
            \left|v\right\rangle =\frac{\sqrt{p_{g_{1}}}\left|g_{1}g_{1}\right\rangle +\sqrt{p_{g_{2}}}\left|g_{2}g_{2}\right\rangle }{N},\,\left|v_{1}\right\rangle =\frac{\sqrt{p_{g_{2}}}\left|g_{1}g_{1}\right\rangle -\sqrt{p_{g_{1}}}\left|g_{2}g_{2}\right\rangle }{N},\,\left|w\right\rangle =\left|h_{1}h_{1}\right\rangle
          \end{displaymath}
        \end{small}

        with $N=\sqrt{p_{g_{1}}+p_{g_{2}}}$ and
        $
          U=\left|w\right\rangle \left\langle v\right|+\left|v\right\rangle \left\langle w\right|+\left|v_{1}\right\rangle \left\langle v_{1}\right|=U^\dagger.
        $
        We evaluate
        \begin{displaymath}
          EU\left|g_{1}g_{1}\right\rangle =\frac{\sqrt{p_{g_{1}}}\left|w\right\rangle }{N} \text{ and } EU\left|g_{2}g_{2}\right\rangle =\frac{\sqrt{p_{g_{2}}}\left|w\right\rangle }{N}.
        \end{displaymath}
        Using these, the TEF constraint
        $
          x_{h}\left|h_{1}h_{1}\right\rangle \left\langle h_{1}h_{1}\right|\ge\sum x_{g_{i}}EU\left|g_{i}g_{i}\right\rangle \left\langle g_{i}g_{i}\right|U^{\dagger}E
        $ becomes
        $
          x_{h}\ge\frac{p_{g_{1}}x_{g_{1}}+p_{g_{2}}x_{g_{2}}}{N^{2}},
        $
        which is precisely the merge condition (see \cref{exa:merge}).

  \item Split: $g_{1}\to h_{1},h_{2}$\\
        Again, from the definitions, we construct
        \begin{small}
          \begin{displaymath}
            \left|v\right\rangle =\left|g_{1}g_{1}\right\rangle ,\,\left|w\right\rangle =\frac{\sqrt{p_{h_{1}}}\left|h_{1}h_{1}\right\rangle +\sqrt{p_{h_{2}}}\left|h_{2}h_{2}\right\rangle }{N},\,\left|w_{1}\right\rangle =\frac{\sqrt{p_{h_{2}}}\left|h_{1}h_{1}\right\rangle -\sqrt{p_{h_{1}}}\left|h_{2}h_{2}\right\rangle }{N}
          \end{displaymath}
        \end{small}

        with $N=\sqrt{p_{h_{1}}+p_{h_{2}}}$ and
        $
          U=\left|v\right\rangle \left\langle w\right|+\left|w\right\rangle \left\langle v\right|+\left|w_{1}\right\rangle \left\langle w_{1}\right|=U^{\dagger}.
        $
        We evaluate $EU\left|g_{1}g_{1}\right\rangle =\left|w\right\rangle $
        which we substitute into the TEF constraint to obtain
        \begin{displaymath}
          x_{h_{1}}\left|h_{1}h_{1}\right\rangle \left\langle h_{1}h_{1}\right|+x_{h_{2}}\left|h_{2}h_{2}\right\rangle \left\langle h_{2}h_{2}\right|-x_{g_{1}}\left|w\right\rangle \left\langle w\right|\ge0.
        \end{displaymath}
        This yields the matrix equation
        \begin{align*}
           & \left[\begin{array}{cc}
                       x_{h_{1}}    \\
                        & x_{h_{2}}
                     \end{array}\right]-\frac{x_{g_{1}}}{N^{2}}\left[\begin{array}{cc}
                                                                       p_{h_{1}}                 & \sqrt{p_{h_{1}}p_{h_{2}}} \\
                                                                       \sqrt{p_{h_{1}}p_{h_{2}}} & p_{h_{2}}
                                                                     \end{array}\right]  \ge0                                            \\
           & \mathbb{I}\ge \frac{x_{g_1}}{N^{2}}\left[\begin{array}{cc}
                                                          \frac{p_{h_1}}{x_{h_{1}} }                                  & \sqrt{\frac{p_{h_1}}{x_{h_{1}} }\frac{p_{h_2}}{x_{h_{2}} }} \\
                                                          \sqrt{\frac{p_{h_1}}{x_{h_{1}} }\frac{p_{h_2}}{x_{h_{2}} }} & \frac{p_{h_2}}{x_{h_{2}} }
                                                        \end{array}\right] \\
           & \frac{x_{g_{1}}}{N^{2}}\left(\frac{p_{h_{1}}}{x_{h_{1}}}+\frac{p_{h_{2}}}{x_{h_{2}}}\right)  \le 1,
        \end{align*}where in the first step we used the fact that for $F>0$, $F-M\ge 0 \equiv \mathbb{I}-\sqrt{F}^{-1}M\sqrt{F}^{-1}\ge 0$, and the last equation is obtained by writing the matrix as $\ket{\psi}\bra{\psi}$, and then demanding $1\ge\langle\psi|\psi\rangle$.
        This last equation is
        exactly the
        split condition (see \cref{exa:split}).
\end{itemize}

The above two conditions can be readily generalized for an $m\to1$ point merge and a $1\to n$ points split, respectively (see Appendix \cref{sec:Blinkered-transition}). Furthermore, for a general $m\to n$: $g_{1},g_{2}\dots g_{m}\to h_{1},h_{2}\dots h_{n}$ transition, the TEF constraint corresponding to the Blinkered Unitary reduces to the following scalar condition (see Appendix \cref{sec:Blinkered-transition} for a proof),
$$
  \frac{1}{\sum_{i=1}^{m}p_{g_{i}}x_{g_{i}}}\ge\sum_{i=1}^{n}p_{h_{i}}\frac{1}{x_{h_{i}}}.
$$
In words, the general $m\to n$ transition affected by the blinkered unitary may be viewed as an $m\to 1$ merge followed by a $1\to n$ split.

Consequently, blinkered unitaries are enough to convert the
$1/6$ game into an explicit protocol. However, they fall short for point games going below this bias which seem to require \emph{advanced} moves---moves beyond splits and merges. Next, we construct the unitaries for such moves to obtain WCF protocols approaching bias $1/10$.
\subsection{Approaching bias \texorpdfstring{$1/10$}{1/10}} \label{subsec:Bias1by10gameprotocol}

In \cref{subsec:mochontipg} we briefly outlined
Mochon's family of TIPGs approaching bias $\epsilon(k)=1/(4k+2)$, where $k$ is the number of points involved in the non-trivial step. Here, we detail the game for
$k=2$, and explicitly find the unitaries that solve the transitions used in the game.

All of Mochon's TIPGs, assume an equally spaced $n$-point lattice given by $x_{j}=x_{0}+j\delta x$
where $\delta x=\delta y$ is small and $x_{0}$ is specified shortly.\footnote{Essentially, $x_0$ provides a bound on $P^*_B$.} Similarly $y_{j}=y_{0}+j\delta y$ and we define $\Gamma_{k+1}=y_{n-k}=x_{n-k}$. We focus on the ``ladder'' stage \atul{(see  \cref{fig:MochonGameStages})}. We first constrain the weights of points along the $x$-axis, by requiring that they arise from the splitting of one point with weight $1/2$ at $(1,0)$ (similarly for the $y$-axis). Let $P(x_{j})$ denote the probability weight associated with the point $(x_{j},0)$ which is such that \begin{displaymath}
  \sum_{j=1}^{n}P(x_{j})=\frac{1}{2} \text{ and }\sum_{j=1}^{n}\frac{P(x_{j})}{x_{j}}=\frac{1}{2}.
\end{displaymath}
Similarly with the point $(0,y_{j})$ we associate $P(y_{j})$ where
$y_{j}=x_{j}$ as we also assume that $x_{0}=y_{0}$. These choices
explicitly impose symmetry between Alice and Bob which in turn means
that we only have to do the analysis for one of them.

We now use Mochon's assignment (see Equation~\eqref{eq:fAssignmentInitial}) to (partially) specify weights on points along vertical lines (see  \cref{fig:1by10correct}). In particular, given set of points (with distinct $y$-coordinates but the same $x$-coordinate), we use $$
  \frac{f(y_{j})c(x_{l})}{\prod_{k\neq j}(y_{k}-y_{j})}
$$ to specify the weight on the point $(x_l,y_j)$
where $f(y_{i})=(y_{-2}-y_{i})\left(\Gamma_{1}-y_{i}\right)(\Gamma_{2}-y_{i})$.
\begin{figure}[H]
  \begin{centering}
    \includegraphics[scale=0.9]{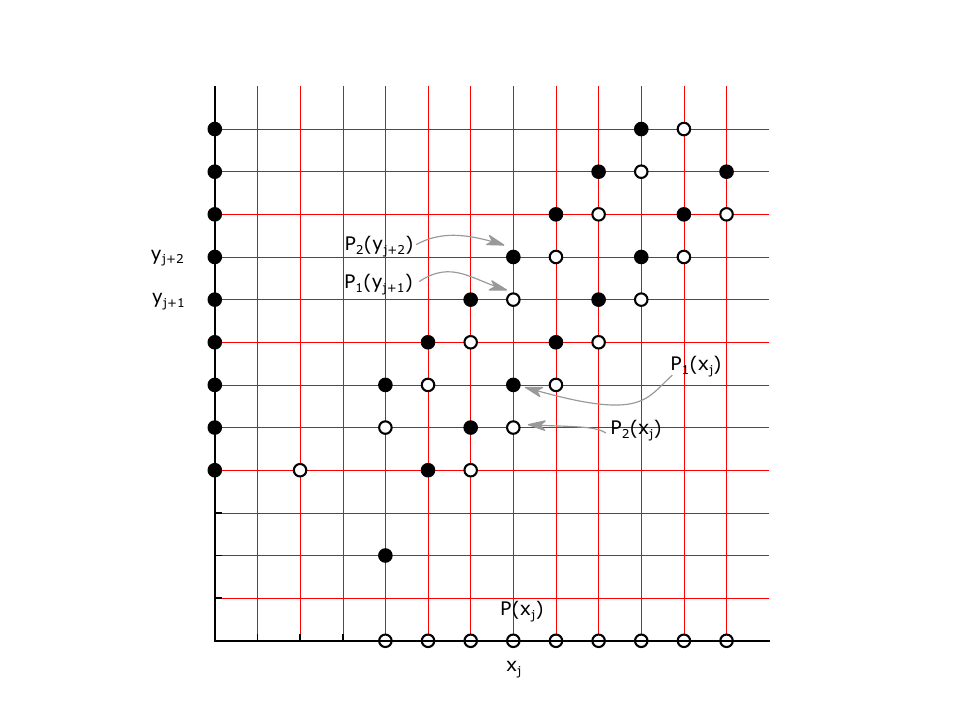}
    \par\end{centering}
  \caption{$1/10$-bias TIPG: The $3\to2$ move \label{fig:1by10correct}}
\end{figure}
Applying the assignment to the points \atul{highlighted in  \cref{fig:1by10correct} yields} %
\begin{align*}
  P_{2}(y_{j+2}) & =\frac{-f(y_{j+2})c(x_{j})}{4\cdot 3(\delta y)^{2}y_{j+2}},\
  P_{1}(y_{j+1})  =\frac{-f(y_{j+1})c(x_{j})}{3\cdot 2(\delta y)^{2}y_{j+1}},    \\
  P_{1}(x_{j})   & =\frac{-f(y_{j-1})c(x_{j})}{3 \cdot 2(\delta y)^{2}y_{j-1}},\
  P_{2}(x_{j})  =\frac{-f(y_{j-2})c(x_{j})}{4\cdot 3(\delta y)^{2}y_{j-2}},\
  P(x_{j})  =\frac{f(0)c(x_{j})\delta y}{y_{j+2}y_{j+1}y_{j-1}y_{j-2}}
\end{align*}
where we added the minus sign to account for the fact that $f$ is
negative for coordinates between $y_{-2}$ and $\Gamma_{1}$. Imposing
the symmetry constraint $P_{1}(y_{j})=P_{1}(x_{j})$ we get
$
  c(x_{j})=\frac{c_{0}f(x_{j})}{x_{j}}
$,
where $c_{0}$ is a constant. Similarly, the symmetry constraint for $P_2$ entails $P_{2}(y_{j})=P_{2}(x_{j})$. Finally, we can evaluate
$
  P(x_{j})=\frac{c_{0}x_{0}(x_{0}-x_{j})}{x_{j}^{5}}\delta x+\mathcal{O}(\delta x^{2})
$
which, in the limit $\delta x\rightarrow 0$, means that
\begin{displaymath}
  \sum P(x_{j})=\frac{1}{2}=\sum\frac{P(x_{j})}{x_{j}}\to\int_{x_{0}}^{\Gamma}\frac{(x_{0}-x)dx}{x^{5}}=\int_{x_{0}}^{\Gamma}\frac{(x_{0}-x)dx}{x^{6}}.
\end{displaymath}
This evaluates to
\begin{align*}
  x_{0}\int_{x_{0}}^{\Gamma}\left(\frac{1}{x^{5}}-\frac{1}{x^{6}}\right)dx  =\int_{x_{0}}^{\Gamma}\left(\frac{1}{x^{4}}-\frac{1}{x^{5}}\right)dx\Rightarrow
  x_{0}  =\frac{3}{5}\implies\epsilon=\frac{3}{5}-\frac{1}{2}=\frac{1}{10}
\end{align*}
as expected. These calculations help us below when we explicitly find unitaries that solve the \emph{advanced} moves which appear in this game. These unitaries, together with those for the basic moves and TEF, yield WCF protocols approaching bias $1/10$. Henceforth, %
\atul{we use the simpler notation introduced in \cref{subsec:TEFfunctions} for clarity.}

\subsubsection{The \texorpdfstring{$3\to2$}{3 -> 2} move and its validity}\label{subsubsec:3to2}
\atul{Our goal here is to find a unitary that solve the move taking 3 points to 2 points, as highlighted in  \cref{fig:1by10correct}. To this end, we first setup some notation and write down the TEF constraint corresponding to a general 3 to 2 move, using it. We then show how to construct the appropriate unitary that satisfies this TEF constraint.} %

Recall \atul{from \cref{def:TEFconstraint_}} that
\begin{displaymath}
  \left|v\right\rangle =\frac{\sqrt{p_{g_{1}}}\left|g_{1}\right\rangle +\sqrt{p_{g_{2}}}\left|g_{2}\right\rangle +\sqrt{p_{g_{3}}}\left|g_{3}\right\rangle }{N_{g}}
\end{displaymath}
and let
\begin{align*}
  \left|v_{1}\right\rangle & =\frac{\sqrt{p_{g_{3}}}\left|g_{2}\right\rangle -\sqrt{p_{g_{2}}}\left|g_{3}\right\rangle }{N_{v_{1}}},\ \ \
  \left|v_{2}\right\rangle   =\frac{-\frac{(p_{g_{2}}+p_{g_{3}})}{\sqrt{p_{g_{1}}}}\left|g_{1}\right\rangle +\sqrt{p_{g_{2}}}\left|g_{2}\right\rangle +\sqrt{p_{g_{3}}}\left|g_{3}\right\rangle }{N_{v_{2}}}
\end{align*}
where $N_{v_{1}}^{2}=p_{g_{3}}+p_{g_{2}}$ and $N_{v_{2}}^{2}=\frac{(p_{g_{2}}+p_{g_{3}})^{2}}{p_{g_{1}}}+p_{g_{2}}+p_{g_{3}}$.
Also,
\begin{align*}
  \left|w\right\rangle  =\frac{\sqrt{p_{h_{1}}}\left|h_{1}\right\rangle +\sqrt{p_{h_{2}}}\left|h_{2}\right\rangle }{N_{h}} \text{ and }
  \left|w_{1}\right\rangle   =\frac{\sqrt{p_{h_{_{2}}}}\left|h_{1}\right\rangle -\sqrt{p_{h_{1}}}\left|h_{2}\right\rangle }{N_{h}}.
\end{align*}
Now we define
\begin{align*}
  \left|v'_{1}\right\rangle & =\cos\theta\left|v_{1}\right\rangle +\sin\theta\left|v_{2}\right\rangle \text{ and }
  \left|v_{2}'\right\rangle   =\sin\theta\left|v_{1}\right\rangle -\cos\theta\left|v_{2}\right\rangle,
\end{align*}
where $\cos\theta\approx 1$, and the full unitary as
\begin{displaymath}
  U=\left|w\right\rangle \left\langle v\right|+\left(\alpha\left|v_{1}'\right\rangle +\beta\left|w_{1}\right\rangle \right)\left\langle v_{1}'\right|+\left|v_{2}'\right\rangle \left\langle v_{2}'\right|+\left(\beta\left|v_{1}'\right\rangle -\alpha\left|w_{1}\right\rangle \right)\left\langle w_{1}\right|+\left|v\right\rangle \left\langle w\right|,
\end{displaymath}
where $\left|\alpha\right|^{2}+\left|\beta\right|^{2}=1$ for $\alpha,\beta\in\mathbb{C}$.\footnote{There is some freedom in choosing $U$ in the sense that $\alpha\left|v\right\rangle +\beta\left|w_{1}\right\rangle $ would also work instead of $\alpha\left|v_1'\right\rangle +\beta\left|w_{1}\right\rangle$ (in that case $\ket{v}\bra{w}$ should be replaced by $\ket{v_1}\bra{w}$), as these do not influence the constraint equation.} %
We need terms of the form $EU\ket{g_i}$ with $E=\mathbb{I}^{\{h_i\}}$. This entails that $EU$ acts on the $\{\left|g_{i}\right\rangle \}$
space as
\begin{align*}
  EUE_{g}  =\left|w\right\rangle \left\langle v\right|+\beta\left|w_{1}\right\rangle \left\langle v_{1}'\right| =\left|w\right\rangle \left\langle v\right|+\beta\left|w_{1}\right\rangle \left(\cos\theta\left\langle v_{1}\right|+\sin\theta\left\langle v_{2}\right|\right),
\end{align*}
where $E_g$ is the projector on the $\{\left|g_{i}\right\rangle \}$ space.
Consequently we have
\begin{align*}
  EU\left|g_{1}\right\rangle & =\frac{\sqrt{p_{g_{1}}}}{N_{g}}\left|w\right\rangle +\left[\cos\theta\cdot 0-\sin\theta\frac{p_{g_{2}}+p_{g_{3}}}{\sqrt{p_{g_{1}}}N_{v_{2}}}\right]\beta\left|w_{1}\right\rangle            \\
  EU\left|g_{2}\right\rangle & =\frac{\sqrt{p_{g_{2}}}}{N_{g}}\left|w\right\rangle +\left[\cos\theta\frac{\sqrt{p_{g_{3}}}}{N_{v_{1}}}+\sin\theta\frac{\sqrt{p_{g_{2}}}}{N_{v_{2}}}\right]\beta\left|w_{1}\right\rangle    \\
  EU\left|g_{3}\right\rangle & =\frac{\sqrt{p_{g_{3}}}}{N_{g}}\left|w\right\rangle +\left[-\cos\theta\frac{\sqrt{p_{g_{2}}}}{N_{v_{1}}}+\sin\theta\frac{\sqrt{p_{g_{3}}}}{N_{v_{2}}}\right]\beta\left|w_{1}\right\rangle .
\end{align*}
Recall that the TEF constraint requires
\begin{displaymath}
  \sum x_{h_{i}}\left|h_{i}\right\rangle \left\langle h_{i}\right|-\sum x_{g_{i}}EU\left|g_{i}\right\rangle \left\langle g_{i}\right|U^{\dagger}E\ge0
\end{displaymath}
where the first sum becomes
\begin{displaymath}
  \left[\begin{array}{cc}
      \left\langle x_{h}\right\rangle & \frac{\sqrt{p_{h_{1}}p_{h_{2}}}}{N_{h}^{2}}(x_{h_{1}}-x_{h_{2}}) \\
      \text{h.c.}                     & \frac{p_{h_{2}}x_{h_{1}}+p_{h_{1}}x_{h_{2}}}{N_{h}^{2}}
    \end{array}\right]
\end{displaymath}
in the $\left|w\right\rangle ,\left|w_{1}\right\rangle $ basis. Since
we plan to use the $3\to2$ move with one point on the axis, we take
$x_{g_{1}}=0$. Consequently we only need to evaluate
\begin{align*}
  x_{g_{2}}EU\left|g_{2}\right\rangle \left\langle g_{2}\right|U^{\dagger}E\dot{=} & x_{g_{2}}\left[\begin{array}{cc}
                                                                                                        \frac{p_{g_{2}}}{N_{g}^{2}} & \beta\left(\cos\theta\frac{\sqrt{p_{g_{3}}p_{g_{2}}}}{N_{g}N_{v_{1}}}+\sin\theta\frac{p_{g_{2}}}{N_{g}N_{v_{2}}}\right)    \\
                                                                                                        \text{h.c.}                 & \left(\cos\frac{\sqrt{p_{g_{3}}}}{N_{v_{1}}}+\sin\theta\frac{\sqrt{p_{g_{2}}}}{N_{v_{2}}}\right)^{2}\left|\beta\right|^{2}
                                                                                                      \end{array}\right] \\
  x_{g_{3}}EU\left|g_{3}\right\rangle \left\langle g_{3}\right|U^{\dagger}E\dot{=} & x_{g_{3}}\left[\begin{array}{cc}
                                                                                                        \frac{p_{g_{3}}}{N_{g}^{2}} & \beta\left(-\cos\theta\frac{\sqrt{p_{g_{2}}p_{g_{3}}}}{N_{g}N_{v_{1}}}+\sin\theta\frac{p_{g_{3}}}{N_{g}N_{v_{2}}}\right) \\
                                                                                                        \text{h.c.}                 & \left(-\cos\frac{\sqrt{p_{g_{2}}}}{N_{v_{1}}}+\sin\frac{\sqrt{p_{g_{3}}}}{N_{v_{2}}}\right)^{2}\left|\beta\right|^{2}
                                                                                                      \end{array}\right]
\end{align*}
which means that the constraint equation becomes {\footnotesize 
    \begin{displaymath}
      \left[\begin{array}{cc}
          \left\langle x_{h}\right\rangle -\left\langle x_{g}\right\rangle & \frac{\sqrt{p_{h_{1}}p_{h_{2}}}}{N_{h}^{2}}(x_{h_{1}}-x_{h_{2}})-\beta\cos\theta\frac{\sqrt{p_{g_{2}}p_{g_{3}}}}{N_{g}N_{v_{1}}}(x_{g_{2}}-x_{g_{3}})-\beta\sin\theta\left\langle x_{g}\right\rangle \frac{N_{g}}{N_{v_{2}}}                                                                                                                                                 \\
          \text{h.c.}                                                      & \frac{p_{h_{2}}x_{h_{1}}+p_{h_{1}}x_{h_{2}}}{N_{h}^{2}}-\left|\beta\right|^{2}\left[\frac{\cos^{2}\theta}{N_{v_{1}}^{2}}(p_{g_{3}}x_{g_{2}}+p_{g_{2}}x_{g_{3}})+\frac{\sin^{2}\theta}{\left(N_{v_{2}}^{2}/N_{g}^{2}\right)}\left\langle x_{g}\right\rangle +\frac{2\cos\theta\sin\theta\sqrt{p_{g_{3}}p_{g_{2}}}}{N_{v_{1}}N_{v_{2}}}\left(x_{g_{2}}-x_{g_{3}}\right)\right]
        \end{array}\right]\ge0.
    \end{displaymath}
  }Since this transition is average non-decreasing
viz. $\left\langle x_{h}\right\rangle -\left\langle x_{g}\right\rangle \ge0$ (see \cref{lem:fAssignmentLemma} and \cref{lem:expectationLemma}), we set the off-diagonal elements of the matrix above to zero and show
that the second diagonal element
is positive. Setting the off-diagonal to zero one can obtain $\theta$ by solving
the quadratic equation in terms of $\beta$ although the expression is not
particularly pretty. To establish existence and positivity we need
to simplify our expressions.

So far, everything was exact. To proceed, we write  $\theta\frac{N_{g}}{N_{v_{2}}}=\mathcal{O}(\delta y)$
at most (where $\delta y=\delta x$ is the lattice spacing) and we take
$\delta y$ to be small. Thus, to first order in $\theta\frac{N_{g}}{N_{v_{2}}}$, the constraints become
\begin{displaymath}
  \frac{\frac{\sqrt{p_{h_{1}}p_{h_{2}}}}{N_{h}^{2}}(x_{h_{1}}-x_{h_{2}})-\beta\frac{\sqrt{p_{g_{2}}p_{g_{3}}}}{N_{g}N_{v_{1}}}(x_{g_{2}}-x_{g_{3}})}{\beta\left\langle x_{g}\right\rangle }=\theta\frac{N_{g}}{N_{v_{2}}}+\mathcal{O}(\delta y^{2})
\end{displaymath}
and
\begin{displaymath}
  \frac{p_{h_{2}}x_{h_{1}}+p_{h_{1}}x_{h_{2}}}{N_{h}^{2}}-\left|\beta\right|^{2}\left[\frac{p_{g_{3}}x_{g_{2}}+p_{g_{2}}x_{g_{3}}}{N_{v_{1}}^{2}}+2\theta\frac{N_{g}}{N_{v_{2}}}\frac{\sqrt{p_{g_{3}}p_{g_{2}}}}{N_{g}N_{v_{1}}}(x_{g_{2}}-x_{g_{3}})\right]+\mathcal{O}(\delta y^{2})\ge0.
\end{displaymath}
If our claim is wrong when we evaluate $\theta\frac{N_{g}}{N_{v_{2}}}$,
we will get zero order terms but as we show later, indeed,
$\theta\frac{N_{g}}{N_{v_{2}}}=\mathcal{O}(\delta y^{2})$. With respect to  \cref{fig:1by10correct} we have
\begin{footnotesize}
  \begin{align*}
    P_{2}(y_{j+2}) & =p_{h_{2}}=\frac{-f(y_{j+2})}{4\cdot 3\delta y^{2}y_{j+2}},\
    P_{1}(y_{j+1}) =p_{g_{3}}=\frac{-f(y_{j+1})}{3\cdot 2\delta y^{2}y_{j+1}}      \\
    P_{1}(x_{j})   & =p_{h_{1}}=\frac{-f(y_{j-1})}{3 \cdot 2\delta y^{2}y_{j-1}},\
    P_{2}(x_{j})  =p_{g_{2}}=\frac{-f(y_{j-2})}{4 \cdot 3\delta y^{2}y_{j-2}},\
    P(x_{j})  =p_{g_{1}}=\frac{f(0)\delta y}{y_{j+2}y_{j+1}y_{j-1}y_{j-2}},
  \end{align*}
\end{footnotesize}
where we assumed $f(0)>0$ and $f(y)<0$ for $y>y_{0}'$, $y_{0}'=y_{0}+\delta y$, and we
scaled by $\delta y$.
We now convert all expressions to first
order in $\delta y$:
\begin{align*}
  f(y_{j+m}) =f(y_{j})+\frac{\partial f}{\partial y}m\delta y+\mathcal{O}(\delta y^{2})\Rightarrow
  \frac{1}{y_{j+m}} =\frac{1}{y_{j}}-m\frac{\delta y}{y_{j}^{2}}+\mathcal{O}(\delta y^{2}),
\end{align*}
where $\frac{\partial f}{\partial y}$ is $\frac{\partial f(y)}{\partial y}|_{y_{j}}$.
We define and evaluate
\begin{align*}
  P_{k}^{m} & =\frac{-f(y_{j+m})}{k\delta y^{2}y_{j+m}} =\frac{1}{ky_{j}\delta y^{2}}\left[-f-m\delta y\left(\frac{\partial f}{\partial y}-\frac{f}{y_{j}}\right)+\mathcal{O}(\delta y^{2})\right],
\end{align*}
where $f$ means $f(y_{j})$. In this notation
\begin{align*}
  p_{h_{2}} =P_{12}^{2},\,p_{h_{1}}=P_{6}^{-1}\ \ \ \text{ and }\ \ \
  p_{g_{2}}  =P_{12}^{-2},\,p_{g_{3}}=P_{6}^{1}.
\end{align*}
With an eye on the off-diagonal condition we evaluate
\begin{displaymath}
  P_{k_{1}}^{m_{1}}P_{k_{2}}^{m_{2}}=\frac{1}{k_{1}k_{2}}\left(\frac{1}{y_{j}\delta y^{2}}\right)^{2}\left[f^{2}+f\delta y\left(\frac{\partial f}{\partial y}-\frac{f}{y_{j}}\right)\left(m_{1}+m_{2}\right)+\mathcal{O}(\delta y^{2})\right]
\end{displaymath}
and
\begin{displaymath}
  P_{k_{1}}^{m_{1}}+P_{k_{2}}^{m_{2}}=\frac{1}{y_{j}\delta y^{2}}\left[-\left(\frac{1}{k_{1}}+\frac{1}{k_{2}}\right)f-\left(\frac{m_{1}}{k_{1}}+\frac{m_{2}}{k_{2}}\right)\delta y\left(\frac{\partial f}{\partial y}-\frac{f}{y_{j}}\right)+\mathcal{O}(\delta y^{2})\right].
\end{displaymath}
Moreover, we have
\begin{align*}
  \sqrt{p_{h_{1}}p_{h_{2}}} & =\sqrt{P_{12}^{2}P_{6}^{-1}}=\frac{1}{y_{j}\delta y^{2}}\sqrt{\frac{1}{12\cdot 6}\left[f^{2}+f\delta y\left(\frac{\partial f}{\partial y}-\frac{f}{y_{j}}\right)+\mathcal{O}(\delta y^{2})\right]} \\
  N_{h}^{2}                 & =P_{12}^{2}+P_{6}^{-1}=\frac{1}{4y_{j}\delta y^{2}}\left[-f+\mathcal{O}(\delta y^{2})\right],
\end{align*}
and similarly

\begin{align*}
   & \sqrt{p_{g_{2}}p_{g_{3}}}  =\sqrt{P_{12}^{-2}P_{6}^{1}}=\frac{1}{y_{j}\delta y^{2}}\sqrt{\frac{1}{12\cdot 6}\left[f^{2}-f\delta y\left(\frac{\partial f}{\partial y}-\frac{f}{y_{j}}\right)+\mathcal{O}(\delta y^{2})\right]} \\
   & N_{g}^{2}  =P_{12}^{-2}+P_{6}^{1}+p_{g_{1}} =\frac{1}{4y_{j}\delta y^{2}}\left[-f+\mathcal{O}(\delta y^{2})\right]\text{ and }
  N_{v_{1}}^{2}  =\frac{1}{4y_{j}\delta y^{2}}\left[-f+\mathcal{O}(\delta y^{2})\right],
\end{align*}where we already neglected the terms that contribute to the ratio $\frac{N_{g}}{N_{v_{2}}}$ in higher than first order.
Actually, for $\beta=1$
\begin{small}
  \begin{equation*}
    \theta\frac{N_{g}}{N_{v_{2}}}  =\frac{4\sqrt{\frac{1}{12\cdot 6}}(-3\delta y)\left[f (\cancel{1}+\frac{\delta y}{2f}\left(\frac{\partial f}{\partial y}-\frac{f}{y_{j}}\right))-f(\cancel{1}-\frac{\delta y}{2f}\left(\frac{\partial f}{\partial y}-\frac{f}{y_{j}}\right))+\mathcal{O}(\delta y^{2})\right]}{\langle x_{g}\rangle }=\mathcal{O}(\delta y^{2}).
  \end{equation*}
\end{small}This shows that to first order the off-diagonal term is zero for $\theta=0$. Now, we show that the second diagonal element is positive to first
order in $\delta y$. Using the fact that $\theta\frac{N_{g}}{N_{v_{2}}}=\mathcal{O}(\delta y^{2})$,
the positivity condition reads
\begin{displaymath}
  \frac{p_{h_{2}}x_{h_{1}}+p_{h_{1}}x_{h_{2}}}{N_{h}^{2}}-\frac{p_{g_{3}}x_{g_{2}}+p_{g_{2}}x_{g_{3}}}{N_{v_{1}}^{2}}+\mathcal{O}(\delta y^{2})\ge0,
\end{displaymath}
which, in turn, becomes
\begin{align*}
  \frac{P_{12}^{2}y_{j-1}+P_{6}^{-1}y_{j+2}}{N_{h}^{2}}-\frac{P_{6}^{1}y_{j-2}+P_{12}^{-2}y_{j+1}}{N_{v_{1}}^{2}}+\mathcal{O}(\delta y^{2})
  =2\delta y+\mathcal{O}(\delta y^{2})\ge0.
\end{align*}
This establishes that $U$ solves the $3\to 2$ transition, for a closely spaced lattice.
Note that only the proof of validity was done perturbatively to first
order in $\delta y$. The unitary itself is known exactly, as $\theta$
can be obtained by solving the quadratic. Using $f(y)=(y_{0}'-y)(\Gamma_{1}-y)(\Gamma_{2}-y)$ we can implement
the last two moves in  \cref{fig:1by10correct} as they constitute a $3\to1$
and a $2\to1$ merge. The only
remaining task is to implement the $2\to2$ move of the last step,
because previously we assumed $\sqrt{p_{g_{2}}}\neq0$.

\subsubsection{The \texorpdfstring{$2\to2$}{2 -> 2} move and its validity}\label{subsubsec:2to2}
\atul{Using the notation from the previous section, our goal is now to construct a unitary that solves the 2 to 2 move as highlighted in  \cref{fig:2to2special}. To this end, consider the following unitary} %

\begin{displaymath}
  U=\left|w\right\rangle \left\langle v\right|+\left(\alpha\left|v\right\rangle +\beta\left|w_{1}\right\rangle \right)\left\langle v_{1}\right|+\left|v\right\rangle \left\langle w\right|+\left(\beta\left|v\right\rangle -\alpha\left|w_{1}\right\rangle \right)\left\langle w_{1}\right|
\end{displaymath}
where as before $\left|\alpha\right|^{2}+\left|\beta\right|^{2}=1$,
\begin{displaymath}
  \left|v\right\rangle =\frac{1}{N_{g}}\left(\sqrt{p_{g_{1}}}\left|g_{1}\right\rangle +\sqrt{p_{g_{2}}}\left|g_{2}\right\rangle \right),
  \left|w\right\rangle =\frac{1}{N_{h}}\left(\sqrt{p_{h_{1}}}\left|h_{1}\right\rangle +\sqrt{p_{h_{2}}}\left|h_{2}\right\rangle \right),
\end{displaymath}
\begin{displaymath}
  \left|v_{1}\right\rangle =\frac{1}{N_{g}}\left(\sqrt{p_{g_{2}}}\left|g_{1}\right\rangle -\sqrt{p_{g_{1}}}\left|g_{2}\right\rangle \right) \text{ and }
  \left|w_{1}\right\rangle =\frac{1}{N_{h}}\left(\sqrt{p_{h_{2}}}\left|h_{1}\right\rangle -\sqrt{p_{h_{1}}}\left|h_{2}\right\rangle \right).
\end{displaymath}
We evaluate the constraint equation using
\begin{footnotesize}
  \begin{align*}
    EU\left|g_{1}\right\rangle   =\frac{\sqrt{p_{g_{1}}}\left|w\right\rangle +\beta e^{-i\phi_{g}}e^{i\phi_{h}}\sqrt{p_{g_{2}}}\left|w_{1}\right\rangle }{N_{g}}\ , \
    EU\left|g_{2}\right\rangle   =\frac{\sqrt{p_{g_{2}}}\left|w\right\rangle -\beta e^{-i\phi_{g}}e^{i\phi_{h}}\sqrt{p_{g_{1}}}\left|w_{1}\right\rangle }{N_{g}},
  \end{align*}
\end{footnotesize}
and
\begin{displaymath}
  EU\left|g_{1}\right\rangle \left\langle g_{1}\right|U^{\dagger}E=\frac{1}{N_{g}^{2}}\begin{array}{c|cc}
                                & \left\langle w\right| & \left\langle w_{1}\right|                               \\
    \hline \left|w\right\rangle & p_{g_{1}}             & \beta e^{i(\phi_{h}-\phi_{g})}\sqrt{p_{g_{2}}p_{g_{1}}} \\
    \left|w_{1}\right\rangle    & \text{h.c.}           & \left|\beta\right|^{2}p_{g_{2}}
  \end{array}
\end{displaymath}
as
\begin{displaymath}
  \left[\begin{array}{cc}
      \left\langle x_{h}\right\rangle -\left\langle x_{g}\right\rangle & \frac{1}{N_{g}^{2}}\left[\sqrt{p_{h_{1}}p_{h_{2}}}(x_{h_{1}}-x_{h_{2}})-\beta\sqrt{p_{g_{1}}p_{g_{2}}}(x_{g_{1}}-x_{g_{2}})\right]  \\
      \text{h.c.}                                                      & \frac{1}{N_{g}^{2}}\left[p_{h_{2}}x_{h_{1}}+p_{h_{1}}x_{h_{2}}-\left|\beta\right|^{2}(p_{g_{2}}x_{g_{1}}+p_{g_{1}}x_{g_{2}})\right]
    \end{array}\right]\ge0,
\end{displaymath}
where we absorbed the phase freedom in $\beta$, a free parameter,
which will be fixed shortly. We use the same strategy as above and
take the first diagonal element to be zero.
We must show that
\begin{small}
  \begin{displaymath}
    \sqrt{\frac{p_{h_{1}}p_{h_{2}}}{p_{g_{1}}p_{g_{2}}}}\frac{(x_{h_{1}}-x_{h_{2}})}{(x_{g_{1}}-x_{g_{2}})}=\beta\le1,
    \text{ and }
    \frac{1}{N_{g}^{2}}\left[p_{h_{2}}x_{h_{1}}+p_{h_{1}}x_{h_{2}}-\left|\beta\right|^{2}(p_{g_{2}}x_{g_{1}}+p_{g_{1}}x_{g_{2}})\right]\ge0.
  \end{displaymath}
\end{small}For this transition $f(y_{j-2})=0$, which we use to write
\begin{small}
  \begin{displaymath}
    f(y_{j+k})=\left.\frac{\partial f}{\partial y}\right|_{y_{j-2}}(k+2)\delta y=-(k+2)\alpha\delta y, \text{ with }
    \alpha=-\left.\frac{\partial f}{\partial y}\right|_{y_{j-2}}=(\Gamma_{1}-y_{j-2})(\Gamma_{2}-y_{j-2}).
  \end{displaymath}
\end{small}
From  \cref{fig:2to2special} we have
\begin{align*}
  p_{h_{1}} & =P_{1}(x_{j})=\frac{-f(y_{j-1})}{3\cdot 2\delta y^{2}y_{j-1}}=\frac{\alpha+\mathcal{O}(\delta y)}{6\delta yy_{j}},\
  p_{h_{2}}  =P_{2}(y_{j+2})=\frac{-f(y_{j+2})}{4\cdot 3\delta y^{2}y_{j+2}}=\frac{\alpha+\mathcal{O}(\delta y)}{3\delta yy_{j}}  \\
  x_{h_{1}} & =y_{j-1},\,x_{h_{2}}=y_{j+2}
\end{align*}
while
\begin{footnotesize}
  \begin{align*}
    p_{g_{1}} & =P(x_{j})=\frac{f(0)\delta y}{y_{j+2}y_{j+1}y_{j-1}y_{j-2}}=\frac{f(0)\delta y+\mathcal{O}(\delta y^{2})}{y_{j}^{4}},\
    p_{g_{2}}  =P_{1}(y_{j+1})=\frac{-f(y_{j+1})}{3\cdot 2\delta y^{2}y_{j+1}}=\frac{\alpha+\mathcal{O}(\delta y)}{2\delta yy_{j}}     \\
    x_{g_{1}} & =0,\,x_{g_{2}}=y_{j+1}.
  \end{align*}
\end{footnotesize}\begin{figure}[H]
  \begin{centering}
    \includegraphics[scale=0.9]{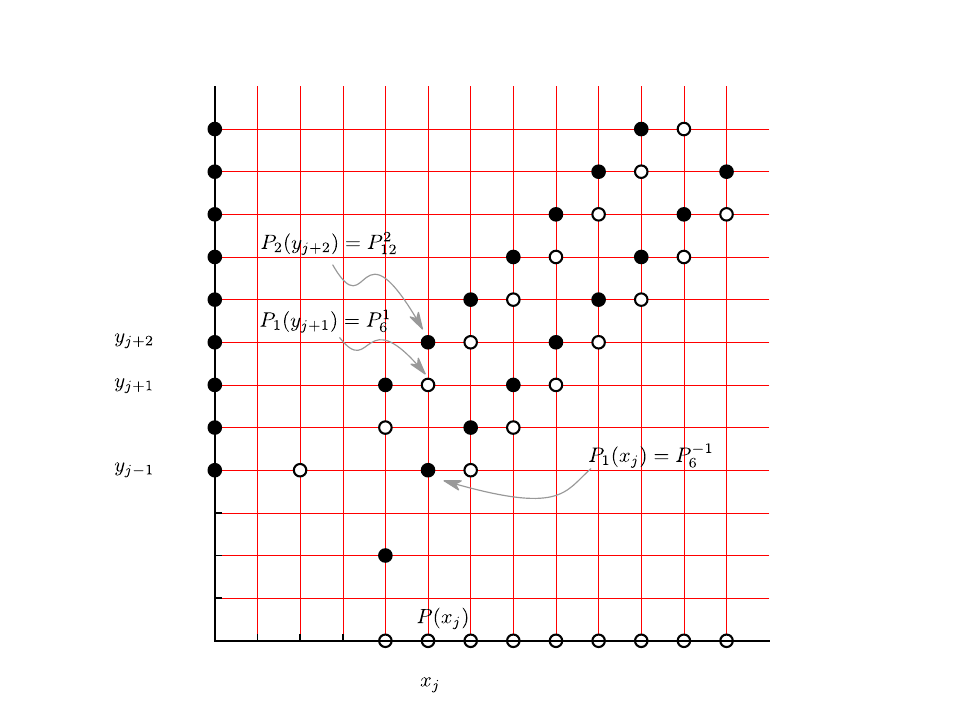}
    \par\end{centering}
  \caption{The first $2\to2$ transition\label{fig:2to2special}}
\end{figure}This entails
\begin{small}
  \begin{align*}
    \beta & =\sqrt{\frac{p_{h_{1}}p_{h_{2}}}{p_{g_{1}}p_{g_{2}}}}\frac{(x_{h_{1}}-x_{h_{2}})}{(x_{g_{1}}-x_{g_{2}})}=\sqrt{\frac{y_{0}'\alpha+\mathcal{O}(\delta y)}{f(0)}}=\sqrt{\frac{(\Gamma_{1}-y_{j-2})(\Gamma_{2}-y_{j-2})+\mathcal{O}(\delta y)}{\Gamma_{1}\Gamma_{2}}}\le1,
  \end{align*}
\end{small}where we used $f(0)=y_{0}'\Gamma_{1}\Gamma_{2}$ and the fact that $\delta y$
is small compared to $\Gamma$s. Analogously, for the second condition we have
\begin{align*}
   & \frac{1}{N_{g}^{2}}\left[p_{h_{2}}x_{h_{1}}+p_{h_{1}}x_{h_{2}}-\left|\beta\right|^{2}(p_{g_{2}}x_{g_{1}}+p_{g_{1}}x_{g_{2}})\right] \ge\frac{1}{N_{g}^{2}}\left[p_{h_{2}}x_{h_{1}}+p_{h_{1}}x_{h_{2}}-p_{g_{2}}x_{g_{1}}\right] \\
   & = \frac{1}{2\delta yN_{g}^{2}}\left[\alpha+\mathcal{O}(\delta y)\right] =\frac{1}{2\delta yN_{g}^{2}}\left[(\Gamma_{1}-y_{j-2})(\Gamma_{2}-y_{j-2})+\mathcal{O}(\delta y)\right]\ge0,
\end{align*}where the last step holds for $\delta y$ small enough. The $2\to2$ move corresponding to the leftmost (see  \cref{fig:Final-2to2})
and bottom-most set of points can be shown to satisfy the TEF constraint similarly. %
\begin{figure}[H]
  \begin{centering}
    \includegraphics[scale=0.9]{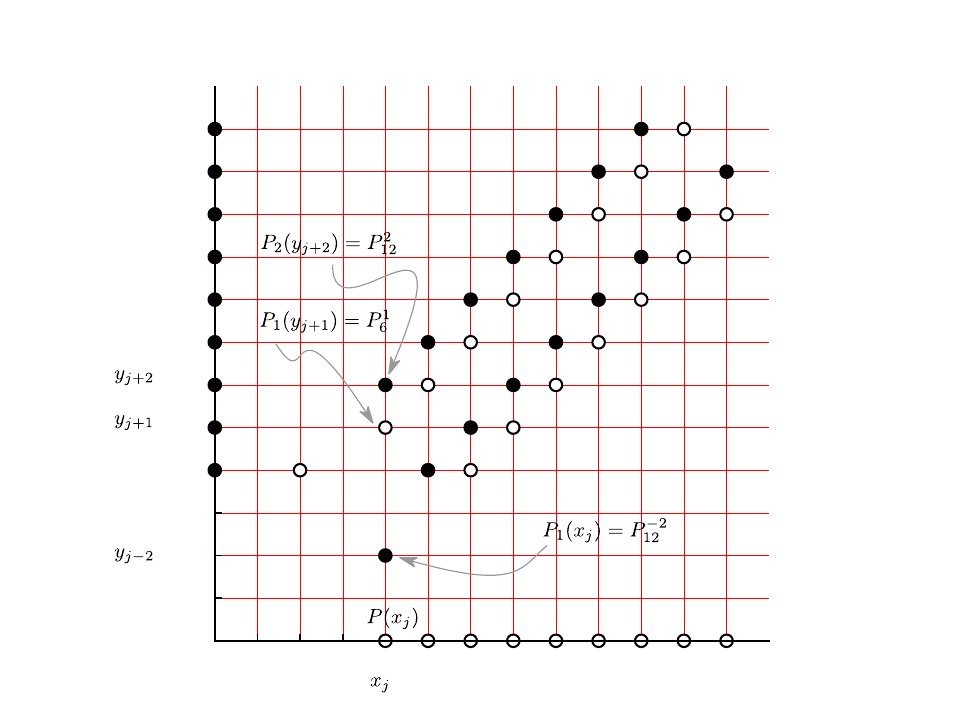}
    \par\end{centering}
  \caption{The final $2\to2$ transition. \label{fig:Final-2to2}}
\end{figure}

\section{Approaching Bias \texorpdfstring{$\epsilon(k)=1/(4k+2)$}{epsilon(k)=1/(4k+2)}} \label{sec:1by4k+2}

While we succeeded at constructing the unitaries involved in the bias
$1/10$ protocol, we did not follow any systematic procedure. Here,
we construct the unitaries corresponding to the valid functions that characterise  Mochon's  point games
(see \cref{def:f_assignment-f_0_assignment-balanced-m_kmonomial}).
These, together with the TEF, allow us to construct explicit WCF protocols with bias approaching $\epsilon(k)=1/(4k+2)$ for arbitrary integers
$k>0$.

Before we begin, we clarify the notation we use. %

\begin{itemize}
  \item For a Hermitian matrix $A$ with spectral decomposition (including zero eigenvalues) $A=\sum_{i}a_{i}\left|i\right\rangle \left\langle i\right|$, we define
        the pseudo-inverse or the generalised inverse of $A$ as $$A^{\dashv}:=\sum_{i:\left|a_{i}\right|>0}a_{i}^{-1}\left|i\right\rangle \left\langle i\right|.$$
  \item We write functions
        $t$ with finite support in the following two ways (unless otherwise stated): (1) as $t=\sum_{i=1}^{n}p_{i}\left\llbracket x_{i}\right\rrbracket $
        where we assume $p_{i}>0$ for all $i\in\{1,2\dots n\}$ and that
        $x_{i}\neq x_{j}$ for $i\neq j$ and (2) as $t=\sum_{i=1}^{n_{h}}p_{h_{i}}\left\llbracket x_{h_{i}}\right\rrbracket -\sum_{i=1}^{n_{g}}p_{g_{i}}\left\llbracket x_{g_{i}}\right\rrbracket $
        where $p_{h_{i}}$ and $p_{g_{i}}$ are strictly positive and $x_{h_{i}}$
        and $x_{g_{i}}$ are all distinct.

\end{itemize}

\subsection{The \texorpdfstring{$f-$}{f-}assignments}\label{subsec:fassignmentequivmonomial}

Even though we already described Mochon's assignment (see \cref{lem:fAssignment}) in \cref{sec:PriorArt}, we now state it formally as an $f$-assignment, to facilitate the analysis that follows.
\begin{definition}[$f$-assignments]
  Given a set of real numbers $0\le x_{1}<x_{2}\dots<x_{n}$ and a
  polynomial of degree at most $n-2$ satisfying $f(-\lambda)\ge0$
  for all $\lambda\ge0$, an \emph{$f$-assignment} is given by
  the function
  $$
    t=\sum_{i=1}^{n}\underbrace{\frac{-f(x_{i})}{\prod_{j\neq i}(x_{j}-x_{i})}}_{:=p_{i}}\left\llbracket x_{i}\right\rrbracket =h-g,
  $$
  (up to a positive multiplicative factor) where $h$ contains the positive
  part of $t$ and $g$ the negative part (without any common support),
  viz. $h=\sum_{i:p_{i}>0}p_{i}\left\llbracket x_{i}\right\rrbracket $
  and $g=\sum_{i:p_{i}<0}\left(-p_{i}\right)\left\llbracket x_{i}\right\rrbracket $.
  \label{def:f_assignment-f_0_assignment-balanced-m_kmonomial}
  \begin{itemize}
    \item When $f$ is a monomial, viz. has the form $f(x)=\text{c}x^{q}$, where $\text{c} >0$ and $q\geq 0$ we call
          the assignment a \emph{monomial assignment}. For $q=0$ we call the assignment an \emph{$f_{0}$-assignment}.
    \item We say that an assignment is \emph{balanced} if the number of points with
          negative weights, $p_{i}<0$, equals the number of points with positive
          weights, $p_{i}>0$. We say an assignment is \emph{unbalanced} if
          it is not balanced.
    \item We say that a monomial assignment is \emph{aligned} if the degree of the monomial is an even number ($q=2(b-1), b\in\mathbb{N}$). We say that a monomial assignment is \emph{misaligned} if it is not aligned.
  \end{itemize}
\end{definition}

An $f_{0}$-assignment starts with a
point that has a negative weight regardless of the total number of points and thereafter, the sign alternates. With
this as the base structure, working out the signs of the weights for
monomial assignments gets easier. %
The only mathematical
property that is needed to find an analytic
solution, turns out to be the following.
\begin{lemma}
  Fix integers $m\le n-2$ and $n\ge2$. Consider an $f$-assignment of
  the form $t=\sum_{i}\frac{-(-x_{i})^{m}}{\prod_{j\neq i}(x_{j}-x_{i})}\left\llbracket x_{i}\right\rrbracket $
  for $n$ points $0\le x_{1} < \dots <x_{n}$ and use it to implicitly
  define $p_{h_{i}}$ and $p_{g_{i}}$ as follows: $t=\sum_{i}(x_{h_{i}})^{m}p_{h_{i}}\left\llbracket x_{h_{i}}\right\rrbracket -\sum_{i}(x_{g_{i}})^{m}p_{g_{i}}\left\llbracket x_{g_{i}}\right\rrbracket $.
  Let $\left\langle x^{l}\right\rangle :=\sum_{i}(x_{h_{i}})^{l}p_{h_{i}}-\sum_{i}(x_{g_{i}})^{l}p_{g_{i}}$. Then, $\left\langle x^{l}\right\rangle =0$ for $0\le l\le n-2$.
  Further, $\left\langle x^{n-1}\right\rangle :=\sum_{i}(x_{h_{i}})^{n-1}p_{h_{i}}-\sum_{i}(x_{g_{i}})^{n-1}p_{g_{i}}=(-1)^{m+n}$
  which is strictly positive when $n+m$ is even (i.e. when
  $t$ is unbalanced misaligned and balanced aligned (see \cref{def:f_assignment-f_0_assignment-balanced-m_kmonomial})).

  \label{lem:expectationLemma}
\end{lemma}
We defer the proofs to \cref{subsec:proofexpectationlemma}.

Suppose that the $f$-assignment\footnote{While an $f$-assignment is a valid function for all polynomials
  $f$ satisfying the conditions in \cref{def:f_assignment-f_0_assignment-balanced-m_kmonomial},
  in what follows, we restrict to polynomials $f$
  with real roots. In fact,  to be consistent with \cref{def:f_assignment-f_0_assignment-balanced-m_kmonomial}, the roots
  must additionally be non-negative.} can be decomposed into a sum of valid functions, and let us call these valid functions in
the decomposition, \emph{constituents}.  Recall, from \cref{subsec:TEFfunctions}, that valid functions are the same as TEF functions---functions that can be \emph{solved} using some unitary $U$. Later, we show how to choose the decomposition such that the constituents can be \emph{solved}. We call such a solution, an $\emph{effective solution}$, and its definition is presented below.
\begin{definition}[Effectively Solving a function (extends \cref{def:TEFconstraint_})]
  Given a finitely supported function $t=\sum_{i=1}^{n_{h}}p_{h_{i}}\left\llbracket x_{h_{i}}\right\rrbracket -\sum_{i=1}^{n_{g}}p_{g_{i}}\left\llbracket x_{g_{i}}\right\rrbracket $
  and $\left\{ \left|g_{1}\right\rangle ,\left|g_{2}\right\rangle \dots\left|g_{n_{g}}\right\rangle ,\left|h_{1}\right\rangle ,\left|h_{2}\right\rangle \dots\left|h_{n_{h}}\right\rangle \right\} $ an orthonormal basis,
  we say that a unitary matrix $O$ \emph{solves}\footnote{\atul{This notion of ``solving'' is essentially identical to that in \cref{def:TEFconstraint_} except that it was stated for transitions. We restate it for functions here for convenience.}} $t$ if $O$ satisfies
  the following: $O\left|v\right\rangle =\left|w\right\rangle $ and
  $X_{h}\ge E_{h}OX_{g}O^{T}E_{h}$ where $\left|v\right\rangle =\sum_{i=1}^{n_{g}}\sqrt{p_{g_{i}}}\left|g_{i}\right\rangle $,
  $\left|w\right\rangle =\sum_{i=1}^{n_{h}}\sqrt{p_{h_{i}}}\left|h_{i}\right\rangle $,
  $X_{h}=\sum_{i=1}^{n_{h}}x_{h_{i}}\left|h_{i}\right\rangle \left\langle h_{i}\right|$,
  $X_{g}=\sum_{i=1}^{n_{g}}x_{g_{i}}\left|g_{i}\right\rangle \left\langle g_{i}\right|$
  and the projector $E_{h}=\sum_{i=1}^{n_{h}}\left|h_{i}\right\rangle \left\langle h_{i}\right|$. Moreover, we say that $t$ has an \emph{effective solution} if $t=\sum_{i\in I}t'_i$ and $t'_i$ has a solution for all $i\in I$, where $I$ is a finite set.
  \label{def:solvingassignment}
\end{definition}

Before constructing these effective solutions, we briefly justify a claim we made in \cref{subsec:contralgebraic}: to implement a valid function (and in particular, an $f$-assignment), it suffices to implement the constituent  functions. The difficulty is that the constituent functions
might be negative at various locations, where there are
no points present. A similar difficulty was encountered while transforming a TIPG into a TDPG, and it was
handled using \emph{catalyst states} %
(as in \cite{Mochon07,ACG+14}). We outlined this procedure in \cref{subsec:tipg} after \cref{thm:TIPG-to-valid-point-games}.
For the $f$-assignment of the TIPG,  one can again use such a procedure: create the catalyst state, apply a scaled down version of the constituent functions, repeat until the $f$-function has been nearly implemented, and finally absorb the catalyst state with a vanishing increase in the final point. This results in a TDPG that uses only constituent functions.
The unitary matrices for the constituent functions are, thus, sufficient to get a TDPG with the same bias as for the $f$-assignment.  This motivates \cref{def:solvingassignment} below. We can then apply the TEF from \cref{sec:TEF} to the TDPG and obtain a WCF protocol
approaching the same bias as the TIPG that we started
with, in the limit of infinite rounds of communication. %

Returning to the construction of effective solutions, we first give a decomposition of an $f$-assignment into a sum of monomial assignments (for another possible decomposition, see \cref{subsec:Restricted-decomposition-into-monomials})

\begin{lemma}[$f$-assignment as a sum of monomials]
  Consider a set of real coordinates satisfying $0\le x_{1}<x_{2}\dots<x_{n}$
  and let $f(x)=(r_{1}-x)(r_{2}-x)\dots(r_{k}-x)$ where $k\le n-2$ and $r_i>0$.
  Let $t=\sum_{i=1}^{n}p_{i}\left\llbracket x_{i}\right\rrbracket $
  be the corresponding $f$-assignment. Then, there exists $\alpha_{l}\ge0$, such that
  $$
    t=\sum_{l=0}^{k}\alpha_{l}\left(\sum_{i=1}^{n}\frac{-(-x_{i})^{l}}{\prod_{j\neq i}(x_{j}-x_{i})}\left\llbracket x_{i}\right\rrbracket \right).
  $$

  \label{lem:generalMonomialDecomposition}
\end{lemma}

In the following sections, we construct solutions to monomial assignments. The analysis there uses matrix inverses and having a coordinate equal to zero breaks the argument. \atul{Mochon's TIPGs, however, do have support on the axes. To avoid these concerns, we translate (any of) Mochon's TIPG horizontally and vertically by an arbitrarily small amount $\dorigin>0$. Clearly, the corresponding final point (which encodes the bias) also only increases by the same arbitraryl small amount, $\dorigin$. We make two additional elementary observations: (i) the translated point game is also a TIPG as the moves involved in the resulting point game stay valid even after translation (see \cref{lem:validAfterTranslating} below, with $c=\dorigin$), and (ii) Mochon's $f$ assignment on $\{x_1,\dots x_n\}$ after translation by $\dorigin$ can be seen as an $f'$ assignment on the translated coordinates $\{x_1 +\dorigin,\dots ,x_n+\dorigin \}$ with $f'(x + \dorigin) = f(x)$ (see \cref{lem:fAssignmentTranslation} below, with $c=\dorigin$). %
  Together with the discussion preceding \cref{lem:generalMonomialDecomposition}, we conclude that one can construct protocols corresponding to the translated point game, granted one can construct effective solutions to $f'$ assignments on positive coordinates where $f'$ is a polynomial with non-negative roots.}

\begin{lemma}\label{lem:validAfterTranslating}
  \atul{Suppose $t=\sum_{i=1}^{n_{h}}p_{h_{i}} \llbracket x_{h_{i}} \rrbracket - \sum_{i=1}^{n_g}p_{g_{i}} \llbracket x_{g_{i}} \rrbracket$ is solved by $O$. Then, for all $c \ge 0$, $t=\sum_{i=1}^{n_{h}}p_{h_{i}} \llbracket x_{h_{i}} + c \rrbracket - \sum_{i=1}^{n_g}p_{g_{i}} \llbracket x_{g_{i}} + c \rrbracket$ is also solved by $O$.}
\end{lemma}
\atul{
  \begin{proof}[Proof sketch]
    Define $X_{h}:=\sum_{i=1}^{n_{h}}x_{h_{i}}\left|h_{i}\right\rangle $,
    $X_{g}:=\sum_{i=1}^{n_{g}}x_{g_{i}}\left|g_{i}\right\rangle $. If
    $t$ is solved by $O$ then we must have $X_{h}\ge E_{h}OX_{g}O^{T}E_{h}$.
    We show this implies $X_{h}+c\mathbb{I}_{h}\ge E_{h}O(X_{g}+c\mathbb{I}_{g})O^{T}E_{h}$,
    where $\mathbb{I}_{h}:=\sum_{i=1}^{n_{h}}\left|h_{i}\right\rangle \left\langle h_{i}\right|$
    and $\mathbb{I}_{g}:=\sum_{i=1}^{n_{g}}\left|g_{i}\right\rangle \left\langle g_{i}\right|$ which in turn means that $t'$ is also solved by $O$. %
    To this end, observe that
    \begin{small}
      \begin{align*}
         & X_{h}  \ge E_{h}OX_{g}O^{T}E_{h} \iff E_{h}(X_{h}-OX_{g}O^{T})E_{h}  \ge0 \qquad\qquad\qquad\qquad \because X_{h}=E_{h}X_{h}E_{h} \\
         & \iff E_{h}(X_{h}+c\mathbb{I}_{hg}-O(X_{g}-c\mathbb{I}_{hg})O^{T})E_{h}  \ge0
        \iff X_{h}+c\mathbb{I}_{h}  \ge E_{h}O(X_{g}+c\mathbb{I}_{hg})O^{T}E_{h},
      \end{align*}
    \end{small}
    where $\mathbb{I}_{hg}:=\mathbb{I}$. Further,
    \begin{align*}
      X_{g}+c\mathbb{I}_{hg}  \ge X_{g}+c\mathbb{I}_{g}
      \implies E_{h}O(X_{g}+c\mathbb{I}_{hg})O^{T}E_{h}  \ge E_{h}O(X_{g}+c\mathbb{I}_{g})O^{T}E_{h}
    \end{align*}
    which together yield
    \begin{displaymath}
      X_{h}\ge E_{h}OX_{g}O^{T}E_{h} \implies X_{h}+c\mathbb{I}_{h}\ge E_{h}O(X_{g}+c\mathbb{I}_{g})O^{T}E_{h}.
    \end{displaymath}
    as asserted.
  \end{proof}
}

\begin{lemma}\label{lem:fAssignmentTranslation}
  \atul{Consider a set of real
  coordinates satisfying $0<x_{1}+c<x_{2}+c\dots<x_{n}+c$ where $c>0$
  and let $f'(x)=(a_{1}+c-x)(a_{2}+c-x)\dots(a_{k}+c-x)$. Let $t'=\sum_{i=1}^{n}p_{i}'\left\llbracket x_{i}'\right\rrbracket $
  be the corresponding $f$-assignment (see \cref{def:f_assignment-f_0_assignment-balanced-m_kmonomial}) with $x'_{i}:=x_{i}+c$.
  Then $p_{i}=p'_{i}$.} %
\end{lemma}
\atul{
  \begin{proof}
    Note that $f'(x_i+c)=f(x_i)$ and $\prod_{j\neq i} (x_j -x_i) = \prod_{j\neq i} (x'_j -x_i)$. Using this and the definition of $f$-assignment (see \cref{def:f_assignment-f_0_assignment-balanced-m_kmonomial}), it follows that $p'_i=p_i$.
  \end{proof}
}

\atul{Henceforth, we restrict ourselves to $f$-assignments defined on positive coordinates.} Having decomposed the $f$-assignment into a sum of monomial assignments, we now give a solution to monomial assignments. We start with $f_0$-assignments (monomial assignment where the monomial is a  constant) to convey the key idea behind the construction and subsequently build on this idea to solve the four types of monomial assignments. %

\subsection{Solution to the \texorpdfstring{$f_0$}{f 0}-assignment}
\label{subsec:f0algebraic}
Let us solve the $f_0$-assignment. We first look at
the balanced case, where the number of points involved, $2n$, is
even. This corresponds to an $n\to n$ transition, i.e. a transition from $n$ initial points to $n$ final points.

\subsubsection{The balanced case}
\label{subsub:balancedf0algebraic}
\begin{proposition}[Solution to balanced $f_{0}$-assignments]
  \label{prop:balancedf0algeb}
  Let
  \begin{itemize}
    \item $
            t=\sum_{i=1}^{n}p_{h_{i}}\left\llbracket x_{h_{i}}\right\rrbracket -\sum_{i=1}^{n}p_{g_{i}}\left\llbracket x_{g_{i}}\right\rrbracket
          $ be an $f_{0}$-assignment over $\{x_{1},x_{2}\dots x_{2n}\}$

    \item $\left\{ \left|h_{1}\right\rangle ,\left|h_{2}\right\rangle \dots\left|h_{n}\right\rangle ,\left|g_{1}\right\rangle ,\left|g_{2}\right\rangle \dots\left|g_{n}\right\rangle \right\} $
          be an orthonormal basis, and
    \item finally
          \begin{footnotesize}
            \begin{displaymath}
              X_{h}:=\sum_{i=1}^{n}x_{h_{i}}\left|h_{i}\right\rangle \left\langle h_{i}\right|\doteq\diag(x_{h_{1}},\dots x_{h_{n}},\underbrace{0,\dots0}_{n\text{-zeros}}),
              X_{g}:=\sum_{i=1}^{n}x_{g_{i}}\left|g_{i}\right\rangle \left\langle g_{i}\right|\doteq\diag(\underbrace{0,\dots0}_{n\text{-zeros}},x_{g_{1}},\dots x_{g_{n}}),
            \end{displaymath}
            \begin{displaymath}
              \left|w\right\rangle :=\sum_{i=1}^{n}\sqrt{p_{h_{i}}}\left|h_{i}\right\rangle \doteq(\sqrt{p_{h_{1}}},\dots\sqrt{p_{h_{n}}},\underbrace{0,\dots0}_{n\text{-zeros}})^{T},
              \left|v\right\rangle :=\sum_{i=1}^{n}\sqrt{p_{g_{i}}}\left|g_{i}\right\rangle \doteq(\underbrace{0,\dots0}_{n\text{-zeros}},\sqrt{p_{g_{1}}},\dots\sqrt{p_{g_{n}}})^{T}.
            \end{displaymath}
          \end{footnotesize}

  \end{itemize}
  Then,
  \begin{displaymath}
    O:=\sum_{i=0}^{n-1}\left(\frac{\Pi_{h_{i-1}}^{\perp}(X_{h})^{i}\left|w\right\rangle \left\langle v\right|(X_{g})^{i}\Pi_{g_{i-1}}^{\perp}}{\sqrt{c_{h_{i}}c_{g_{i}}}}+\hc\right)
  \end{displaymath}
  satisfies
  $
    X_{h}\ge E_{h}OX_{g}O^{T}E_{h}\text{ and }O\left|v\right\rangle =\left|w\right\rangle
  $, where $E_{h}:=\sum_{i=1}^{n}\left|h_{i}\right\rangle \left\langle h_{i}\right|$,
  $\Pi_{h_{-1}}^{\perp}=\Pi_{g_{-1}}^{\perp}=\mathbb{I}$,
  \begin{footnotesize}
    \begin{displaymath}\Pi_{h_{i}}^{\perp}:=\text{projector orthogonal to }\text{span}\{(X_{h})^{i}\left|w\right\rangle ,(X_{h})^{i-1}\left|w\right\rangle ,\dots\left|w\right\rangle \},
      c_{h_{i}}:=\left\langle w\right|(X_{h})^{i}\Pi_{h_{i-1}}^{\perp}(X_{h})^{i}\left|w\right\rangle,\end{displaymath}
  \end{footnotesize}
  and analogously
  \begin{footnotesize}
    \begin{displaymath}
      \Pi_{g_{i}}^{\perp}:=\text{projector orthogonal to }\text{span}\{(X_{g})^{i}\left|v\right\rangle ,(X_{g})^{i-1}\left|v\right\rangle ,\dots\left|v\right\rangle \},
      c_{g_{i}}:=\left\langle v\right|(X_{g})^{i}\Pi_{g_{i-1}}^{\perp}(X_{g})^{i}\left|v\right\rangle. \end{displaymath}
  \end{footnotesize}

\end{proposition}

\begin{proof}
  Using \cref{lem:expectationLemma} for $2n$ points, we get
  \begin{equation}
    \left\langle x^{k}\right\rangle =0\ \ \  \text{ for }\ \ \ \  k\in\{0,1,2\dots,2n-2\},\label{eq:Mochonsf0equality}
  \end{equation}
  and
  \begin{equation}
    \left\langle x^{2n-1}\right\rangle >0.\label{eq:Mochonsf0Positivity}
  \end{equation}
  We define the basis of interest here, essentially using the Gram-Schmidt
  method. Let
  \begin{align}
    \left|w_{0}\right\rangle & :=\left|w\right\rangle \nonumber                                                                                                                                     \\
    \left|w_{1}\right\rangle & :=\frac{\left(\mathbb{I}-\left|w_{0}\right\rangle \left\langle w_{0}\right|\right)(X_{h})\left|w\right\rangle }{\sqrt{c_{h_{1}}}}\nonumber                           \\
                             & \vdots\nonumber                                                                                                                                                      \\
    \left|w_{k}\right\rangle & :=\frac{\left(\mathbb{I}-\sum_{i=0}^{k-1}\left|w_{i}\right\rangle \left\langle w_{i}\right|\right)(X_{h})^{k}\left|w\right\rangle }{\sqrt{c_{h_{k}}}}.\label{eq:w_k}
  \end{align}
  We indicate the term with the highest power of $X_{h}$ appearing
  in $\left|w_{k}\right\rangle $ by
  \begin{displaymath}
    \mathcal{M}(\left|w_{k}\right\rangle )=\left\langle x_{h}^{2k}\right\rangle \cdot(X_{h})^{k}\left|w\right\rangle
  \end{displaymath}
  where the scalar factor represents the dependence on the
  highest power of $x_{h}$ (appearing as $\left\langle x_{h}^{l}\right\rangle $)
  in $\left|w_{k}\right\rangle $. For instance, here the $\left\langle x_{h}^{2k}\right\rangle $
  factor comes from $\sqrt{c_{h_{k}}}$. Note that the projectors can
  be expressed in terms of these vectors more concisely,
  \begin{displaymath}
    \Pi_{h_{i}}:=\mathbb{I}-\Pi_{h_{i}}^{\perp}=\sum_{j=0}^{i}\left|w_{j}\right\rangle \left\langle w_{j}\right|.
  \end{displaymath}
  It also follows that $O$ can be re-written as
  $
    O=\sum_{j=0}^{n-1}\left(\left|w_{j}\right\rangle \left\langle v_{j}\right|+\left|v_{j}\right\rangle \left\langle w_{j}\right|\right),
  $
  where $\left|v_{j}\right\rangle $ is analogously defined. It is evident that $O\left|v\right\rangle =\left|w\right\rangle $.
  Let $D=X_{h}-E_{h}OX_{g}O^{T}E_{h}$ and note that $\left\langle v_{j}\right|D\left|v_{i}\right\rangle =0$
  (because $X_{h}\left|v_{i}\right\rangle =0$ and $E_{h}\left|v_{i}\right\rangle =0$\footnote{The conclusion holds even without the projector as $O$ maps $\text{span}(\left|v_{1}\right\rangle ,\left|v_{2}\right\rangle ,\dots\left|v_{n}\right\rangle )$
    to $\text{span}(\left|w_{1}\right\rangle ,\left|w_{2}\right\rangle \dots\left|w_{n}\right\rangle )$
    on which $X_{g}$ has no support.}). We assert that it has the following rank-1 form
  \begin{displaymath}
    D=\left[\begin{array}{cccc}
        0      & \dots  & 0                                                      \\
        \vdots & \ddots & \vdots                                                 \\
        0      & \dots  & \left\langle w_{n-1}\right|D\left|w_{n-1}\right\rangle
      \end{array}\right]
  \end{displaymath}
  in the $\left(\left|w_{0}\right\rangle ,\left|w_{1}\right\rangle ,\dots\left|w_{n-1}\right\rangle \right)$
  basis, together with $\left\langle w_{n-1}\right|D\left|w_{n-1}\right\rangle >0$.
  To see this, we simply compute
  \begin{align*}
    \left\langle w_{i}\right|D\left|w_{j}\right\rangle   =\left\langle w_{i}\right|X_{h}\left|w_{j}\right\rangle -\left\langle w_{i}\right|OX_{g}O^{T}\left|w_{j}\right\rangle  =\left\langle w_{i}\right|X_{h}\left|w_{j}\right\rangle -\left\langle v_{i}\right|X_{g}\left|v_{j}\right\rangle .
  \end{align*}
  For $(i,j)$ for any $0\le i,j\le n-1$ except for the case where
  both $i=j=n-1$, the two terms are the same. This is because the term
  with the highest possible power $l$ (of $\left\langle x^{l}\right\rangle $)
  in $\left\langle w_{i}\right|X_{h}\left|w_{j}\right\rangle $ can
  be deduced by observing
  \begin{equation}
    \mathcal{M}(\left\langle w_{i}\right|)X_{h}\mathcal{M}(\left|w_{j}\right\rangle )=\left\langle x_{h}^{2i}\right\rangle \cdot\left\langle x_{h}^{2j}\right\rangle \cdot\left\langle x_{h}^{i+j+1}\right\rangle .\label{eq:highestPowerInH}
  \end{equation}
  For the analogous expression with $g$s to be the same, we must have
  $2i,2j$ and $i+j+1\leq2n-2$, using Equation~\eqref{eq:Mochonsf0equality}. The first two conditions
  are always satisfied (for $0\le i,j\le n-1$). The last can only be
  violated when $i=j=n-1$. This establishes that the matrix has the
  asserted form.	To prove the positivity of $\left\langle w_{n-1}\right|D\left|w_{n-1}\right\rangle $,
  consider $\left\langle w_{n-1}\right|X_{h}\left|w_{n-1}\right\rangle $
  and $\left\langle v_{n-1}\right|X_{g}\left|v_{n-1}\right\rangle $.
  When these terms are expanded in powers of $\left\langle x_{h}^{k}\right\rangle $
  and $\left\langle x_{g}^{k}\right\rangle $ respectively, only terms
  with $k>2n-2$ would remain; the others would get canceled due to
  Equation~\eqref{eq:Mochonsf0equality}. From Equation~\eqref{eq:w_k} it follows that
  \begin{displaymath}
    \left\langle w_{n-1}\right|D\left|w_{n-1}\right\rangle =\frac{1}{c_{h_{n-1}}}\left\langle w\right|(X_{h})^{2n-2+1}\left|w\right\rangle -\frac{1}{c_{g_{n-1}}}\left\langle v\right|(X_{g})^{2n-2+1}\left|v\right\rangle
  \end{displaymath}
  and it is not hard to see that $c_{h_{n-1}}=c_{h_{n-1}}(\left\langle x_{h}^{2n-2}\right\rangle ,\left\langle x_{h}^{2n-3}\right\rangle ,\dots,\left\langle x_{h}^{1}\right\rangle )$
  does not depend on $\left\langle x_{h}^{2n-1}\right\rangle $ (and analogously for $c_{g_{n-1}}$). Also, $c_{h_{n-1}}=c_{g_{n-1}}=:c_{n-1}$.
  We thus have
  \begin{displaymath}
    \left\langle w_{n-1}\right|D\left|w_{n-1}\right\rangle =\frac{\left\langle x_{h}^{2n-1}\right\rangle }{c_{n-1}}>0
  \end{displaymath}
  using Equation~\eqref{eq:Mochonsf0Positivity}. Hence, $X_{h}-E_{h}OX_{g}O^{T}E_{h}\ge0$.\\
  In the above, we assumed $\text{span}\{\left|w\right\rangle ,X_{h}\left|w\right\rangle ,X_{h}^{2}\left|w\right\rangle ,\dots,X_{h}^{n}\left|w\right\rangle \}$
  equals $\text{span}\{\left|h_{1}\right\rangle ,\left|h_{2}\right\rangle \dots\left|h_{n}\right\rangle \}$
  which is justified by \cref{lem:spanningLemma}.
\end{proof}

\subsubsection{The unbalanced case}
\label{subsubsec:unbalancedf0algebraic}

We now consider unbalanced $f_{0}$-assignments. We start by reviewing the result we just proved from a slightly different perspective. This helps us see where the previous analysis fails, when applied in the present case.  We write $D_{ij}=\left\langle w_{i}\right|D\left|w_{j}\right\rangle $,
and note that the maximum power, $l$, which appears as $\left\langle x_{g/h}^{l}\right\rangle $
is given by $\max\{2i,2j,i+j+1\}$. This yields a matrix with each
term depending on the power as
\begin{displaymath}
  D=\left[\begin{array}{cccccc}
      D_{00}(\left\langle x\right\rangle )                                                                                                                                       \\
      D_{10}(\left\langle x^{2}\right\rangle ,\dots) & D_{11}(\left\langle x^{3}\right\rangle ,\dots) &                                                &  & \text{h.c.}          \\
      D_{20}(\left\langle x^{4}\right\rangle ,\dots) & D_{21}(\left\langle x^{4}\right\rangle ,\dots) & D_{22}(\left\langle x^{5}\right\rangle ,\dots)                           \\
                                                     &                                                &                                                &  &             & \ddots
    \end{array}\right].
\end{displaymath}
We represent this dependence as
\begin{displaymath}
  \mathcal{M}(D)=\left[\begin{array}{cccc}
      \left\langle x\right\rangle                                                                                  \\
      \left\langle x^{2}\right\rangle & \left\langle x^{3}\right\rangle                                            \\
      \left\langle x^{4}\right\rangle & \left\langle x^{4}\right\rangle & \left\langle x^{5}\right\rangle          \\
                                      &                                 &                                 & \ddots
    \end{array}\right].
\end{displaymath}
For concreteness, consider the balanced $f_{0}$-case over $\{x_{1},x_{2},x_{3},x_{4}\}$,
where $\left\langle x\right\rangle =\left\langle x^{2}\right\rangle =0$
and $\left\langle x^{3}\right\rangle >0$. For this two-dimensional case, we have
\begin{displaymath}
  \mathcal{M}(D)=\left[\begin{array}{cc}
      0 & 0                               \\
      0 & \left\langle x^{3}\right\rangle
    \end{array}\right]\ge0.
\end{displaymath}
Using the same method for an $f_{0}$-assignment
over $\left\{ x_{1},x_{2}\dots x_{5}\right\} $, we have $\left\langle x\right\rangle =\left\langle x^{2}\right\rangle =\left\langle x^{3}\right\rangle =0$
and $\left\langle x^{4}\right\rangle >0$, and trying to solve in three
dimensions, we would obtain
\begin{equation}
  \mathcal{M}(D)=\left[\begin{array}{ccc}
      0                               & 0                               & \left\langle x^{4}\right\rangle \\
      0                               & 0                               & \left\langle x^{4}\right\rangle \\
      \left\langle x^{4}\right\rangle & \left\langle x^{4}\right\rangle & \left\langle x^{5}\right\rangle
    \end{array}\right]\label{eq:exampleSubMatrix}
\end{equation}
which does not seem to work directly. It turns out that the projector
appearing in the TEF constraint, removes the troublesome
part and yields a zero matrix. This unbalanced assignment takes three points to
two points. We define $X_{h}:=\diag(x_{h_{1}},x_{h_{2}},0,0,0)$,
$\left|w\right\rangle =(\sqrt{p_{h_{1}}},\sqrt{p_{h_{2}}},0,0,0)$ along with $\left|w_{0}\right\rangle   :=\left|w\right\rangle$  and
$\left|w_{1}\right\rangle   :=\left(\mathbb{I}-\left|w_{0}\right\rangle \left\langle w_{0}\right|\right)X_{h}\left|w_{0}\right\rangle$. We can write $E_{h}=\sum_{i=0}^{1}\left|w_{i}\right\rangle \left\langle w_{i}\right|$
and have the same unitary as before, except that now $\left|v_{2}\right\rangle $ is left
unchanged, i.e. $O=\sum_{i=0}^{1}\left|w_{i}\right\rangle \left\langle v_{i}\right|+\left|v_{2}\right\rangle \left\langle v_{2}\right|$.
We can show that $D'=X_{h}-E_{h}OX_{g}O^{T}E_{h}\ge0$ because
every vector $\ket{\psi}$ in the subspace $\text{span}\{\left|v_{0}\right\rangle ,\left|v_{1}\right\rangle ,\left|v_{2}\right\rangle \}$
satisfies $D'\left|\psi\right\rangle =0$ (as $X_{h}\left|\psi\right\rangle =0$
and $E_{h}\left|\psi\right\rangle =0$). This entails that it suffices
to restrict to a $2\times2$ matrix in $\text{span}\{\left|w_{0}\right\rangle ,\left|w_{1}\right\rangle \}$.
From \ref{eq:exampleSubMatrix} this is zero,
hence $D'=0$. By generalizing this example, we can obtain the solution for an unbalanced $f_{0}$-assignment, as presented in the following Proposition:
\begin{proposition}[Solution to unbalanced $f_{0}$-assignments]\label{prop:unbalancedf0algeb}
  Let
  \begin{itemize}
    \item $
            t=\sum_{i=1}^{n-1}p_{h_{i}}\left\llbracket x_{h_{i}}\right\rrbracket -\sum_{i=1}^{n}p_{g_{i}}\left\llbracket x_{g_{i}}\right\rrbracket ,
          $ be an $f_{0}$-assignment over $0<x_{1}<x_{2}\dots<x_{2n-1}$
    \item $\left\{ \left|h_{1}\right\rangle ,\left|h_{2}\right\rangle \dots\left|h_{n-1}\right\rangle ,\left|g_{1}\right\rangle ,\left|g_{2}\right\rangle \dots\left|g_{n}\right\rangle \right\} $
          be an orthonormal basis, and
    \item finally
          \begin{footnotesize}
            \begin{displaymath}
              X_{h}:=\sum_{i=1}^{n-1}x_{h_{i}}\left|h_{i}\right\rangle \left\langle h_{i}\right|\doteq\diag(x_{h_{1}},\dots x_{h_{n-1}},\underbrace{0,\dots0}_{n\text{ zeros}}),
              X_{g}:=\sum_{i=1}^{n}x_{g_{i}}\left|g_{i}\right\rangle \left\langle g_{i}\right|\doteq\diag(\underbrace{0,\dots0}_{n-1\text{ zeros}},x_{g_{1}},\dots,x_{g_{n}}),
            \end{displaymath}
            \begin{displaymath}
              \left|w\right\rangle :=\sum_{i=1}^{n-1}\sqrt{p_{h_{i}}}\left|h_{i}\right\rangle \doteq(\sqrt{p_{h_{1}}},\dots\sqrt{p_{h_{n-1}}},\underbrace{0,\dots0}_{n\text{ zeros}})^{T},
              \left|v\right\rangle :=\sum_{i=1}^{n}\sqrt{p_{g_{i}}}\left|g_{i}\right\rangle \doteq(\underbrace{0,\dots0}_{n-1\text{ zeros}},\sqrt{p_{g_{1}}},\dots\sqrt{p_{g_{n}}})^{T}
            \end{displaymath}
          \end{footnotesize}

    \item and $E_{h}:=\sum_{i=1}^{n-1}\left|h_{i}\right\rangle \left\langle h_{i}\right|$.
  \end{itemize}
  Then,
  \begin{small}
    \begin{displaymath}
      O:=\left(\sum_{i=0}^{n-2}\frac{\Pi_{h_{i-1}}^{\perp}(X_{h})^{i}\left|w\right\rangle \left\langle v\right|(X_{g})^{i}\Pi_{g_{i-1}}^{\perp}}{\sqrt{c_{h_{i}}c_{g_{i}}}}+\hc\right)+\frac{\Pi_{g_{n-2}}^{\perp}(X_{g})^{n-1}\left|v\right\rangle \left\langle v\right|(X_{g})^{n-1}\Pi_{g_{n-2}}^{\perp}}{c_{g_{i}}}
    \end{displaymath}
  \end{small}	satisfies $X_{h}\ge E_{h}OX_{g}O^{T}E_{h}$ and $E_{h}O\left|v\right\rangle =\left|w\right\rangle$,
  where $\Pi_{h_{-1}}^{\perp}=\Pi_{g_{-1}}^{\perp}=\mathbb{I}$,
  \begin{small}
    \begin{displaymath}
      \Pi_{h_{i}}^{\perp}:=\text{projector orthogonal to }\text{span}\{(X_{h})^{i}\left|w\right\rangle ,(X_{h})^{i-1}\left|w\right\rangle ,\dots\left|w\right\rangle \},
      c_{h_{i}}:=\left\langle w\right|(X_{h})^{i}\Pi_{h_{i-1}}^{\perp}(X_{h})^{i}\left|w\right\rangle,\end{displaymath}
  \end{small} and analogously
  \begin{small}
    \begin{displaymath}
      \Pi_{g_{i}}^{\perp}:=\text{projector orthogonal to }\text{span}\{(X_{g})^{i}\left|v\right\rangle ,(X_{g})^{i-1}\left|v\right\rangle ,\dots\left|v\right\rangle \},
      c_{g_{i}}:=\left\langle v\right|(X_{g})^{i}\Pi_{g_{i-1}}^{\perp}(X_{g})^{i}\left|v\right\rangle .\end{displaymath}
  \end{small}

\end{proposition}

\begin{proof}
  In this case, we use \cref{lem:expectationLemma} for $2n-1$ points. We have
  \begin{equation}
    \left\langle x^{k}\right\rangle =0\label{eq:mochonf0unbalanced}
  \end{equation}
  but this time, $k\in\{0,1,\dots2n-3\}$ and
  $
    \left\langle x^{2n-2}\right\rangle >0.
  $
  We define the basis similarly by setting $\left|w_{0}\right\rangle :=\left|w\right\rangle $
  and for all $k\in\mathbb{Z}$ satisfying $0\le k\le n-2$ we have
  \begin{displaymath}
    \left|w_{k}\right\rangle :=\frac{\Pi_{h_{k-1}}^{\perp}(X_{h})^{k}\left|w\right\rangle }{\sqrt{c_{h_{k}}}}=\frac{\left(\mathbb{I}-\sum_{i=0}^{k-1}\left|w_{i}\right\rangle \left\langle w_{i}\right|\right)(X_{h})^{k}\left|w\right\rangle }{\sqrt{c_{h_{k}}}}.
  \end{displaymath}
  We also define $\left|v_{0}\right\rangle :=\left|v\right\rangle $ and
  for all $k\in\mathbb{Z}$ satisfying  $0\le k\le n-1$ we have
  \begin{displaymath}
    \left|v_{k}\right\rangle :=\frac{\Pi_{g_{k-1}}^{\perp}(X_{g})^{k}\left|v\right\rangle }{\sqrt{c_{g_{k}}}}=\frac{\left(\mathbb{I}-\sum_{i=0}^{k-1}\left|v_{i}\right\rangle \left\langle v_{i}\right|\right)(X_{g})^{k}\left|v\right\rangle }{\sqrt{c_{h_{k}}}}.
  \end{displaymath}
  This means that $O=\sum_{i=0}^{n-2}\left(\left|w_{i}\right\rangle \left\langle v_{i}\right|+\left|v_{i}\right\rangle \left\langle w_{i}\right|\right)+\left|v_{n}\right\rangle \left\langle v_{n}\right|$
  and so $E_{h}O\left|v\right\rangle =\left|w\right\rangle $ follows
  directly. To establish $D:=X_{h}-E_{h}OX_{g}O^{T}E_{h}\ge0$,
  it suffices to show $\left\langle w_{i}\right|D\left|w_{j}\right\rangle \ge0$ for
  $i,j\in\mathbb{Z}$ satisfying $0\le i,j\le n-2$. Just as in the balanced case, this is because $D\left|v_{i}\right\rangle =0$, as
  $X_{h}\left|v_{i}\right\rangle =0$ and $E_{h}\left|v_{i}\right\rangle =0$.
  As before, we denote the highest-power term of $X_{h}$
  appearing in $\left|w_{k}\right\rangle $, for $k$ in $\{0,1\dots n-2\}$,
  by
  \begin{displaymath}
    \mathcal{M}(\left|w_{k}\right\rangle )=\left\langle x_{h}^{2k}\right\rangle \cdot(X_{h})^{k}\left|w\right\rangle
  \end{displaymath}
  and analogously, the highest power of $X_{g}$ appearing in $\left|v_{k}\right\rangle $
  for $k$ in $\{0,1,\dots n-2\}$, by
  \begin{displaymath}
    \mathcal{M}(\left|v_{k}\right\rangle )=\left\langle x_{g}^{2k}\right\rangle \cdot(X_{g})^{k}\left|v\right\rangle .
  \end{displaymath}
  Again, the highest power $l$ of $\left\langle x^{l}\right\rangle $
  in $\left\langle w_{i}\right|D\left|w_{j}\right\rangle $
  is $\max\{2j,2i,i+j+1\}$ which can be deduced by evaluating
  \begin{displaymath}
    \mathcal{M}(\left\langle w_{i}\right|)X_{h}\mathcal{M}(\left|w_{j}\right\rangle )=\left\langle x_{h}^{2j}\right\rangle \cdot\left\langle x_{h}^{2i}\right\rangle \cdot\left\langle x_{h}^{i+j+1}\right\rangle, \text{ and similarly }
  \end{displaymath}
  \begin{displaymath}
    \mathcal{M}(\left\langle v_{i}\right|)E_{h}OX_{g}OE_{h}\mathcal{M}(\left|v_{i}\right\rangle )=\left\langle x_{g}^{2j}\right\rangle \cdot\left\langle x_{g}^{2i}\right\rangle \cdot\left\langle x_{g}^{i+j+1}\right\rangle .
  \end{displaymath}
  The highest possible power is attained for $i=j=n-2$. This yields
  $2n-3$ and thus, using Equation~\eqref{eq:mochonf0unbalanced}, we conclude that
  $\left\langle w_{i}\right|D\left|w_{j}\right\rangle =0$ for
  all $0\le i,j\le n-2$.
\end{proof}

\subsection{Solution to monomial assignments}
\label{subsec:solutionmonomialalgebraic}
As described in \cref{subsec:fassignmentequivmonomial}, there are four different types of monomial assignments depending on whether they are balanced or unbalanced and aligned or misaligned (nomenclature is justified below).
While one could find a single expression
for all of them, it does not seem to aid clarity. We, therefore, present the four solutions separately.
To go beyond the solutions to $f_0$-assignments, we additionally need to use pseudo-inverses $X_{h}^{\dashv}$ and $X_{g}^{\dashv}$. However, the key idea is essentially unchanged.

\subsubsection{The balanced case}
\label{subsubsec:balancedmonomialalgebraic}
Even (resp. odd) monomials align properly (resp. do not align properly) at the bottom (see  \cref{fig:balancedAlignedmAssignment}). This justifies our choice to call them aligned (resp. misaligned).
\begin{figure}[H]
  \begin{centering}
    \hfill{}\subfloat[$2n=8$, $m=2b=2$. Balanced aligned  monomial assignment \label{fig:balancedAlignedmAssignment}]{\begin{centering}
        \includegraphics[scale=1]{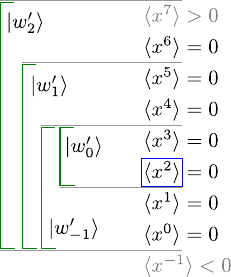}
        \par\end{centering}
    }\hfill{}\subfloat[$2n=8$, $m=2b-1=3$. Balanced misaligned  monomial assignment \label{fig:balancedMisalignedMassignment}]{\begin{centering}
        \includegraphics[scale=1]{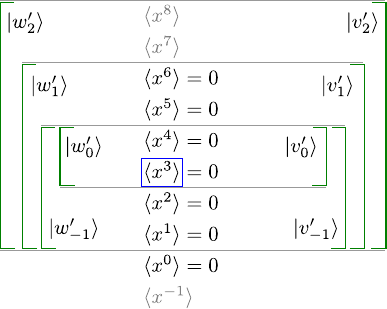}
        \par\end{centering}

    }\hfill{}
    \par\end{centering}
  \caption{Balanced monomial assignments. \atul{These are explained after the proof of \cref{prop:ExactSolnBalancedMonomialAligned}.}}%
\end{figure}
\begin{proposition}[Solution to balanced aligned monomial assignments]
  \label{prop:ExactSolnBalancedMonomialAligned}Let
  \begin{itemize}
    \item $m=2b$ be an even non-negative integer
    \item $
            t=\sum_{i=1}^{n}x_{h_{i}}^{m}p_{h_{i}}\left\llbracket x_{h_{i}}\right\rrbracket -\sum_{i=1}^{n}x_{g_{i}}^{m}p_{g_{i}}\left\llbracket x_{g_{i}}\right\rrbracket ,
          $ be a monomial assignment over $0<x_{1}<x_{2}\dots<x_{2n}$
    \item $\left\{ \left|h_{1}\right\rangle ,\left|h_{2}\right\rangle \dots\left|h_{n}\right\rangle ,\left|g_{1}\right\rangle ,\left|g_{2}\right\rangle \dots\left|g_{n}\right\rangle \right\} $
          be an orthonormal basis, and
    \item finally
          \begin{footnotesize}
            \begin{displaymath}
              X_{h}:=\sum_{i=1}^{n}x_{h_{i}}\left|h_{i}\right\rangle \left\langle h_{i}\right|\doteq\diag(x_{h_{1}},\dots x_{h_{n}},\underbrace{0,\dots0}_{n\text{ zeros}}),
              X_{g}:=\sum_{i=1}^{n}x_{g_{i}}\left|g_{i}\right\rangle \left\langle g_{i}\right|\doteq\diag(\underbrace{0,\dots0}_{n\text{ zeros}},x_{g_{1}},\dots x_{g_{n}}),
            \end{displaymath}
          \end{footnotesize}
          \begin{displaymath}
            \left|w\right\rangle :=\sum_{i=1}^{n}\sqrt{p_{h_{i}}}\left|h_{i}\right\rangle \doteq(\sqrt{p_{h_{1}}},\dots\sqrt{p_{h_{n}}},\underbrace{0,\dots0}_{n\text{ zeros}})^{T} \text{ and } \left|w'\right\rangle :=(X_{h})^{b}\left|w\right\rangle,
          \end{displaymath}
          \begin{displaymath}
            \left|v\right\rangle :=\sum_{i=1}^{n}\sqrt{p_{g_{i}}}\left|g_{i}\right\rangle \doteq(\underbrace{0,\dots0}_{n\text{ zeros}},\sqrt{p_{g_{1}}},\dots\sqrt{p_{g_{n}}})^{T} \text{ and }\left|v'\right\rangle :=(X_{g})^{b}\left|v\right\rangle.
          \end{displaymath}

  \end{itemize}
  Then,
  \begin{displaymath}
    O:=\sum_{i=-b}^{n-b-1}\left(\frac{\Pi_{h_{i}}^{\perp}(X_{h})^{i}\left|w'\right\rangle \left\langle v'\right|(X_{g})^{i}\Pi_{g_{i}}^{\perp}}{\sqrt{c_{h_{i}}c_{g_{i}}}}+\hc\right)
  \end{displaymath}
  satisfies
  $X_{h}\ge E_{h}OX_{g}O^{T}E_{h}$ and $E_{h}O\left|v'\right\rangle =\left|w'\right\rangle$,
  where we write $(X_{h/g})^{-k}$ instead of $(X_{h/g}^{\dashv})^{k}$
  (for $k>0$),
  $
    E_{h}:=\sum_{i=1}^{n}\left|h_{i}\right\rangle \left\langle h_{i}\right|, c_{h_{i}}:=\left\langle w'\right|(X_{h})^{i}\Pi_{h_{i}}^{\perp}(X_{h})^{i}\left|w'\right\rangle
  $
  \begin{small}
    \begin{displaymath}
      \Pi_{h_{i}}^{\perp}:=\begin{cases}
        \text{projector orthogonal to }\text{span}\{(X_{h})^{-|i|+1}\left|w'\right\rangle ,(X_{h})^{-|i|+2}\left|w'\right\rangle \dots,\left|w'\right\rangle \}        & i<0  \\
        \text{projector orthogonal to }\text{span}\{(X_{h})^{-b}\left|w'\right\rangle ,(X_{h})^{-b+1}\left|w'\right\rangle ,\dots(X_{h})^{i-1}\left|w'\right\rangle \} & i>0  \\
        \mathbb{I}                                                                                                                                                     & i=0,
      \end{cases}
    \end{displaymath}
  \end{small}
  and analogously $c_{g_{i}}:=\left\langle v'\right|(X_{g})^{i}\Pi_{g_{i}}^{\perp}(X_{g})^{i}\left|v'\right\rangle $ and
  \begin{small}
    \begin{displaymath}
      \Pi_{g_{i}}^{\perp}:=\begin{cases}
        \text{projector orthogonal to }\text{span}\{(X_{g})^{-|i|+1}\left|v'\right\rangle ,(X_{g})^{-|i|+2}\left|v'\right\rangle \dots,\left|v'\right\rangle \}        & i<0  \\
        \text{projector orthogonal to }\text{span}\{(X_{g})^{-b}\left|v'\right\rangle ,(X_{g})^{-b+1}\left|v'\right\rangle ,\dots(X_{g})^{i-1}\left|v'\right\rangle \} & i>0  \\
        \mathbb{I}                                                                                                                                                     & i=0.
      \end{cases}
    \end{displaymath}
  \end{small}

\end{proposition}

\begin{proof}
  The orthonormal basis
  of interest here is
  \begin{equation}
    \left|w'_{i}\right\rangle :=\frac{\Pi_{h_{i}}^{\perp}(X_{h})^{i}\left|w'\right\rangle }{\sqrt{c_{h_{i}}}}, \text{ which entails }\label{eq:alignedBalancedMonomialVectors}
  \end{equation}
  \begin{equation}
    \Pi_{h_{i}}^{\perp}=\begin{cases}
      \mathbb{I}_{h}                                                                       & i=0 \\
      \mathbb{I}_{h}-\sum_{j=i+1}^{0}\left|w'_{j}\right\rangle \left\langle w'_{j}\right|  & i<0 \\
      \mathbb{I}_{h}-\sum_{j=-b}^{i-1}\left|w'_{j}\right\rangle \left\langle w'_{j}\right| & i>0
    \end{cases}\label{eq:alignedBalancedMonomialProjectors}
  \end{equation}
  where $\mathbb{I}_{h}:=E_{h}$. We define $\left|v'_{i}\right\rangle $
  and $\Pi_{g_{i}}^{\perp}$ analogously. Here, we keep
  track of both the highest and lowest power, $l$ in $\left\langle w'\right|X_{h}^{l}\left|w'\right\rangle $
  and $\left\langle v'\right|X_{g}^{l}\left|v'\right\rangle $, which
  appear in the matrix elements $\left\langle w'_{i}\right|D\left|w'_{j}\right\rangle $.
  To this end, we use $\left\langle x_{h}^{l}\right\rangle ^{\prime}:=\left\langle w'\right|X_{h}^{l}\left|w'\right\rangle =\left\langle w\right|X_{h}^{l+2b}\left|w\right\rangle $
  and $\left\langle x_{g}^{l}\right\rangle ^{\prime}:=\left\langle v'\right|X_{g}^{l}\left|v'\right\rangle =\left\langle v\right|X_{g}^{l+2b}\left|v\right\rangle $.
  We denote the minimum and maximum powers, $l$, by
  \begin{displaymath}
    \mathcal{M}(\left|w'_{i}\right\rangle )=\begin{cases}
      \minmax{\left\langle x_{h}^{0}\right\rangle ^{\prime}\left|w'\right\rangle }{\left\langle x_{h}^{0}\right\rangle ^{\prime}\left|w'\right\rangle }                           & i=0  \\
      \minmax{\left\langle x_{h}^{-2|i|}\right\rangle ^{\prime}(X_{h})^{-|i|}\left|w'\right\rangle }{\left\langle x_{h}^{0}\right\rangle ^{\prime}\left|w'\right\rangle }         & i<0  \\
      \minmax{\left\langle x_{h}^{-2b}\right\rangle ^{\prime}(X_{h})^{-b}\left|w'\right\rangle }{\left\langle x_{h}^{2i}\right\rangle ^{\prime}(X_{h})^{i}\left|w'\right\rangle } & i>0,
    \end{cases}
  \end{displaymath}and we define $D:=X_{h}-E_{h}OX_{g}O^{T}E_{h}\doteq\left\langle w'_{i}\right|\left(X_{h}-E_{h}OX_{g}O^{T}E_{h}\right)\left|w_{j}'\right\rangle$, as usual.
  It suffices to restrict to the span of the $\{\left|w_{i}'\right\rangle \}$
  basis because $X_{h}\left|v'_{i}\right\rangle =0$ and $E_{h}\left|v'_{i}\right\rangle =0$.
  The lowest power, $l$, appearing in $D$ is attained for $i=j=-b$ (as $-b\le i,j\le n-b-1$).
  This can be evaluated to be $-2b$ by observing that
  \begin{displaymath}
    \mathcal{M}(\left\langle w'_{-b}\right|)X_{h}\mathcal{M}(\left|w'_{-b}\right\rangle )=\minmax{\left\langle x_{h}^{-2b}\right\rangle ^{\prime}\left\langle x_{h}^{-2b}\right\rangle ^{\prime}\left\langle x_{h}^{-2b+1}\right\rangle ^{\prime}}{\left\langle x_{h}^{0}\right\rangle ^{\prime}\left\langle x_{h}^{0}\right\rangle ^{\prime}\left\langle x_{h}\right\rangle ^{\prime}},
  \end{displaymath}
  where we multiplied component-wise. To find the highest power, $l$,
  in the matrix $D$, note that for $i,j>0$ we have
  \begin{displaymath}
    \mathcal{M}(\left\langle w'_{i}\right|)X_{h}\mathcal{M}(\left|w'_{j}\right\rangle )=\minmax{\left\langle x_{h}^{-2b}\right\rangle ^{\prime}\left\langle x_{h}^{-2b+1}\right\rangle ^{\prime}\left\langle x_{h}^{-2b}\right\rangle ^{\prime}}{\left\langle x_{h}^{2i}\right\rangle ^{\prime}\left\langle x_{h}^{2j}\right\rangle ^{\prime}\left\langle x_{h}^{i+j+1}\right\rangle ^{\prime}}
  \end{displaymath}
  so $l=\max\{2i,2j,i+j+1\}$. As argued for the $f_{0}$-assignment,
  $l=2n-2b-1$ for $i=j=n-b-1$, otherwise $l<2n-2b-1$.
  Thus, only the $D_{n-b-1,n-b-1}$ term in $D$, depends on $\left\langle x_{h}^{2n-2b-1}\right\rangle ^{\prime}$.
  All other terms, at most, depend on $\left\langle x_{h}^{-2b}\right\rangle ^{\prime},\left\langle x_{h}^{-2b+1}\right\rangle ^{\prime},\dots\left\langle x_{h}^{2n-2b-2}\right\rangle ^{\prime}$,
  i.e. $\left\langle x_{h}^{0}\right\rangle ,\left\langle x_{h}^{1}\right\rangle ,\dots\left\langle x_{h}^{2n-2}\right\rangle $.
  The analogous argument for $\left\langle v'_{i}\right|X_{g}\left|v'_{j}\right\rangle $,
  the observation that $\left\langle w'_{i}\right|D\left|w'_{j}\right\rangle =\left\langle w'_{i}\right|X_{h}\left|w'_{j}\right\rangle -\left\langle v'_{i}\right|X_{g}\left|v'_{j}\right\rangle $,
  and the fact that $\left\langle x^{0}\right\rangle =\left\langle x^{1}\right\rangle =\dots=\left\langle x^{2n-2}\right\rangle =0$
  entail that these terms vanish. It remains to show that $D_{n-b-1,n-b-1}\ge0$.
  Noting that in $\left\langle w'_{n-b-1}\right|D\left|w'_{n-b-1}\right\rangle $,
  the only term which would not get cancelled due to the aforesaid reasoning,
  must come from the part of $\left|w'_{n-b-1}\right\rangle $ containing
  $X_{h}^{n-b-1}\left|w'\right\rangle $. It suffices to show that the
  coefficient of this term is positive because we know that $\left\langle x^{2n-2b-1}\right\rangle ^{\prime}=\left\langle x^{2n-1}\right\rangle >0$.
  We know this coefficient to be exactly $1/c_{h_{n-b-1}}$
  (see Equation~\eqref{eq:alignedBalancedMonomialProjectors} and Equation~\eqref{eq:alignedBalancedMonomialVectors})
  establishing that $D\ge0$.
\end{proof}

To proceed further, it is helpful to have a more concise way of viewing
the proof. Let us consider a concrete example of a balanced
aligned monomial assignment with $2n=8$ and $m=2b=2$ (see  \cref{fig:balancedAlignedmAssignment}).
We represent the range of dependence of $\left\langle w'_{0}\right|X_{h}\left|w'_{0}\right\rangle $
on $\left\langle x_{h}^{l}\right\rangle $ diagrammatically by enclosing
in a left bracket, the terms $\left\langle x^{3}\right\rangle =\left\langle x\right\rangle ^{\prime}$
and $\left\langle x^{2}\right\rangle =\left\langle x^{0}\right\rangle ^{\prime}$
(replacing $\left|w\right\rangle $ with $\left|w'_{0}\right\rangle $)
and writing $\left|w'_{0}\right\rangle $ next to it. Similarly, for
$\left|w'_{-1}\right\rangle ,\left|w'_{1}\right\rangle $ and $\left|w'_{2}\right\rangle $
we enclose in a left bracket, the terms
\begin{displaymath}
  \left\{ \left\langle x^{0}\right\rangle ,\left\langle x^{1}\right\rangle ,\left\langle x^{2}\right\rangle ,\left\langle x^{3}\right\rangle \right\} =\left\{ \left\langle x^{-2}\right\rangle ^{\prime},\left\langle x^{-1}\right\rangle ^{\prime},\dots\left\langle x\right\rangle ^{\prime}\right\} ,\end{displaymath}
\begin{displaymath}\left\{ \left\langle x^{0}\right\rangle ,\left\langle x^{1}\right\rangle ,\dots,\left\langle x^{5}\right\rangle \right\} =\left\{ \left\langle x^{-2}\right\rangle ^{\prime},\left\langle x^{-1}\right\rangle ^{\prime},\dots\left\langle x^{3}\right\rangle ^{\prime}\right\},
\end{displaymath}
\begin{displaymath}\text{ and }
  \left\{ \left\langle x^{0}\right\rangle ,\left\langle x^{1}\right\rangle ,\dots\left\langle x^{7}\right\rangle \right\} =\left\{ \left\langle x^{-2}\right\rangle ^{\prime},\left\langle x^{-1}\right\rangle ^{\prime},\dots\left\langle x^{5}\right\rangle ^{\prime}\right\},
\end{displaymath}respectively. The highest power $l$ of $\left\langle x_{h}^{l}\right\rangle $
that appears in $\left\langle w'_{i}\right|X_{h}\left|w'_{j}\right\rangle $
is $l=7$ when (and only when) $i=j=2$. Thus, the matrix $D$, restricted
to the subspace spanned by the $\{\left|w_{i}'\right\rangle \}$ basis
(again, we can safely ignore the subspace $\text{span}\{\left|v_{i}'\right\rangle \}$
because $D\left|v'_{i}\right\rangle =0$), has only one non-zero entry
which we saw was positive as $\left\langle x^{7}\right\rangle >0$.

A direct extension of this analysis to the balanced
misaligned monomial assignment fails, as we can see concretely in the case with $2n=8$ and
$m=2b-1=3$ (see  \cref{fig:balancedMisalignedMassignment}). From hindsight,
we write both the $\left|v'_{i}\right\rangle $s and the $\left|w'_{i}\right\rangle $s.
We start with $\left|w'_{0}\right\rangle =X_{h}^{3/2}\left|w\right\rangle $
and $\left|v'_{0}\right\rangle =X_{g}^{3/2}\left|v_{0}\right\rangle $,
and as before, enclose the terms $\left\{ \left\langle x^{0}\right\rangle ^{\prime}=\left\langle x^{3}\right\rangle ,\left\langle x^{1}\right\rangle ^{\prime}=\left\langle x^{4}\right\rangle \right\} $
in a left bracket. We then multiply $\left|w'_{0}\right\rangle $
with $X_{h}^{-1}$ (and $\left|v'_{0}\right\rangle $ with $X_{g}^{-1}$
respectively) and project out the components along the previous
vectors. We represent these by $\left|w_{-1}'\right\rangle $ and
$\left|v'_{-1}\right\rangle $, and in the figure we enclose the terms
$\left\{ \left\langle x\right\rangle =\left\langle x^{-2}\right\rangle ^{\prime},\left\langle x^{2}\right\rangle =\left\langle x^{-1}\right\rangle ^{\prime}\dots\left\langle x^{4}\right\rangle =\left\langle x\right\rangle ^{\prime}\right\} $
in the left and right brackets. We do not go lower, because then we
pickup a dependence on $\left\langle x^{-1}\right\rangle $ which persists for subsequent vectors. In general, we stop after taking
$b$ steps down (here $b=1$). We go up by multiplying $\left|w_{0}'\right\rangle $ with $X_{h}$ (and $\left|v_{0}'\right\rangle $
with $X_{g}$ resp.) and projecting out the components along the previous
vectors. We represent these by $\left|w_{1}'\right\rangle $ and  $\left|v_{1}'\right\rangle $, and in
the figure we enclose the terms $\left\{ \left\langle x\right\rangle =\left\langle x^{-2}\right\rangle ^{\prime},\left\langle x^{2}\right\rangle =\left\langle x^{-1}\right\rangle ^{\prime}\dots\left\langle x^{6}\right\rangle =\left\langle x^{3}\right\rangle ^{\prime}\right\} $
in the brackets. Finally, we construct $\left|w'_{2}\right\rangle $
and $\left|v'_{2}\right\rangle $ by taking a step up using $X_{h}$
and $X_{g}$, respectively. These are essentially fixed to be the vectors orthogonal to the previous ones, once we restrict to $\text{span}\{\ket{h_1},\ket{h_2,}\ldots\ket{h_n}\}$ and $\text{span}\{\ket{g_1},\ket{g_2,}\ldots\ket{g_n}\}$.  Taking a step down using $X_{h}^{-1}$ and $X_{g}^{-1}$ we could
have constructed $\left|w'_{-2}\right\rangle $ and $\left|v'_{-2}\right\rangle $, but these are the same as $\left|w'_{2}\right\rangle $
and $\left|v'_{2}\right\rangle$, as we have a 3-dimensional space. If
we were to use $O=\sum_{i=-1}^{2}\left(\left|w'_{i}\right\rangle \left\langle v'_{i}\right|+\hc\right)$
then we would have obtained dependence on $\left\langle x^{7}\right\rangle $
in the row corresponding to $\left|w'_{2}\right\rangle $ and
a dependence on $\left\langle x^{8}\right\rangle $ for the term
$\left\langle w'_{2}\right|D\left|w'_{2}\right\rangle $. This
already hints that the matrix is negative because it has the form
$\left[\begin{array}{cc}
      0 & b \\
      b & c
    \end{array}\right]$ with $b\neq0$; thus this choice cannot work. We therefore
define $O:=\left(\sum_{i=-1}^{1}\left|w'_{i}\right\rangle \left\langle v'_{i}\right|+\hc\right)+\left|w'_{2}\right\rangle \left\langle w'_{2}\right|+\left|v'_{2}\right\rangle \left\langle v'_{2}\right|$. Further, instead of using
\begin{equation}
  X_{h}\ge E_{h}OX_{g}O^{T}E_{h}\label{eq:balancedMisalignedCaseEx}
\end{equation}
for establishing positivity, we equivalently use
\begin{equation}
  E_{h}\ge\left(X_{h}^{\dashv}\right)^{1/2}OX_{g}O^{T}\left(X_{h}^{\dashv}\right)^{1/2},\label{eq:balancedMisalignedInvertedEx}
\end{equation}
which is easily obtained by multiplying by $(X_{h}^{\dashv})^{1/2}$
on both sides. The reason is that to establish positivity, we
must include $\left|w'_{2}\right\rangle $ in the basis (we can neglect
the null vectors of $E_{h}$), and even though the RHS of Equation~\eqref{eq:balancedMisalignedCaseEx}
would not contribute, the LHS would get non-trivial contributions
along the rows. Using the
inverses allows us to remove this dependence. To see this, note that $\text{span}\{\left|w'_{-1}\right\rangle ,\left|w'_{0}\right\rangle \dots\left|w'_{2}\right\rangle \}$
equals the $h$-space, i.e. $\text{span}\{\left|h_{1}\right\rangle ,\left|h_{2}\right\rangle \dots\left|h_{n}\right\rangle \}$.
Further, $\text{span}\{X_{h}^{1/2}\left|w'_{i}\right\rangle \}_{i=-1}^{2}$
also equals the $h$-space (but the vectors are not, in general, orthonormal
any more). Finally, observe that $X_{h}^{1/2}\left|w'_{2}\right\rangle $
is a null vector of the RHS of Equation~\eqref{eq:balancedMisalignedInvertedEx}.
Therefore, to prove the positivity it suffices to restrict
to $\text{span}\{X_{h}^{1/2}\left|w'_{i}\right\rangle \}_{i=-1}^{1}$.
An arbitrary normalized vector in this space can be written as
\begin{footnotesize}
  \begin{align*}
     & \left|\psi\right\rangle   =\frac{\sum_{i=-1}^{1}\alpha_{i}X_{h}^{1/2}\left|w_{i}'\right\rangle }{\sqrt{\sum_{i,j=-1}^{1}\alpha_{i}\alpha_{j}\left\langle w'_{i}\right|X_{h}\left|w'_{j}\right\rangle }}
    \implies X_{g}^{1/2}O^{T}(X_{h}^{\dashv})^{1/2}\left|\psi\right\rangle  =\frac{\sum_{i=-1}^{1}\alpha_{i}X_{g}^{1/2}\left|v'_{i}\right\rangle }{\sqrt{\sum_{i,j=-1}^{1}\alpha_{i}\alpha_{j}\left\langle w'_{i}\right|X_{h}\left|w'_{j}\right\rangle }}                                                                         \\
     & \implies\left\langle \psi\right|(X_{h}^{\dashv})^{1/2}OX_{g}O^{T}(X_{h}^{\dashv})^{1/2}\left|\psi\right\rangle   =\frac{\sum_{i,j=-1}^{1}\alpha_{i}\alpha_{j}\left\langle v_{i}'\right|X_{g}\left|v_{j}'\right\rangle }{\sum_{i,j=-1}^{1}\alpha_{i}\alpha_{j}\left\langle w'_{i}\right|X_{h}\left|w'_{j}\right\rangle }=1,
  \end{align*}
\end{footnotesize}
where we get equality by noting that $\left\langle v'_{i}\right|X_{g}\left|v'_{j}\right\rangle $s
depend on (at most) $\left\{ \left\langle x_{g}\right\rangle ,\left\langle x_{g}^{2}\right\rangle \dots\left\langle x_{g}^{6}\right\rangle \right\} $
and analogously $\left\langle w'_{i}\right|X_{h}\left|w'_{j}\right\rangle $
depend on (at most) $\left\{ \left\langle x_{h}\right\rangle ,\left\langle x_{h}^{2}\right\rangle \dots\left\langle x_{h}^{6}\right\rangle \right\} $,
which are the same as $\left\langle x^{i}\right\rangle =0$
for $i\in\{0,1,\dots6\}$. Since we proved the RHS of Equation~\eqref{eq:balancedMisalignedInvertedEx}
equals $1$ for all normalized $\left|\psi\right\rangle $s, we conclude that
we have the correct unitary.

\begin{proposition}[Solution to balanced misaligned monomial assignments]
  \label{prop:ExactSolnBalancedMonomialMisaligned}Let
  \begin{itemize}
    \item $m=2b-1$ be an odd non-negative integer (i.e. $b\ge1$)
    \item $
            t=\sum_{i=1}^{n}x_{h_{i}}^{m}p_{h_{i}}\left\llbracket x_{h_{i}}\right\rrbracket -\sum_{i=1}^{n}x_{g_{i}}^{m}p_{g_{i}}\left\llbracket x_{g_{i}}\right\rrbracket ,
          $ be a monomial assignment over $\{x_{1},x_{2}\dots x_{2n}\}$
    \item $\left(\left|h_{1}\right\rangle ,\left|h_{2}\right\rangle \dots\left|h_{n}\right\rangle ,\left|g_{1}\right\rangle ,\left|g_{2}\right\rangle \dots\left|g_{n}\right\rangle \right)$
          be an orthonormal basis
    \item finally
          \begin{footnotesize}
            \begin{displaymath}
              X_{h}:=\sum_{i=1}^{n}x_{h_{i}}\left|h_{i}\right\rangle \left\langle h_{i}\right|\doteq\diag(x_{h_{1}},\dots x_{h_{n}},\underbrace{0,\dots0}_{n\text{ zeros}}),
              X_{g}:=\sum_{i=1}^{n}x_{g_{i}}\left|g_{i}\right\rangle \left\langle g_{i}\right|\doteq\diag(\underbrace{0,\dots0}_{n\text{ zeros}},x_{g_{1}},\dots x_{g_{n}}),
            \end{displaymath}
          \end{footnotesize}

          \begin{displaymath}
            \left|w\right\rangle :=(\sqrt{p_{h_{1}}},\dots\sqrt{p_{h_{n}}},\underbrace{0,\dots0}_{n\text{ zeros}}) \text{ and } \left|w'\right\rangle :=(X_{h})^{b-\frac{1}{2}}\left|w\right\rangle
          \end{displaymath}
          \begin{displaymath}
            \left|v\right\rangle :=(\underbrace{0,\dots0}_{n\text{ zeros}},\sqrt{p_{g_{1}}},\dots\sqrt{p_{g_{n}}}) \text{ and }\left|v'\right\rangle :=(X_{g})^{b-\frac{1}{2}}\left|v\right\rangle.
          \end{displaymath}

  \end{itemize}
  Then,\begin{footnotesize}
    \begin{align*}
      O & :=\sum_{i=-b+1}^{n-b-1}\left(\frac{\Pi_{h_{i}}^{\perp}(X_{h})^{i}\left|w'\right\rangle \left\langle v'\right|(X_{g})^{i}\Pi_{g_{i}}^{\perp}}{\sqrt{c_{h_{i}}c_{g_{i}}}}+\hc\right)
      \quad+\frac{\Pi_{g_{n-b}}^{\perp}(X_{g})^{n-b}\left|v'\right\rangle \left\langle v'\right|(X_{g})^{n-b}\Pi_{g_{n-b}}^{\perp}}{c_{g_{n-b+1}}}                                           \\
        & \quad+\frac{\Pi_{h_{n-b}}^{\perp}(X_{h})^{n-b}\left|w'\right\rangle \left\langle w'\right|(X_{h})^{n-b}\Pi_{h_{n-b}}^{\perp}}{c_{h_{n-b}}}
    \end{align*}
  \end{footnotesize}
  satisfies $X_{h}\ge E_{h}OX_{g}O^{T}E_{h}$ and $E_{h}O\left|v'\right\rangle =\left|w'\right\rangle$,
  where we write $X_{h/g}^{-k}$ instead of $(X_{h/g}^{\dashv})^{k}$
  for $k>0$, $c_{h_{i}}:=\left\langle w'\right|(X_{h})^{i}\Pi_{h_{i}}^{\perp}(X_{h})^{i}\left|w'\right\rangle $,
  \begin{footnotesize}
    \begin{displaymath}
      \Pi_{h_{i}}^{\perp}:=\begin{cases}
        \text{projector orthogonal to }\text{span}\{(X_{h}^{\dashv})^{|i|-1}\left|w'\right\rangle ,(X_{h}^{\dashv})^{|i|-2}\left|w'\right\rangle \dots,\left|w'\right\rangle \}                                                                   & i<0  \\
        \text{projector orthogonal to }\text{span}\{(X_{h}^{\dashv})^{b-1}\left|w'\right\rangle ,(X_{h}^{\dashv})^{b-2}\left|w'\right\rangle ,\dots,\left|w'\right\rangle ,X_{h}\left|w'\right\rangle ,\dots(X_{h})^{i-1}\left|w'\right\rangle \} & i>0  \\
        \mathbb{I}                                                                                                                                                                                                                                & i=0,
      \end{cases}
    \end{displaymath}
  \end{footnotesize}
  and analogously  $c_{g_{i}}:=\left\langle v'\right|(X_{g})^{i}\Pi_{g_{i}}^{\perp}(X_{g})^{i}\left|v'\right\rangle $,
  \begin{footnotesize}
    \begin{displaymath}
      \Pi_{g_{i}}^{\perp}:=\begin{cases}
        \text{projector orthogonal to }\text{span}\{(X_{g}^{\dashv})^{|i|-1}\left|v'\right\rangle ,(X_{g}^{\dashv})^{|i|-2}\left|v'\right\rangle \dots,\left|v'\right\rangle \}                                                                  & i<0  \\
        \text{projector orthogonal to }\text{span}\{(X_{g}^{\dashv})^{b-1}\left|v'\right\rangle ,(X_{g}^{\dashv})^{b-2}\left|v'\right\rangle ,\dots\left|v'\right\rangle ,X_{g}\left|v'\right\rangle ,\dots(X_{g})^{i-1}\left|v'\right\rangle \} & i>0  \\
        \mathbb{I}                                                                                                                                                                                                                               & i=0.
      \end{cases}
    \end{displaymath}
  \end{footnotesize}
\end{proposition}

\begin{proof}
  The proof is very similar to that of \cref{prop:ExactSolnBalancedMonomialAligned}.
  The orthonormal basis of interest here is
  \begin{displaymath}
    \left|w'_{i}\right\rangle :=\frac{\Pi_{h_{i}}^{\perp}(X_{h})^{i}\left|w'\right\rangle }{\sqrt{c_{h_{i}}}}
  \end{displaymath}
  which entails
  \begin{displaymath}
    \Pi_{h_{i}}^{\perp}=\begin{cases}
      \mathbb{I}_{h}                                                                       & i=0 \\
      \mathbb{I}_{h}-\sum_{j=i-1}^{0}\left|w'_{j}\right\rangle \left\langle w'_{j}\right|  & i<0 \\
      \mathbb{I}_{h}-\sum_{j=-b+1}^{i}\left|w'_{j}\right\rangle \left\langle w'_{j}\right| & i>0
    \end{cases}
  \end{displaymath}
  where $\mathbb{I}_{h}:=E_{h}$. We define $\left|v'_{i}\right\rangle $
  and $\Pi_{g_{i}}^{\perp}$ analogously. Our strategy is to keep track
  of the highest and lowest powers, $l$, in $\left\langle w'\right|X_{h}^{l}\left|w'\right\rangle $
  and $\left\langle v'\right|X_{g}^{l}\left|v'\right\rangle $, which
  appear in the matrix elements $\left\langle w_{i}'\right|X_{h}\left|w'_{j}\right\rangle $
  and $\left\langle v'_{i}\right|X_{g}\left|v'_{j}\right\rangle $.
  For brevity  we write $\left\langle x_{h}^{l}\right\rangle ^{\prime}:=\left\langle w'\right|X_{h}^{l}\left|w'\right\rangle $
  and $\left\langle x_{g}^{l}\right\rangle ^{\prime}:=\left\langle v'\right|X_{g}^{l}\left|v'\right\rangle $.
  The minimum and maximum powers, $l$, are denoted by
  \begin{displaymath}
    \mathcal{M}(\left|w'_{i}\right\rangle )=\begin{cases}
      \left(\left\langle x_{h}^{0}\right\rangle ^{\prime}\left|w'\right\rangle ,\left\langle x_{h}^{0}\right\rangle ^{\prime}\left|w'\right\rangle \right)                                         & i=0  \\
      \left(\left\langle x_{h}^{-2\left|i\right|}\right\rangle ^{\prime}(X_{h})^{-\left|i\right|}\left|w'\right\rangle ,\left\langle x_{h}^{0}\right\rangle ^{\prime}\left|w'\right\rangle \right) & i<0  \\
      \left(\left\langle x_{h}^{-2(b-1)}\right\rangle ^{\prime}(X_{h})^{-(b-1)}\left|w'\right\rangle ,\left\langle x_{h}^{2i}\right\rangle ^{\prime}(X_{h})^{i}\left|w'\right\rangle \right)       & i>0.
    \end{cases}
  \end{displaymath}
  Establishing $X_{h}\ge E_{h}OX_{g}O^{T}E_{h}$ is equivalent
  to establishing
  \begin{equation}
    E_{h}\ge X_{h}^{-1/2}OX_{g}O^{T}X_{h}^{-1/2}.\label{eq:invertedBalancedMisaligned}
  \end{equation}
  It is easy to see that $X_{h}^{1/2}\left|w'_{n-b}\right\rangle $
  is a vector with zero eigenvalue for the RHS as $X_{g}O^{T}\left|w'_{n-b}\right\rangle =0$.
  Any vector $\left|\psi\right\rangle \in\text{span}\{\left|g_{1}\right\rangle ,\left|g_{2}\right\rangle \dots\left|g_{n}\right\rangle \}$
  is a vector with zero eigenvalue for both the LHS and the RHS. Thus, for
  the positivity we can restrict to $\text{span}\{\left|h_{1}\right\rangle ,\left|h_{2}\right\rangle ,\dots\left|h_{n}\right\rangle \}\backslash\text{span}\{X_{h}^{1/2}\left|w'_{n-b}\right\rangle \}$,
  i.e. to vectors in the $h$-space orthogonal to $X_{h}^{1/2}\left|w'_{n-b}\right\rangle $.
  It turns out to be easier to test for positivity on a larger
  space. It is clear that $\text{span}\left\{ X_{h}^{1/2}\left|w'_{i}\right\rangle \right\} _{i=-b+1}^{n-b}=\text{span}\{\left|h_{1}\right\rangle ,\left|h_{2}\right\rangle \dots\left|h_{n}\right\rangle \}=\text{span}\{\left|w' _{i}\right\rangle\}_{i=-b+1}^{n-b}$,
  (due to \cref{lem:spanningLemma}). As neglecting vectors with components
  along $X_{h}^{1/2}\left|w'_{n-b}\right\rangle $ suffices to satisfy Equation~\eqref{eq:invertedBalancedMisaligned}, we can restrict
  to $\text{span}\{X_{h}^{1/2}\left|w'_{i}\right\rangle \}_{i=-b+1}^{n-b-1}$
  (which might still contain vectors with components along $X_{h}^{1/2}\left|w'_{n-b}\right\rangle $
  as the basis vectors are not orthogonal but it only means that we
  check for positivity over a larger set of vectors). These ensure that
  the troublesome vectors  $\left|w'_{n-b}\right\rangle $ and
  $\left|v'_{n-b}\right\rangle $ do not appear in the remaining analysis.
  Let $\left|\psi\right\rangle =\left(\sum_{i=-b+1}^{n-b-1}\alpha_{i}X_{h}^{1/2}\left|w'_{i}\right\rangle \right)/c$
  where $c=\sqrt{\left\langle \psi|\psi\right\rangle }$. To establish
  Equation~\eqref{eq:invertedBalancedMisaligned}, it is enough to show that for
  all choices of $\alpha_{i}$s,
  \begin{align}
    1  \ge\left\langle \psi\right|X_{h}^{-1/2}OX_{g}O^{T}X_{h}^{-1/2}\left|\psi\right\rangle
    =\frac{\sum_{i,j=-b+1}^{n-b-1}\alpha_{i}\alpha_{j}\left\langle v'_{i}\right|X_{g}\left|v'_{j}\right\rangle }{\sum_{i,j=-b+1}^{n-b-1}\alpha_{i}\alpha_{j}\left\langle w'_{i}\right|X_{h}\left|w'_{j}\right\rangle } =1 \label{eq:ratioForInequality}
  \end{align}
  where the second step follows from $X_{g}^{1/2}O^{T}X_{h}^{-1/2}\left|\psi\right\rangle =\sum_{i=-b+1}^{n-b-1}\alpha_{i}X_{g}^{1/2}\left|v'_{i}\right\rangle $
  and the last step follows from the counting argument below. Start by noting that
  \begin{equation}
    \left\langle x_{h}^{i}\right\rangle ^{\prime}=\left\langle x_{h}^{i+2b-1}\right\rangle  \text{ and  }
    \left\langle x^{0}\right\rangle =\left\langle x\right\rangle =\dots=\left\langle x^{2n-2}\right\rangle =0.\label{eq:monomialmisalignedMochonPowers}
  \end{equation}
  To determine the highest power of $l$ in $\left\langle w'\right|X_{h}^{l}\left|w'\right\rangle $
  which appears in the matrix elements $\left\langle w'_{i}\right|X_{h}\left|w'_{j}\right\rangle $
  (for $-b+1\le i,j\le n-b-1$) it suffices to consider the expectation values $\left\langle w'_{n-b-1}\right|X_{h}\left|w'_{n-b-1}\right\rangle $.
  To this end, we evaluate
  \begin{small}
    \begin{align*}
       & \mathcal{M}(\left\langle w'_{n-b-1}\right|)X_{h}\mathcal{M}(\left|w'_{n-b-1}\right\rangle )                                                                                                                                                                                                                                              \\
       & =\left(\left\langle x_{h}^{-2(b-1)}\right\rangle ^{\prime}\left\langle x_{h}^{-2(b-1)}\right\rangle ^{\prime}\left\langle x_{h}^{-2(b-1)+1}\right\rangle ^{\prime},\left\langle x_{h}^{2(n-b-1)}\right\rangle ^{\prime}\left\langle x_{h}^{2(n-b-1)}\right\rangle ^{\prime}\left\langle x_{h}^{2(n-b-1)+1}\right\rangle ^{\prime}\right) \\
       & =\left(\left\langle x_{h}\right\rangle \left\langle x_{h}\right\rangle \left\langle x_{h}^{2}\right\rangle ,\left\langle x_{h}^{2n-3}\right\rangle \left\langle x_{h}^{2n-3}\right\rangle \left\langle x_{h}^{2n-2}\right\rangle \right).
    \end{align*}
  \end{small}	The highest power is, manifestly, $l=2n-2$. To find the lowest power
  $l$ in $\left\langle w'\right|X_{h}^{l}\left|w'\right\rangle $
  appearing in  $\left\langle w'_{i}\right|X_{h}\left|w'_{j}\right\rangle $
  (for $-b+1\le i,j\le n-b-1$) it suffices to consider $\left\langle w'_{-b+1}\right|X_{h}\left|w'_{-b+1}\right\rangle $.
  To this end, we evaluate
  \begin{small}
    \begin{align*}
      \mathcal{M}(\left\langle w'_{-b+1}\right|)X_{h}\mathcal{M}(\left|w'_{-b+1}\right\rangle ) & =\left(\left\langle x_{h}^{-2(b-1)}\right\rangle ^{\prime}\left\langle x_{h}^{-2(b-1)}\right\rangle ^{\prime}\left\langle x_{h}^{-2(b-1)+1}\right\rangle ^{\prime},\left\langle x_{h}^{0}\right\rangle ^{\prime}\left\langle x_{h}^{0}\right\rangle ^{\prime}\left\langle x_{h}\right\rangle ^{\prime}\right) \\
                                                                                                & =\left(\left\langle x_{h}\right\rangle \left\langle x_{h}\right\rangle \left\langle x_{h}^{2}\right\rangle ,\left\langle x_{h}^{2b-1}\right\rangle \left\langle x_{h}^{2b-1}\right\rangle \left\langle x_{h}^{2b}\right\rangle \right).
    \end{align*}
  \end{small}
  The lowest power is, manifestly, $l=1$. We thus conclude that the
  numerator of Equation~\eqref{eq:ratioForInequality} is a function of $\left\langle x_{h}\right\rangle ,\left\langle x_{h}^{2}\right\rangle ,\dots\text{\ensuremath{\left\langle x_{h}^{2n-2}\right\rangle }}$ and,
  an analogous argument entails that the denominator is a function of
  $\left\langle x_{g}\right\rangle ,\left\langle x_{g}^{2}\right\rangle ,\dots\left\langle x_{g}^{2n-2}\right\rangle $
  with the same form. Using Equation~\eqref{eq:monomialmisalignedMochonPowers},
  we conclude that the numerator and the denominator are the same.
\end{proof}

\subsubsection{The unbalanced case}
\label{subsubsec:unbalancedmonomialalgebraic}

The techniques we have used so far also work when the number of points
in a monomial assignment are odd (i.e. for unbalanced monomial assignments), both aligned and misaligned.  We illustrate how
the solution is constructed by considering a concrete example of an
unbalanced aligned monomial assignment. We start with $2n-1=7$
points and $m=2b=2$ (see  \cref{fig:EvenUnbalancedMassignment}). We
use the diagrammatic representation introduced previously. In this case,
we have $4$ initial and $3$ final points; the standard basis
is $\left\{ \left|g_{1}\right\rangle ,\left|g_{2}\right\rangle ,\dots\left|g_{4}\right\rangle ,\left|h_{1}\right\rangle ,\left|h_{2}\right\rangle ,\left|h_{3}\right\rangle \right\} $.
\begin{figure}[H]
  \hfill{}\subfloat[$2n-1=7;$ $m=2b=2$. Unbalanced aligned monomial assignment.\label{fig:EvenUnbalancedMassignment}]{\begin{centering}
      \includegraphics[scale=1]{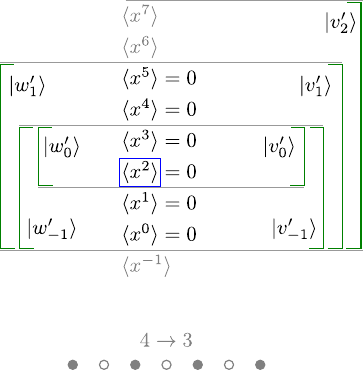}
      \par\end{centering}

  }\hfill{}\subfloat[$2n-1=7$; $m=2b-1=1$. Unbalanced misaligned monomial assignment.\label{fig:OddUnbalancedMassignment}]{\begin{centering}
      \includegraphics[scale=1]{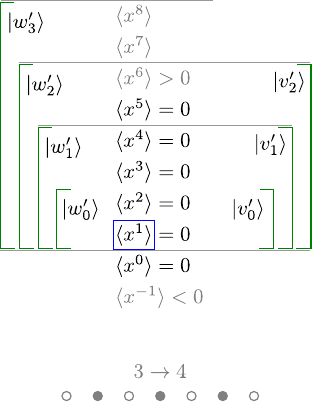}
      \par\end{centering}

  }\hfill{}

  \caption{\atul{Unbalanced monomial assignments. These are explained at the start of Section~\ref{subsubsec:unbalancedmonomialalgebraic}.}}
\end{figure}
The basis of interest is again constructed by starting at $\left|w'\right\rangle $ and using $X_{h}^{-1}$ until we reach $\left\langle x^{0}\right\rangle $, and then
by using $X_{h}$ until the space is spanned (analogously for $\left|v'\right\rangle$ with $X^{-1}_g $ and $X_g$). It is $\left\{ \left|v'_{-1}\right\rangle ,\left|v'_{0}\right\rangle ,\left|v'_{1}\right\rangle ,\left|v'_{2}\right\rangle \right\} $
and $\left\{ \left|w'_{-1}\right\rangle ,\left|w'_{0}\right\rangle ,\left|w'_{1}\right\rangle \right\} $.
In the same vein as the earlier solutions, we define $O:=\sum_{i=-1}^{1}\left(\left|w'_{i}\right\rangle \left\langle v'_{i}\right|+\hc\right)+\left|v'_{2}\right\rangle \left\langle v'_{2}\right|$.
In $X_{h}\ge E_{h}OX_{g}O^{T}E_{h}$, the $\left|v'_{2}\right\rangle $
term is removed by the projector, $E_{h}:=\sum_{i=1}^{3}\left|h_{i}\right\rangle \left\langle h_{i}\right|$.
Using $\left\langle x^{0}\right\rangle =\left\langle x\right\rangle =\dots=\left\langle x^{5}\right\rangle =0$
and the counting arguments from before, it follows that $D=X_{h}-E_{h}OX_{g}O^{T}E_{h}=0.$

For an
unbalanced misaligned monomial assignment let us consider the example with $2n-1=7$
and $m=2b-1=1$ \atul{(see  \cref{fig:OddUnbalancedMassignment})}. We have $3$ initial and $4$
final points; the standard basis is\\ $\left\{ \left|g_{1}\right\rangle ,\left|g_{2}\right\rangle ,\left|g_{3}\right\rangle ,\left|h_{1}\right\rangle ,\left|h_{2}\right\rangle ,\dots\left|h_{4}\right\rangle \right\}$.
We construct the basis of interest by starting at
$\left|w'\right\rangle $
and using $X_{h}$ until the space is spanned (analogously for $\left|v'\right\rangle $ with $X_g$). More generally, we
first go down for $b-2$ steps (which is zero in this case), until
$\left\langle x\right\rangle $ is reached in the diagram. The bases
are $\left\{ \left|v'_{0}\right\rangle ,\left|v'_{1}\right\rangle ,\left|v'_{2}\right\rangle \right\} $
and $\left\{ \left|w'_{0}\right\rangle ,\left|w'_{1}\right\rangle ,\left|w'_{2}\right\rangle ,\left|w'_{3}\right\rangle \right\} $.
As before, we define $O:=\sum_{i=0}^{2}\left(\left|w'_{i}\right\rangle \left\langle v'_{i}\right|+\hc\right)+\left|w'_{3}\right\rangle \left\langle w'_{3}\right|$.
This time we
use $E_{h}\ge X_{h}^{-1/2}OX_{g}O^{T}X_{h}^{-1/2}$ which is equivalent
to $X_{h}\ge E_{h}OX_{g}O^{T}E_{h}$ for $E_{h}:=\sum_{i=1}^{4}\left|h_{i}\right\rangle \left\langle h_{i}\right|$.
Using an argument similar to the balanced misaligned case, we can
reduce the positivity condition to
\begin{displaymath}
  1\ge\frac{\sum_{i,j=0}^{2}\alpha_{i}\alpha_{j}\left\langle v'_{i}\right|X_{g}\left|v'_{j}\right\rangle }{\sum_{i,j=0}^{2}\alpha_{i}\alpha_{j}\left\langle w'_{i}\right|X_{h}\left|w'_{j}\right\rangle }
\end{displaymath}
but the counting argument doesn't make the fraction $1$. This is
because we now have an $\left\langle x_{h}^{6}\right\rangle $ dependence
in the denominator and an $\left\langle x_{g}^{6}\right\rangle $ dependence
in the numerator. However, we also know that this term only appears
in $\left\langle w'_{2}\right|X_{h}\left|w'_{2}\right\rangle $ that
too with a positive coefficient (as we saw in the unbalanced $f_{0}-$assignment). Further, we know $\left\langle x_{h}^{6}\right\rangle >\left\langle x_{g}^{6}\right\rangle $
and therefore we can conclude that the numerator is smaller than the
denominator ensuring the inequality is always satisfied.
We state the general solution for both these cases and prove their
correctness below.

\begin{proposition}[Solution to unbalanced aligned monomial assignments]
  \label{prop:ExactSolnUnbalancedMonomialAligned} Let
  \begin{itemize}
    \item $m=2b$ be an even non-negative integer
    \item $
            t=\sum_{i=1}^{n-1}x_{h_{i}}^{m}p_{h_{i}}\left\llbracket x_{h_{i}}\right\rrbracket -\sum_{i=1}^{n}x_{g_{i}}^{m}p_{g_{i}}\left\llbracket x_{g_{i}}\right\rrbracket ,
          $ be a monomial assignment over $\{x_{1},x_{2}\dots x_{2n-1}\}$
    \item $\left(\left|h_{1}\right\rangle ,\left|h_{2}\right\rangle \dots\left|h_{n-1}\right\rangle ,\left|g_{1}\right\rangle ,\left|g_{2}\right\rangle \dots\left|g_{n}\right\rangle \right)$
          be an orthonormal basis
    \item finally
          \begin{footnotesize}
            \begin{displaymath}
              X_{h}:=\sum_{i=1}^{n-1}x_{h_{i}}\left|h_{i}\right\rangle \left\langle h_{i}\right|\doteq\diag(x_{h_{1}},\dots x_{h_{n-1}},\underbrace{0,\dots0}_{n\text{ zeros}}),
              X_{g}:=\sum_{i=1}^{n}x_{g_{i}}\left|g_{i}\right\rangle \left\langle g_{i}\right|\doteq\diag(\underbrace{0,\dots0}_{n-1\text{ zeros}},x_{g_{1}},\dots x_{g_{n}}),
            \end{displaymath}
          \end{footnotesize}
          \begin{displaymath}
            \left|w\right\rangle :=(\sqrt{p_{h_{1}}},\dots\sqrt{p_{h_{n-1}}},\underbrace{0\dots0}_{n\text{ zeros}}) \text{ and }\left|w'\right\rangle :=(X_{h})^{b}\left|w\right\rangle,
          \end{displaymath}
          \begin{displaymath}
            \left|v\right\rangle :=(\underbrace{0,0,\dots0}_{n-1\text{ zeros}},\sqrt{p_{g_{1}}},\sqrt{p_{g_{2}}}\dots\sqrt{p_{g_{n}}}) \text{ and }\left|v'\right\rangle :=(X_{g})^{b}\left|v\right\rangle.
          \end{displaymath}
  \end{itemize}
  Then \begin{small}
    \begin{align*}
      O  :=\sum_{i=-b}^{n-b-2}\left(\frac{\Pi_{h_{i}}^{\perp}(X_{h})^{i}\left|w'\right\rangle \left\langle v'\right|(X_{g})^{i}\Pi_{g_{i}}^{\perp}}{\sqrt{c_{h_{i}}c_{g_{i}}}}+\hc\right)
      \quad+\frac{\Pi_{g_{n-b-1}}^{\perp}(X_{g})^{n-b-1}\left|v'\right\rangle \left\langle v'\right|(X_{g})^{n-b-1}\Pi_{g_{n-b-1}}^{\perp}}{c_{g_{n-b-1}}}
    \end{align*}
  \end{small}
  satisfies $X_{h}\ge E_{h}OX_{g}O^{T}E_{h}$ and $E_{h}O\left|v'\right\rangle =\left|w'\right\rangle$,
  where by $X_{h/g}^{-k}$ we mean $(X_{h/g}^{\dashv})^{k}$
  for $k>0$, and all $c_{h_{i}},c_{g_{i}},\Pi_{h_{i}}^{\perp},\Pi_{g_{i}}^{\perp}$
  are as defined in \cref{prop:ExactSolnBalancedMonomialAligned}.
\end{proposition}

\begin{proof}
  Many observations from the proof of \cref{prop:ExactSolnBalancedMonomialAligned}
  carry over to this case. We import the definitions of $\left\{ \left|w'_{i}\right\rangle \right\} _{i=-b}^{n-b-2}$
  and $\{\left|v'_{i}\right\rangle \}_{i=-b}^{n-b-1}$, together with
  the observations that $$\mathcal{M}(\left\langle w'_{-b}\right|)X_{h}\mathcal{M}(\left|w'_{-b}\right\rangle )$$
  has no dependence on a term $\left\langle x_{h}^{l}\right\rangle ^{\prime}$
  with $l<-2b$
  and that $$\mathcal{M}(\left\langle w'_{n-b-2}\right|)X_{h}\mathcal{M}(\left|w'_{n-b-2}\right\rangle )$$
  has no dependence on a term $\left\langle x_{h}^{l}\right\rangle ^{\prime}$
  with $l>2n-2b-4+1=2n-3-2b$. We can restrict to $\text{span}\{\left|w'_{-b}\right\rangle ,\left|w'_{-b+1}\right\rangle \dots\left|w'_{n-b-2}\right\rangle \}$
  to establish the positivity of $D:=X_{h}-E_{h}OX_{g}O^{T}E_{h}$.
  Using the analogous observation for $\mathcal{M}(\left\langle v'_{-b}\right|)X_{g}\mathcal{M}(\left|v'_{-b}\right\rangle )$
  and $\mathcal{M}(\left\langle v'_{n-b-2}\right|)X_{g}\mathcal{M}(\left|v'_{n-b-2}\right\rangle )$,
  along with the fact that $\left\langle x^{l}\right\rangle ^{\prime}=\left\langle x^{l+2b}\right\rangle $
  and $\left\langle x^{0}\right\rangle =\left\langle x^{1}\right\rangle =\dots=\left\langle x^{2n-3}\right\rangle =0$,
  it follows that $D=0$.
\end{proof}

\begin{proposition}[Solution to unbalanced misaligned monomial assignments]
  \label{prop:ExactSolnUnbalancedMonomialMisaligned}Let
  \begin{itemize}
    \item $m=2b-1$ be an odd non-negative integer
    \item $
            t=\sum_{i=1}^{n}x_{h_{i}}^{m}p_{h_{i}}\left\llbracket x_{h_{i}}\right\rrbracket -\sum_{i=1}^{n-1}x_{g_{i}}^{m}p_{g_{i}}\left\llbracket x_{g_{i}}\right\rrbracket
          $ be a monomial assignment over $\{x_{1},x_{2}\dots x_{2n-1}\}$
    \item $\left(\left|h_{1}\right\rangle ,\left|h_{2}\right\rangle \dots\left|h_{n}\right\rangle ,\left|g_{1}\right\rangle ,\left|g_{2}\right\rangle \dots\left|g_{n-1}\right\rangle \right)$
          be an orthonormal basis
    \item finally
          \begin{footnotesize}
            \begin{displaymath}
              X_{h}:=\sum_{i=1}^{n}x_{h_{i}}\left|h_{i}\right\rangle \left\langle h_{i}\right|\doteq\diag(x_{h_{1}},\dots x_{h_{n}},\underbrace{0,\dots0}_{n-1\text{ zeros}})
              X_{g}:=\sum_{i=1}^{n-1}x_{g_{i}}\left|g_{i}\right\rangle \left\langle g_{i}\right|\doteq\diag(\underbrace{0,\dots0}_{n\text{ zeros}},x_{g_{1}},\dots x_{g_{n-1}}),
            \end{displaymath}
          \end{footnotesize}\begin{displaymath}
            \left|w\right\rangle :=(\sqrt{p_{h_{1}}},\dots\sqrt{p_{h_{n}}},\underbrace{0,\dots0}_{n-1\text{ zeros}}) \text{ and }\left|w'\right\rangle :=(X_{h})^{b-\frac{1}{2}}\left|w\right\rangle,
          \end{displaymath}
          \begin{displaymath}
            \left|v\right\rangle :=(\underbrace{0,\dots0}_{n\text{ zeros}},\sqrt{p_{g_{1}}},\dots\sqrt{p_{g_{n-1}}}) \text{ and }\left|v'\right\rangle :=(X_{g})^{b-\frac{1}{2}}\left|v\right\rangle.
          \end{displaymath}

  \end{itemize}
  Then
  \begin{footnotesize}
    \begin{align*}
      O  :=\sum_{i=-b+1}^{n-b-1}\left(\frac{\Pi_{h_{i}}^{\perp}(X_{h})^{i}\left|w'\right\rangle \left\langle v'\right|(X_{g})^{i}\Pi_{g_{i}}^{\perp}}{\sqrt{c_{h_{i}}c_{g_{i}}}}+\hc\right) +\frac{\Pi_{h_{n-b}}^{\perp}(X_{h})^{n-b}\left|w'\right\rangle \left\langle w'\right|(X_{h})^{n-b}\Pi_{h_{n-b}}^{\perp}}{c_{h_{n-b}}},
    \end{align*}
  \end{footnotesize}
  satisfies $X_{h}\ge E_{h}OX_{g}O^{T}E_{h}$ and $E_{h}O\left|v'\right\rangle =\left|w'\right\rangle $,
  where by $X_{h/g}^{-k}$ we mean $(X_{h/g}^{\dashv})^{k}$
  for $k>0$, and all $c_{h_{i}},c_{g_{i}},\Pi_{h_{i}}^{\perp},\Pi_{g_{i}}^{\perp}$
  are as defined in \cref{prop:ExactSolnBalancedMonomialMisaligned}.
\end{proposition}

\begin{proof}
  For this proof, we can use the definitions and observations from the
  proof of \cref{prop:ExactSolnBalancedMonomialMisaligned}. We import
  the definitions of $\left\{ \left|w'_{i}\right\rangle \right\} {}_{i=-b+1}^{n-b}$
  and $\left\{ \left|v'_{i}\right\rangle \right\} _{i=-b+1}^{n-b-1}$
  along with the observation that
  \begin{displaymath}
    \mathcal{M}(\left\langle w'_{-b+1}\right|)X_{h}\mathcal{M}(\left|w'_{-b+1}\right\rangle )
  \end{displaymath}
  has no dependence on a term $\left\langle x_{h}^{l}\right\rangle ^{\prime}$
  with $l<-2b+2$
  and
  \begin{displaymath}
    \mathcal{M}(\left\langle w'_{n-b-1}\right|)X_{h}\mathcal{M}(\left|w'_{n-b-1}\right\rangle )
  \end{displaymath}
  has no dependence on a term $\left\langle x^{l}\right\rangle $ with
  $l>2n-2b-1$. Also from the previous
  proof we have that establishing $X_{h}\ge E_{h}OX_{g}O^{T}E_{h}$ is equivalent
  to establishing
  \begin{displaymath}
    1\ge\frac{\sum_{i,j=-b+1}^{n-b-1}\alpha_{i}\alpha_{j}\left\langle v'_{i}\right|X_{g}\left|v'_{j}\right\rangle }{\sum_{i,j=-b+1}^{n-b-1}\alpha_{i}\alpha_{j}\left\langle w'_{i}\right|X_{h}\left|w'_{j}\right\rangle }
  \end{displaymath}
  for all real $\{\alpha_{i}\}_{i=-b+1}^{n-b-1}$. We know that $\left\langle x\right\rangle =\left\langle x^{2}\right\rangle =\dots=\left\langle x^{2n-3}\right\rangle =0$.
  As we have the dependence on $\left\langle x_{h}^{2n-2}\right\rangle $,
  we can't conclude that the fraction is one. However, as we saw in
  the proof of \cref{prop:ExactSolnBalancedMonomialAligned}, dependence
  on $\left\langle x_{h}^{2n-2}\right\rangle $ in the denominator
  only appears in the $\left\langle w'_{n-b-1}\right|X_{h}\left|w'_{n-b-1}\right\rangle $
  term, that too with the positive coefficient, $1/c_{h_{n-b-1}}$.
  The analogous statement holds for the numerator. This, using $\left\langle x^{2n-2}\right\rangle >0$,
  entails that the denominator is larger than or equal to the numerator,
  concluding the proof.
\end{proof}
\subsection{Main result}
\label{subsec:mainalgebraic}
Our observations so far can be combined to prove \cref{thm:Main}, which
we formally state here.
\begin{theorem}
  Let $t$ be an $f$-assignment (see \cref{def:f_assignment-f_0_assignment-balanced-m_kmonomial})
  on strictly positive coordinates. 
  Suppose $f$ has real and strictly positive roots. Then, $t$ admits an effective solution (see \cref{def:solvingassignment}). More explicitly, decompose $t=\sum_{i}\alpha_{i}t_{i}'$ where $\alpha_{i}$ are positive and $t_{i}'$ are monomial assignments (see \cref{def:f_assignment-f_0_assignment-balanced-m_kmonomial} and \cref{lem:generalMonomialDecomposition}). Then, each $t_{i}'$ admits a solution
  given by either \cref{prop:ExactSolnBalancedMonomialAligned}, \cref{prop:ExactSolnBalancedMonomialMisaligned},
  \cref{prop:ExactSolnUnbalancedMonomialAligned}, or \cref{prop:ExactSolnUnbalancedMonomialMisaligned}.
\end{theorem}

\begin{proof}
  In \cref{subsec:fassignmentequivmonomial} we established that it suffices to express an $f$-assignment as a sum of monomial assignments and find the solution for each one of them, in order to find the solution to the $f$-assignment. A monomial assignment now, can be balanced or unbalanced and aligned or misaligned (see \cref{def:f_assignment-f_0_assignment-balanced-m_kmonomial}). The solution in each case is given by either
  \cref{prop:ExactSolnBalancedMonomialAligned},
  \cref{prop:ExactSolnBalancedMonomialMisaligned}, \cref{prop:ExactSolnUnbalancedMonomialAligned},
  or \cref{prop:ExactSolnUnbalancedMonomialMisaligned}.
\end{proof}
\subsection{Example: a bias-\texorpdfstring{$1/14$}{1/14} protocol}\label{subsec:1over14}
We conclude the discussion by briefly outlining how all the pieces fit together to give a WCF protocol with bias $1/14$ as an example. The $f$-assignment for the TIPG approaching bias $\epsilon(3)=1/14$
($k=3$ for $\epsilon(k)=\frac{1}{4k+2}$)  has the following
form. Let
\begin{displaymath}
  x_{0}'=0<r_{1}'<r_{2}'<x_{1}'<x_{2}'<x_{3}'<x_{4}'<x_{5}'<x_{6}'<r_{3}'<r_{4}'<r_{5}'.
\end{displaymath}
This is an $f$-assignment
on $\{x_{0}',x_{1}'\dots x_{6}'\}$ with $f'(x)=(r_{1}'-x)(r_{2}'-x)(r_{3}'-x)(r_{4}'-x)(r_{5}'-x)$
viz.
\begin{displaymath}
  t'=\sum_{i=0}^{6}\frac{-f'(x_{i}')}{\prod_{j\neq i}(x_{j}'-x_{i}')}\left\llbracket x_{i}'\right\rrbracket.
\end{displaymath}
\atul{ \cref{fig:TDPG-1by14} details how this TIPG can be viewed as three stages of the corresponding TDPG: the split, the ladder and the merge.}

\begin{figure}[ht]
  \begin{centering}
    \includegraphics[scale=1]{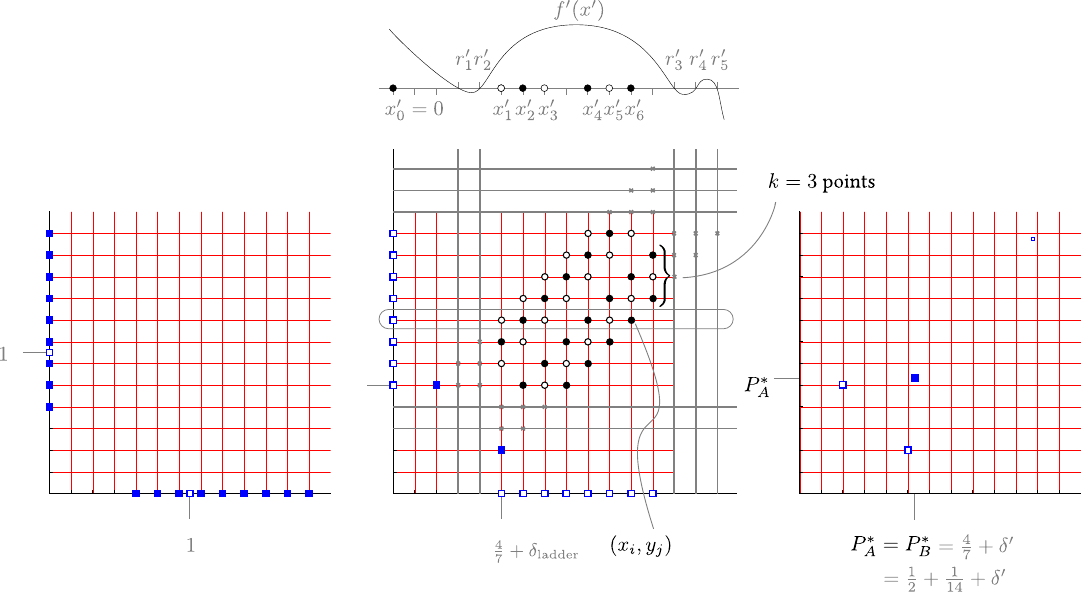}
    \par\end{centering}
  \caption{\label{fig:TDPG-1by14}The TDPG (or equivalently, the reversed protocol)
    approaching bias $\epsilon(k=3)=1/14$ may be seen as proceeding in
    three stages, as illustrated by the three images (left to right).
    \emph{First}, the initial points (indicated by unfilled squares) are
    split along the axes (indicated by the filled squares). \emph{Second},
    the points on the axes (unfilled squares) are transferred, by means of the
    ladder described in \cref{subsec:mochontipg} (indicated by the circles), into two final points (filled squares).
    \emph{Third}, the two points from the previous step (unfilled squares)
    and the catalyst state (indicated, after being raised into one point
    by the little unfilled box) are merged into the final point (filled
    box).
    The second stage is illustrated by the TIPG,---or more precisely, by its main move,
    the ladder---approaching bias $1/14$. The weight of these points is given (up to a
    constant) by the $f$\textendash assignment shown above. The roots
    of the polynomial correspond to the locations of the vertical lines
    and the location of the points in the graph is representative of the
    general construction.}
\end{figure}

\begin{figure}[ht]
  \begin{centering}
    \includegraphics[scale=1.03]{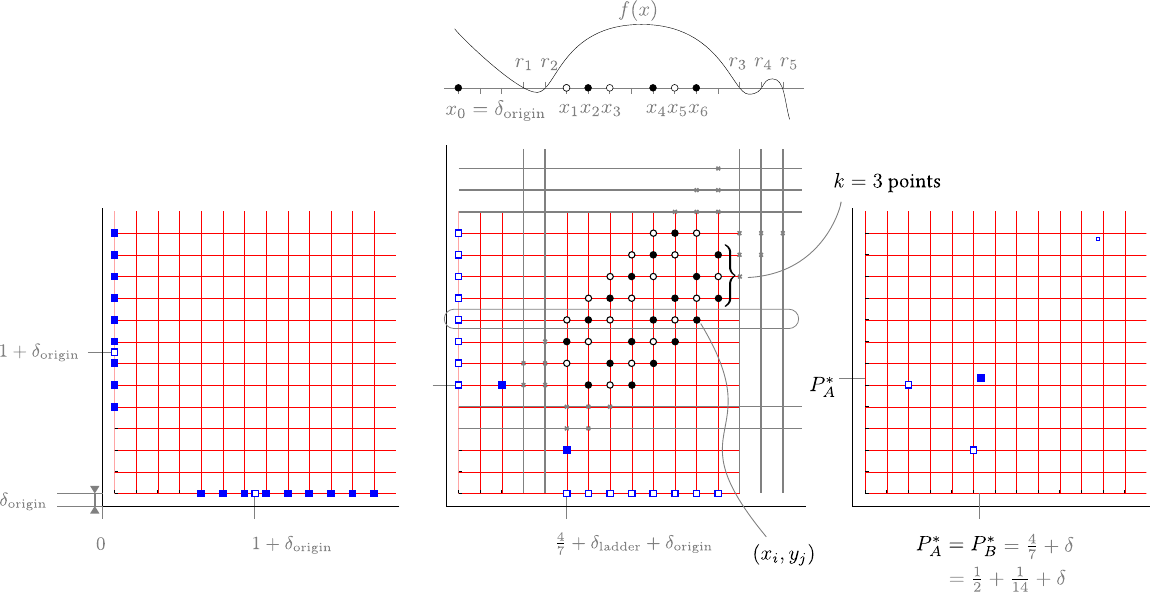}
    \par\end{centering}
  \caption{\label{fig:translated-TDPG-1by14}\atul{The $\dorigin$-translated TDPG (or equivalently, the reversed protocol)
      approaching bias $\epsilon(k=3) + \dorigin =1/14 + \dorigin $ may also be seen as proceeding in
      three stages, as illustrated by the three images (left to right). The difference compared to  \cref{fig:TDPG-1by14} is that everything is translated by $\dorigin$. In particular, the function $f'$ has been replaced with $f$, the coordinates $(x'_0,\dots x'_6)$ and the roots $(r'_1,\dots r'_5)$ have been replaced by $(x_0,\dots x_6)$ and $(r_1,\dots r_5)$ respectively. The split and the merge can be implemented just as before, while the ladder can be implemented, granted one can (effectively) solve the $f$-assignments (which are now all defined on positive coordinates) as illustrated above.}}

\end{figure}
\atul{However, as we explained after \cref{lem:generalMonomialDecomposition}, our construction only works when the coordinates involved in the $f$-assignment are positive (and not zero). Implementing the approach described there, we fix an arbitrarily small $\dorigin>0$ and consider a $\dorigin$-translated version Mochon's TIPG approaching bias $\epsilon(3) + \dorigin = 1/14 + \dorigin$. The translated TIPG is also valid (see \cref{lem:validAfterTranslating}).  \cref{fig:translated-TDPG-1by14} details how this translated TIPG can be viewed as three stages of the corresponding TDPG. The relevant $f$-assignment is (see \cref{lem:fAssignmentTranslation}) then defined on}
$\{x_{0},x_{1}\dots x_{6}\}$
where $x_{i}=x_{i}'+\dorigin$, with $f(x)=(r_{1}-x)(r_{2}-x)\dots(r_{5}-x)$
where $r_{i}=r_{i}'+\dorigin$ viz.
\begin{displaymath}
  t=\sum_{i=0}^{6}\frac{-f(x_{i})}{\prod_{j\neq i}(x_{j}-x_{i})}\left\llbracket x_{i}\right\rrbracket .
\end{displaymath}
We decompose $t$
into a sum of monomial assignments, i.e.
\begin{small}
  \begin{align*}
     & t  =\underbrace{\sum_{i=0}^{6}\frac{-r_{1}r_{2}r_{3}r_{4}r_{5}}{\prod_{j\neq i}(x_{j}-x_{i})}\left\llbracket x_{i}\right\rrbracket }_{\text{I}}+\underbrace{\sum_{i=0}^{6}\frac{-\overbrace{(r_{2}r_{3}r_{4}r_{5}+r_{1}r_{3}r_{4}r_{5}+r_{1}r_{2}r_{3}r_{5}+r_{1}r_{2}r_{3}r_{4})}^{:=\alpha_{1}}(-x_{i})}{\prod_{j\neq i}(x_{j}-x_{i})}\left\llbracket x_{i}\right\rrbracket }_{\text{II}}                                                                                                                                                                               \\
     & +\underbrace{\sum_{i=0}^{6}\frac{-\alpha_{2}(-x_{i})^{2}}{\prod_{j\neq i}(x_{j}-x_{i})}\left\llbracket x_{i}\right\rrbracket }_{\text{III}}+\underbrace{\sum_{i=0}^{6}\frac{-\alpha_{3}(-x_{i})^{3}}{\prod_{j\neq i}(x_{j}-x_{i})}\left\llbracket x_{i}\right\rrbracket }_{\text{IV}}+\underbrace{\sum_{i=0}^{6}\frac{-\alpha_{4}(-x_{i})^{4}}{\prod_{j\neq i}(x_{j}-x_{i})}\left\llbracket x_{i}\right\rrbracket }_{\text{V}}+\underbrace{\sum_{i=0}^{6}\frac{-\alpha_{5}(-x_{i})^{5}}{\prod_{j\neq i}(x_{j}-x_{i})}\left\llbracket x_{i}\right\rrbracket }_{\text{VI}},
  \end{align*}
\end{small}where $\alpha_{l}$ is the coefficient of $(-x)^{l}$ in $f(x)$.
Since the total number of points in each assignment are $7$, they
are unbalanced monomial assignments. Terms I, III and V each have
an even powered monomial therefore they correspond to the aligned
case. Their solutions, thus, are given in \cref{prop:ExactSolnUnbalancedMonomialAligned}.
Analogously, the remaining terms II, IV and VI each have an odd powered
monomial therefore they correspond to the misaligned case. Their solutions,
thus, given in \cref{prop:ExactSolnUnbalancedMonomialMisaligned}.

Let us now see how all these pieces fit together
to give the full protocol. We describe the procedure
in the language of TDPGs each step of which can be thought of as a short-hand to denote
an exchange and manipulation of qubits between Alice and Bob, granted that
the associated unitaries are known. As we have already done all the
hard work in finding these unitaries,\footnote{In this section we found the unitaries for $f$-assignments and in \cref{sec:TEF} we found those corresponding to splits and merges.} we can now proceed at this level of description. Concretely, the $\dorigin$-translated bias $1/14$ game (see  \cref{fig:TDPG-1by14}) \atul{proceeds} as follows:
\begin{enumerate}
  \item The first frame. This simply corresponds to the function $$\frac{1}{2}\left(\left\llbracket \dorigin,1 + \dorigin \right\rrbracket +\left\llbracket 1 + \dorigin,\dorigin \right\rrbracket \right).$$
  \item The split. \atul{Deposit weights along the axis as specified by the TIPG; more precisely, split the point $\left\llbracket \dorigin,1 + \dorigin\right\rrbracket $
          into a set of points along the $y$\textendash axis (offset by $\dorigin$) and analogously,
          split the point $\left\llbracket 1 + \dorigin ,\dorigin\right\rrbracket $ into a set of points
          along the $x$\textendash axis (offset by $\dorigin$), to match the distribution of points
          along the axes by the $\dorigin$-translated bias $1/14$ game.}
  \item The Catalyst State. \atul{Deposit a small amount of weight, $\delta_{\text{catalyst}}$,
        at all the points that appear in the translated TIPG. This can be done
        by raising  the points which are
        along the $y$\textendash axis, i.e. if the points along the axes
        are denoted as $\sum_{i}p_{\text{split},i}\left\llbracket \dorigin,y_{i}\right\rrbracket $,
        then raise them to obtain $\sum_{i}(p_{\text{split},i}-\delta_{\text{split},i})\left\llbracket \dorigin,y_{i}\right\rrbracket +\sum_{i,j}\delta_{\text{catalyst}}\left\llbracket x_{i},y_{j}\right\rrbracket $,
        where $\delta_{\text{catalyst}}>0$ can be chosen to be arbitrarily
        small and
        the second sum is over points $(x_{i},y_{j})$ which appear in the
        translated TIPG (excluding those along the axes (offset by $\dorigin$)\footnote{One needs to use the analogous procedure, i.e. use $\sum_{i}p_{\text{split},i}\left\llbracket x_{i},\dorigin\right\rrbracket $
        as well for the one point of the TIPG which has a $y$\textendash coordinate
        smaller than that of the points along the $y$\textendash axis.}). }
  \item The Ladder.
        \begin{enumerate}
          \item Denote the monomial decomposition of the valid functions by constituent
                valid functions. Globally scale these constituent valid functions
                sufficiently so that no negative weight appears when they are applied.
          \item Apply all the scaled down constituent horizontal valid functions.
          \item Apply all the scaled down constituent vertical valid functions.
          \item Repeat these two steps until all the weight has been transferred
                from the axes into the two final points of the ladder.\footnote{It would automatically become impossible to apply the moves once the
                  weights on the axes becomes sufficiently small.}
        \end{enumerate}
        The unitaries corresponding to these constituent valid functions
        correspond to the solutions of the monomial assignments.
  \item Raise and merge. Raise and merge the last two points into \atul{the point $$(1-\delta_{\text{sum-catalyst}})\left\llbracket \frac{4}{7}+\delta_{\textrm{ladder}}+\dorigin,\frac{4}{7}+\delta_{\textrm{ladder}}+\dorigin\right\rrbracket$$}where $\delta_{\text{sum-catalyst}}$ represents the total weight used by the catalyst,
        while $\delta_{\textrm{ladder}}$ comes from the truncation of the ladder.
        Then, using the method developed in the proof of \cref{thm:TIPG-to-valid-point-games} in \cite{ACG+14,Mochon07}, the catalyst state can be absorbed to
        obtain a single point $\left\llbracket \frac{4}{7}+\delta,\frac{4}{7}+\delta\right\rrbracket $.
        Thus, $P_{A}^{*}=P_{B}^{*}=\frac{1}{2}+\frac{1}{14}+\delta$, where
        $\delta$ can be made arbitrarily small by making the
        catalyst state smaller, the ladder longer and \atul{choosing a smaller translation parameter $\dorigin$}.
\end{enumerate}
The protocol is the reverse: it starts with a single point
corresponding to uncorrelated states and whose coordinates encode the
cheating probabilities, and ends with two points along the axis \atul{(offset by $\dorigin$)} with
equal weights, corresponding to the state $\frac{\left|AA\right\rangle +\left|BB\right\rangle }{\sqrt{2}}$.

\section{Future Work \label{sec:conclusions}}

Now that we have quantum WCF protocols, one can investigate questions about optimality, relaxation of underlying assumptions and connections to other cryptographic primitives.

\paragraph{Optimality} Various questions about the optimality of WCF protocols are unanswered.
\begin{itemize}
  \item \emph{Mochon's Games.} In \cref{sec:1by4k+2},  in order to find the solution to the $f$-assignment, we
        expressed it as a sum of monomial assignments; this yields an increase
        in dimensions, which in turn corresponds to an increase in the number
        of qubits required.\footnote{The dimension of the Hilbert space is expected to scale exponentially with the
          number of points involved in the $f$-assignment.}  One approach
        towards reducing this, could be to understand the connection between
        the perturbatively defined unitary from \cref{sec:TEF} and
        the exact one in \cref{sec:1by4k+2}, corresponding to the $1/10$-bias
        protocols. Another approach could be to try reducing the dimension
        using a standard technical lemma from \cite{Mochon07}, which is stated
        as \cref{lem:EBRMisCOF} here.
  \item \emph{Round complexity.} Recently, Miller \cite{Miller2019} established that round efficient (in terms of the bias) quantum WCF is impossible. However, unlike conventional security parameters (that must be taken to be large to have any practically relevant security), the security of quantum WCF is information theoretic, even for a fixed bias. Thus, it is conceivable that practical (in terms of round complexity) WCF protocols can be constructed for a fixed bias, say, 0.01. On the other hand, Miller's lower bound applies to TIPGs and there is scope for improvement by bounding the rounds needed to convert certain families of TIPGs to TDPGs.
  \item \emph{Pelchat-Høyer games.}  Pelchat and  Høyer \cite{pelchat13} proposed another family of TIPGs which achieve arbitrarily low bias as well. It will be interesting to see if an explicit WCF protocol can be obtained corresponding to these games, potentially, in fewer dimensions.
  \item \emph{Framework.} Constructing general tools to optimise and test the
        optimality of a TIPG for the number of points (and rounds, as mentioned above)
        in the associated TDPG would be very useful to both constructing better
        protocols as well as benchmarking the existing ones. For instance, we have a WCF protocol which uses constant space and approaches bias $\epsilon=1/6$. However, if we go lower and consider say a Mochon's next TIPG with bias $\epsilon=1/10$, then the corresponding TDPG suddenly seems to require points that tend to infinity as the TDPG approaches bias $1/10$. It is unclear whether this is an artefact of our construction or a fundamental characteristic.
\end{itemize}
\paragraph{Relaxing assumptions} The assumptions we made to obtain the protocols are not realistic.
\begin{itemize}
  \item \emph{System size.} The size of the incoming system containing the
        message is assumed to be known, however, this is hard to enforce physically.
        One possibility is to impose a more physically realistic constraint, such restricting the average energy in the fibre optic implementing the channel, as analysed in \cite{Himbeeck2017}.

  \item \emph{Noise.} Adding noise in a WCF protocol
        can cause a disagreement even when both parties are honest. It has
        been shown that in the absence of noise but in the present of losses, WCF can still be
        performed with a certain bias \cite{Berlin2009}. An interesting
        question is whether there exist lower bounds to the lossy but noiseless setting.
        Returning to noise, it is clear that quantum computation is realistic due to
        error correction.
        This, however, does not necessarily mean that WCF can be performed
        in such a setting, as it is not obvious how we can correct errors
        in this adversarial scenario without compromising the
        security. Thus, a systematic study of noise in the adversarial setting is crucial and recent techniques in this direction~\cite{Gutoski2018} may help.
  \item \emph{Device Dependence.} Device-independent WCF protocols have been suggested and involve the exchange
        of quantum boxes \cite{Aharon2014}.
        Their bias, however, remains quite high and since then, only modest progress has been reported \cite{ASH24}. Furthermore, no lower bound on the bias is known. The first step could
        be to redefine the protocol in a generalizable way; perhaps construct
        successively worse protocols---by, for instance, using fewer boxes---and subsequently, consider them as belonging to the same family.
        One could try to use PR-boxes or non-signalling
        boxes to understand the behaviour better. A complementary approach
        could be to construct the analogue of the Kitaev/Mochon framework
        where instead of qubits and unitaries, one studies more abstract
        objects which simulate the exchange of boxes and are only constrained
        by their statistics. Recently, WCF protocols were also considered in the context of  general probabilistic theories \cite{Sikora2019}, that are used to
        extend the impossibility results theories beyond quantum. They used
        conic duality which is the key point of Kitaev/Mochon frameworks and hence, this approach could be a starting
        point.
\end{itemize}
\paragraph{A fundamental connection} It is known that nearly perfect
WCF implies optimal strong coin flipping \cite{CK09}. Does this work the other way around?
This question may be more general than quantum, since the construction
in \cite{CK09} is purely classical. One way of proceeding could be
to try and construct optimal strong coin flipping protocols directly by adapting the Kitaev/Mochon
technique and using known, simpler protocols as a starting point.
The insight might not only help answer this question but also
yield another construction for nearly perfect WCF.

\section*{Acknowledgements}
This research was supported by the Belgian
Fonds de la Recherche Scientifique - FNRS, under grants
no R.50.05.18.F (QuantAlgo), R.8015.21 (QOPT) and O.0013.22 (EoS CHEQS). The QuantAlgo and QOPT projects have received funding from the QuantERA European Research Area Network (ERA-NET) Cofund in Quantum Technologies implemented within the European Union's Horizon 2020 program.
The most substantial part of the work was done while ASA was at the Universit{\' e} libre de Bruxelles, Belgium. He acknowledges support from the Belgian Fonds pour la Formation {\' a} la Recherche dans l'Industrie et dans l'Agriculture - FRIA, under grant no 1.E.081.17F. ASA acknowledges support from IQIM, an NSF Physics Frontier Center (GBMF-1250002), MURI grant FA9550-18-1-0161 and the U.S. Department of Defense through a QuICS Hartree Fellowship. Part of the work was carried out while ASA was visiting the Simons Institute for the Theory of Computing.
CV also acknowledges support from the FCT -- Funda\c c\~ao para a Ci\^encia e a Tecnologia through national funds FCT I.P. and, when eligible, co-funded by EU funds, under the
Scientific Employment Stimulus - Individual Call 2020.03274.CEECIND/CP1621/CT0003, under
project/support UIDB/50008/2025 -- Instituto de Telecomunicações with DOI identifier https://doi.org/10.54499/UID/50008/2025, and through the project PUFSeQure (2023.14154.PEX).
We are thankful to Nicolas Cerf, Mathieu Brandeho, Tom Van Himbeeck, Kishor Bharti, Stefano Pironio and Ognyan Oreshkov for various insightful discussions.

\vspace{3em}
\printbibliography
\appendix

\section{Proof of  Lemma \ref{lem:setequality}\label{sec:TEF-functions=00003DValid-functions=00003DclosureOfEBMfunctions}}

For the proof that the closure of EBM functions equals the set of valid functions, the reader is referred to \cite{ACG+14}. At the end of \cref{subsec:tdpgvalid} we also outlined the main arguments. Here, we prove the following:

\begin{lemma}
The closure of the set of EBM functions equals the set of TEF functions.\label{lem:closureEBMisTEF}
\end{lemma}

For simplicity, in the following discussion, we restrict to transitions (see \cref{def:transition}) with disjoint
support. This allows us to use transitions and functions interchangeably, as explained at the end of \cref{subsec:tdpgvalid}. 

The proof uses the following characterization of EBM functions presented in Lemma \ref{lem:EBRMisCOF}, which is originally due
 to Mochon \cite{Mochon07} (the proof therein had a minor
error, though, that we correct). 

Below, when we say EBM transition with spectrum in $[a,b]$, we refer
to an EBM transition with the additional constraint that the matrices
$H,G$,  as introduced in Definition \ref{def:EBMlineTransition}, have eigenvalues in the interval $[a,b]$. 
\begin{lemma}
 
Consider the transition $g\to h$ where $g:=\sum_{i=1}^{m}p_{g_{i}}\left\llbracket x_{g_{i}}\right\rrbracket $
and $h:=\sum_{i=1}^{m}p_{h_{i}}\left\llbracket x_{h_{i}}\right\rrbracket $.
For every EBM transition $g\to h$ with spectrum in $[a,b]$ there
exists a unitary matrix $U$, diagonal matrices $X_{h}$, $X_{g}$
(with no multiplicities except possibly those of $a$ and $b$) of
size at most $m+n-1$ such that 
\begin{equation}
U\underbrace{\left[\begin{array}{ccccc}
x_{g_{1}}\\
 & \ddots\\
 &  & x_{g_{n_{g}}}\\
 &  &  & a\\
 &  &  &  & \ddots
\end{array}\right]}_{:=X_{g}}U^{\dagger}\le\left[\begin{array}{ccccc}
x_{h_{1}}\\
 & \ddots\\
 &  & x_{h_{n_{h}}}\\
 &  &  & b\\
 &  &  &  & \ddots
\end{array}\right]=X_{h},\label{eq:conditionEBMcharacterised}
\end{equation}
and the vector $\left|\psi\right\rangle :=(\sqrt{p_{h_{1}}},\dots,\sqrt{p_{h_{n}}},0\dots0)^{T}=U(\sqrt{p_{g_{1}}},\dots\sqrt{p_{g_{m}}},0\dots0)^{T}$.
\label{lem:EBRMisCOF}
\end{lemma}

We will prove this lemma shortly. Let us first see how this almost immediately
yields Lemma \ref{lem:closureEBMisTEF}.
\begin{proof}[Proof Sketch of Lemma \ref{lem:closureEBMisTEF}]
 In this proof, we restrict to EBM functions with spectrum in $[a,b]\subseteq[0,\infty)$.
For any such EBM transition $g\to h$, one can verify that Equation~\eqref{eq:conditionEBMcharacterised}
implies the following (for any $b'\ge b$ and an appropriate $\tilde{U}$)

\begin{align*}
\tilde{U}\left[\begin{array}{ccc|ccc}
0 &  & \\
 & \ddots & \\
 &  & 0\\
\hline  &  &  & x_{g_{1}}\\
 &  &  &  & \ddots\\
 &  &  &  &  & x_{g_{m}}
\end{array}\right]\tilde{U}^{\dagger} & \le\\
\tilde{U}\left[\begin{array}{ccc|ccc}
a &  & \\
 & \ddots & \\
 &  & a\\
\hline  &  &  & x_{g_{1}}\\
 &  &  &  & \ddots\\
 &  &  &  &  & x_{g_{n}}
\end{array}\right]\tilde{U}^{\dagger} & \le\left[\begin{array}{ccc|ccc}
x_{h_{1}} &  & \\
 & \ddots & \\
 &  & x_{h_{m}}\\
\hline  &  &  & b\\
 &  &  &  & \ddots\\
 &  &  &  &  & b
\end{array}\right]\\
\le\left[\begin{array}{ccc|ccc}
1 &  & \\
 & \ddots & \\
 &  & 1\\
\hline  &  &  & b'\\
 &  &  &  & \ddots\\
 &  &  &  &  & b'
\end{array}\right] & \left[\begin{array}{ccc|ccc}
x_{h_{1}} &  & \\
 & \ddots & \\
 &  & x_{h_{m}}\\
\hline  &  &  & 1/b'\\
 &  &  &  & \ddots\\
 &  &  &  &  & 1/b'
\end{array}\right]\ \left[\begin{array}{ccc|ccc}
1 &  & \\
 & \ddots & \\
 &  & 1\\
\hline  &  &  & b'\\
 &  &  &  & \ddots\\
 &  &  &  &  & b'
\end{array}\right]
\end{align*}
\end{proof}

where the matrices are of size $m+n$. The inequality involving the
first and the last term may equivalently be expressed as 
\begin{align}
\left[\begin{array}{ccc|ccc}
1 &  & \\
 & \ddots & \\
 &  & 1\\
\hline  &  &  & 1/b'\\
 &  &  &  & \ddots\\
 &  &  &  &  & 1/b'
\end{array}\right]\tilde{U}\left[\begin{array}{ccc|ccc}
0 &  & \\
 & \ddots & \\
 &  & 0\\
\hline  &  &  & x_{g_{1}}\\
 &  &  &  & \ddots\\
 &  &  &  &  & x_{g_{m}}
\end{array}\right]\tilde{U}^{\dagger}\left[\begin{array}{ccc|ccc}
1 &  & \\
 & \ddots & \\
 &  & 1\\
\hline  &  &  & 1/b'\\
 &  &  &  & \ddots\\
 &  &  &  &  & 1/b'
\end{array}\right]\nonumber \\
\le\left[\begin{array}{ccc|ccc}
x_{h_{1}} &  & \\
 & \ddots & \\
 &  & x_{h_{m}}\\
\hline  &  &  & 1/b'\\
 &  &  &  & \ddots\\
 &  &  &  &  & 1/b'
\end{array}\right].\label{eq:noDivergence}
\end{align}
This condition yields, in the $b'\to\infty$ limit, 
\begin{equation}
E_{h}\tilde{U}\underbrace{\left(\sum_{i=1}^{n}x_{g_{i}}\left|g_{i}\right\rangle \left\langle g_{i}\right|\right)}_{:=G'}\tilde{U}^{\dagger}E_{h}\le\underbrace{\left(\sum_{i=1}^{m}x_{h_{i}}\left|h_{i}\right\rangle \left\langle h_{i}\right|\right)}_{:=H'}\label{eq:TEFcondition_noDivergence}
\end{equation}
where $(\left|g_{i}\right\rangle )_{i=1}^{n}$ represent the last
$n$ coordinates, $(\left|h_{i}\right\rangle )_{i=1}^{m}$ represent
the first $m$ coordinates and $E_{h}:=\sum_{i=1}^{m}\left|h_{i}\right\rangle \left\langle h_{i}\right|$.
Further, for $\left|v\right\rangle =\sum_{i=1}^{n}\sqrt{p_{g_{i}}}\left|g_{i}\right\rangle $,
one can check that $\tilde{U}\left|v\right\rangle =\sum_{i=1}^{m}\sqrt{p_{h_{i}}}\left|h_{i}\right\rangle $
using the definition of $\left|\psi\right\rangle $ and $\tilde{U}$.
Thus, any EBM transition $g\to h$ is also a TEF transition. 

One can easily extend this reasoning to establish that the closure
of EBM functions is also contained in the set of TEF functions. Consider
a sequence of EBM functions $(h_{i}-g_{i})_{i=1}^{\infty}$ with support
in $[a_{i},b_{i}]\subseteq[0,\infty)$ such that the limiting function, $h-g$
is well defined (i.e. support of $h-g$ is contained in $[0,\infty)$; support of a function $f$ is simply ${x:f(x)\neq 0}$) but $h-g$ is not EBM. The only way this can happen is if $b_{i}\to\infty$
tends to infinity as $i\to\infty$. However, using the reasoning above,
one can consider Equation~\eqref{eq:noDivergence} and there, it is clear that
the limiting procedure yields Equation~\eqref{eq:TEFcondition_noDivergence} which
is precisely the TEF constraint. Thus, the limiting function is a
TEF function.

One can similarly argue that every TEF function is contained in the
closure of EBM functions.

\begin{proof}[Proof of Lemma \ref{lem:EBRMisCOF}]
Let $n_{g}:=n$ and $n_{h}:=m$. An EBM entails that we are given
$G\le H$ with their spectrum in $[a,b]$ and a $\left|\psi\right\rangle $
such that 
\[
g=\text{Prob}[G,\left|\psi\right\rangle ]=\sum_{i=1}^{n_{g}}p_{g_{i}}\left\llbracket x_{g_{i}}\right\rrbracket 
\]
 and 
\[
h=\text{Prob}[H,\left|\psi\right\rangle ]=\sum_{i=1}^{n_{h}}p_{h_{i}}\left\llbracket x_{h_{i}}\right\rrbracket 
\]
with $p_{g_{i}},p_{h_{i}}>0$ and $x_{g_{i}}\neq x_{g_{j}}$, $x_{h_{i}}\neq x_{h_{j}}$
for $i\neq j$ but the dimension and multiplicities can be arbitrary.
First we show that one can always choose the eigenvectors $\left|g_{i}\right\rangle $
of $G$ with eigenvalue $x_{g_{i}}$ such that 
\[
\left|\psi\right\rangle =\sum_{i=1}^{n_{g}}\sqrt{p_{g_{i}}}\left|g_{i}\right\rangle .
\]
Consider $P_{g_{i}}$ to be the projector on the eigenspace with eigenvalue
$x_{g_{i}}$. Note that 
\[
\left|g_{i}\right\rangle :=\frac{P_{g_{i}}\left|\psi\right\rangle }{\sqrt{\left\langle \psi\right|P_{g_{i}}\left|\psi\right\rangle }}
\]
fits the bill. Similarly we choose/define $\left|h_{i}\right\rangle $
so that 
\[
\left|\psi\right\rangle =\sum_{i=1}^{n_{h}}\sqrt{p_{h_{i}}}\left|h_{i}\right\rangle .
\]
Consider now the projector onto the $\{\left|g_{i}\right\rangle \}$
space 
\[
\Pi_{g}=\sum_{i=1}^{n_{g}}\left|g_{i}\right\rangle \left\langle g_{i}\right|.
\]
Note that this will not have all eigenvectors with eigenvalues $\in\{x_{g_{i}}\}$.
Similarly we define 
\[
\Pi_{h}=\sum_{i=1}^{n_{h}}\left|h_{i}\right\rangle \left\langle h_{i}\right|.
\]
We further define $G':=\Pi_{g}G\Pi_{g}+a(\mathbb{I}-\Pi_{g})$ and
$H':=\Pi_{h}H\Pi_{h}+b(\mathbb{I}-\Pi_{h})$. These definitions are
useful as we can show 
\[
G'\le H'.
\]
From $G=\Pi_{g}G\Pi_{g}+(\mathbb{I}-\Pi_{g})G(\mathbb{I}-\Pi_{g})$
we can conclude that $\Pi_{g}G\Pi_{g}+a(\mathbb{I}-\Pi_{g})\le G$.
This entails $G'\le G$. Using a similar argument one can also establish
that $H\le H'$. Combining these we get $G'\le H'$. \\
Consider the projector 
\[
\Pi:=\text{projector on span}\{\{\left|g_{i}\right\rangle \}_{i=1}^{n_{g}},\{\left|h_{i}\right\rangle \}_{i=1}^{n_{h}}\}
\]
and note that this has at most $n_{g}+n_{h}-1$ dimension because
$\left|\psi\right\rangle $ lives in the span of $\{\left|g_{i}\right\rangle \}$
and in the span of $\{\left|h_{i}\right\rangle \}$ so one of the
basis vectors at least is not independent. Now note that 
\[
G'':=\Pi G'\Pi\le\Pi H'\Pi=:H''
\]
because we can always conjugate an inequality by a positive semi-definite
matrix on both sides. Note also that $\Pi\left|\psi\right\rangle =\left|\psi\right\rangle $
which means the matrices and the vectors have the claimed dimension.
We now establish that $\text{Prob}[H'',\left|\psi\right\rangle ]=h$
and $\text{Prob}[G'',\left|\psi\right\rangle ]=g$. For this we first
write the projector tailored to the $g$ basis as $\Pi=\Pi_{g}+\Pi_{g_{\perp}}$
where $\Pi_{g_{\perp}}$ is meant to enlarge the space to the $\text{span}\{h_{i}\}_{i=1}^{n_{h}}$.
With this we evaluate 
\begin{align*}
G'' & =\left(\Pi_{g}+\Pi_{g_{\perp}}\right)\left[\Pi_{g}G\Pi_{g}+a(\mathbb{I}-\Pi_{g})\right]\left(\Pi_{g}+\Pi_{g_{\perp}}\right)\\
 & =\Pi_{g}G\Pi_{g}+a\Pi_{g_{\perp}}.
\end{align*}
Manifestly then $\text{Prob}[G'',\left|\psi\right\rangle ]=g$. By
a similar argument one can establish the $h$ claim. Note that that
$G''$ and $H''$ have no multiplicities except possibly in $a$ and
$b$ respectively. Thus we conclude we can always restrict to the
claimed dimension and form.
\end{proof}

\section{Blink \texorpdfstring{$m\to n$}{m -> n} transition}\label{sec:Blinkered-transition}

\subsection{Completing an orthonormal basis}\label{subsec:completebasis}

Consider an orthonormal complete set of basis vectors $\left\{ \left|g_{i}\right\rangle \right\} $
and a vector $\left|v\right\rangle =\frac{\sum_{i}\sqrt{p_{i}}\left|g_{i}\right\rangle }{\sqrt{\sum_{i}p_{i}}}$.
We describe a scheme for constructing vectors $\left|v_{i}\right\rangle $
such that $\left\{ \left|v\right\rangle ,\left\{ \left|v_{i}\right\rangle \right\} \right\} $
is a complete orthonormal set of basis vectors. We can do
it inductively, but here instead we choose to do it by examples, as we believe it helps gain some
intuition and demonstrates the generalizable argument right away.
We define the first vector to be 
\begin{displaymath}
\left|v_{1}\right\rangle =\frac{\sqrt{p_{1}}\left|g_{1}\right\rangle -\frac{p_{1}}{\sqrt{p_{2}}}\left|g_{2}\right\rangle }{\sqrt{p_{1}+\frac{p_{1}^{2}}{p_{2}}}}=\frac{\sqrt{p_{1}}\left|g_{1}\right\rangle -\sqrt{p_{2}}\left|g_{2}\right\rangle }{\sqrt{p_{1}+p_{2}}},
\end{displaymath}
which is normalized and orthogonal to $\left|v\right\rangle $. The next
vector is 
\begin{displaymath}
\left|v_{2}\right\rangle =\frac{\sqrt{p_{1}}\left|g_{1}\right\rangle +\sqrt{p_{2}}\left|g_{2}\right\rangle -\frac{\left(p_{1}+p_{2}\right)}{\sqrt{p_{3}}}\left|g_{3}\right\rangle }{\sqrt{p_{1}+p_{2}+\frac{(p_{1}+p_{2})^{2}}{p_{3}}}}
\end{displaymath}
which is again normalized and orthogonal to $\left|v_{1}\right\rangle $.

Similarly we can construct the $\left(k+1\right)^{\text{th}}$ basis
vector as 
\begin{displaymath}
\left|v_{k}\right\rangle =\frac{\sum_{i=1}^{k}\sqrt{p_{k}}\left|g_{k}\right\rangle -\frac{\sum_{i=1}^{k}p_{k}}{\sqrt{p_{k+1}}}\left|g_{k+1}\right\rangle }{N_{k}},
\end{displaymath}
where $N_{k}=\sqrt{\sum_{i=1}^{k}p_{k}+\frac{(\sum_{i=1}^{k}p_{k})^{2}}{p_{k+1}}}$
and, thus, obtain the full set.

\subsection{Analysis of the \texorpdfstring{$3\rightarrow 2$}{3 -> 2} transition}
\label{subsec:appendix3to2}
Recall that the constraint equation is 
\begin{displaymath}
\underbrace{\sum x_{h_{i}}\left|h_{ii}\right\rangle \left\langle h_{ii}\right|}_{\text{I}}+\underbrace{x\mathbb{I}^{\{g_{ii}\}}}_{\text{II}}\ge\underbrace{\sum x_{g_{i}}U\left|g_{ii}\right\rangle \left\langle g_{ii}\right|U^{\dagger}}_{\text{III}},
\end{displaymath}
where we have introduced the notation $\left|h_{ii}\right\rangle =\left|h_{i}h_{i}\right\rangle $. The $g_{1},g_{2},g_{3}\to h_{1},h_{2}$
transition requires us to know 
\begin{displaymath}
U=\left|v\right\rangle \left\langle w\right|+\left|w\right\rangle \left\langle v\right|+\left|v_{1}\right\rangle \left\langle v_{1}\right|+\left|v_{2}\right\rangle \left\langle v_{2}\right|+\left|w_{1}\right\rangle \left\langle w_{1}\right|.
\end{displaymath}
Using the procedure above we can evaluate the vectors of interest as
\begin{align*}
&\left|v\right\rangle   =\frac{\sqrt{p_{g_{1}}}\left|g_{11}\right\rangle +\sqrt{p_{g_{2}}}\left|g_{22}\right\rangle +\sqrt{p_{g_{3}}}\left|g_{33}\right\rangle }{N_{g}},\quad
\left|v_{1}\right\rangle   =\frac{\sqrt{p_{g_{1}}}\left|g_{11}\right\rangle -\frac{p_{g_{1}}}{\sqrt{p_{g_{2}}}}\left|g_{22}\right\rangle }{N_{g_{1}}},\\
&\left|v_{2}\right\rangle   =\frac{\sqrt{p_{g_{1}}}\left|g_{11}\right\rangle +\sqrt{p_{g_{2}}}\left|g_{22}\right\rangle -\frac{\left(p_{g_{1}}+p_{g_{2}}\right)}{\sqrt{p_{g_{3}}}}\left|g_{33}\right\rangle }{N_{g_{2}}},\\
&\left|w\right\rangle   =\frac{\sqrt{p_{h_{1}}}\left|h_{11}\right\rangle +\sqrt{p_{h_{2}}}\left|h_{22}\right\rangle }{N_{h}}\quad\text{ and }\quad
\left|w_{1}\right\rangle   =\frac{\sqrt{p_{h_{2}}}\left|h_{11}\right\rangle -\sqrt{p_{h_{1}}}\left|h_{22}\right\rangle }{N_{h}},
\end{align*}
where $N_{g},\,N_{g_{1}},\,N_{g_{2}},\,N_{h}$ are normalization factors.
In fact we want to express the constraints in this basis, and to evaluate
the first term of the LHS in the constraint equation we use the above to find 
\begin{align*}
\left|h_{11}\right\rangle   =\frac{\sqrt{p_{h_{1}}}\left|w\right\rangle +\sqrt{p_{h_{2}}}\left|w_{1}\right\rangle }{N_{h}}\quad \text{ and }\quad
\left|h_{22}\right\rangle   =\frac{\sqrt{p_{h_{2}}}\left|w\right\rangle -\sqrt{p_{h_{1}}}\left|w_{1}\right\rangle }{N_{h}},
\end{align*}
which leads to 
\begin{align*}
&\text{I}  =x_{h_{1}}\left|h_{11}\right\rangle \left\langle h_{11}\right|+x_{h_{2}}\left|h_{22}\right\rangle \left\langle h_{22}\right|\\
&=\frac{1}{N_{h}^{2}}\left[\begin{array}{c|cc}
& \left\langle w\right| & \left\langle w_{1}\right|\\
\hline \left|w\right\rangle  & p_{h_{1}}x_{h_{1}}+p_{h_{2}}x_{h_{2}} & \sqrt{p_{h_{1}}p_{h_{2}}}(x_{h_{1}}-x_{h_{2}})\\
\left|w_{1}\right\rangle  & \sqrt{p_{h_{1}}p_{h_{2}}}(x_{h_{1}}-x_{h_{2}}) & p_{h_{2}}x_{h_{1}}+p_{h_{1}}x_{h_{2}}
\end{array}\right].
\end{align*}

Evaluation
of II is nearly trivial after expressing the identity in this basis
\begin{align*}
\text{II}  =x(\left|v\right\rangle \left\langle v\right|+\left|v_{1}\right\rangle \left\langle v_{1}\right|+\left|v_{2}\right\rangle \left\langle v_{2}\right|)=\left[\begin{array}{c|ccc}
& \left\langle v\right| & \left\langle v_{1}\right| & \left\langle v_{2}\right|\\
\hline \left|v\right\rangle  & x\\
\left|v_{1}\right\rangle  &  & x\\
\left|v_{2}\right\rangle  &  &  & x
\end{array}\right].
\end{align*}
For the last term 
$
\text{III}=\underbrace{x_{g_{1}}U\left|g_{11}\right\rangle \left\langle g_{11}\right|U^{\dagger}}_{\text{(i)}}+\underbrace{x_{g_{2}}U\left|g_{22}\right\rangle \left\langle g_{22}\right|U^{\dagger}}_{\text{(ii)}}+\underbrace{x_{g_{3}}U\left|g_{33}\right\rangle \left\langle g_{33}\right|U^{\dagger}}_{\text{(iii)}}
$, 
we evaluate 
\begin{align*}
&U\left|g_{11}\right\rangle   =\frac{\sqrt{p_{g_{1}}}}{N_{g}}\left|w\right\rangle +\frac{\sqrt{p_{g_{1}}}}{N_{g_{1}}}\left|v_{1}\right\rangle +\frac{\sqrt{p_{g_{1}}}}{N_{g_{2}}}\left|v_{2}\right\rangle, \\
&U\left|g_{22}\right\rangle   =\frac{\sqrt{p_{g_{2}}}}{N_{g}}\left|w\right\rangle +\frac{\left(-\frac{p_{g_{1}}}{\sqrt{p_{g_{2}}}}\right)}{N_{g_{1}}}\left|v_{1}\right\rangle +\frac{\sqrt{p_{g_{2}}}}{N_{g_{2}}}\left|v_{2}\right\rangle \text{ and } \\
&U\left|g_{33}\right\rangle   =\frac{\sqrt{p_{g_{3}}}}{N_{g}}\left|w\right\rangle +0\left|v_{1}\right\rangle +\frac{\left(-\frac{p_{g_{1}}+g_{g_{2}}}{\sqrt{p_{g_{3}}}}\right)}{N_{g_{2}}}\left|v_{2}\right\rangle .
\end{align*}
\begin{displaymath}\text{For the first term we have }
\text{(i)}=x_{g_{1}}p_{g_{1}}\left[\begin{array}{c|ccc}
& \left\langle v_{1}\right| & \left\langle v_{2}\right| & \left\langle w\right|\\
\hline \left|v_{1}\right\rangle  & \frac{1}{N_{g_{1}}^{2}} & \frac{1}{N_{g_{1}}N_{g_{2}}} & \frac{1}{N_{g_{1}}N_{g}}\\
\left|v_{2}\right\rangle  & \frac{1}{N_{g_{2}}N_{g_{1}}} & \frac{1}{N_{g_{2}}^{2}} & \frac{1}{N_{g_{2}}N_{g}}\\
\left|w\right\rangle  & \frac{1}{N_{g}N_{g_{1}}} & \frac{1}{N_{g}N_{g_{2}}} & \frac{1}{N_{g}^{2}}
\end{array}\right].
\end{displaymath}
For the second term, we re-write $U\left|g_{22}\right\rangle =\sqrt{p_{g_{2}}}\left(\frac{1}{N_{g}}\left|w\right\rangle -\frac{1}{N'_{g_{1}}}\left|v_{1}\right\rangle +\frac{1}{N_{g_{2}}}\left|v_{2}\right\rangle \right)$
with 
$
N'_{g_{1}}=\frac{p_{g_{2}}}{p_{g_{1}}}N_{g_{1}}
$, \begin{displaymath}
\text{to obtain  (ii)}=x_{g_{2}}p_{g_{2}}\left[\begin{array}{c|ccc}
& \left\langle v_{1}\right| & \left\langle v_{2}\right| & \left\langle w\right|\\
\hline \left|v_{1}\right\rangle  & \frac{1}{N_{g_{1}}^{\prime2}} & -\frac{1}{N'_{g_{1}}N_{g_{2}}} & -\frac{1}{N'_{g_{1}}N_{g}}\\
\left|v_{2}\right\rangle  & -\frac{1}{N_{g_{2}}N'_{g_{1}}} & \frac{1}{N_{g_{2}}^{2}} & \frac{1}{N_{g_{2}}N_{g}}\\
\left|w\right\rangle  & -\frac{1}{N_{g}N'_{g_{1}}} & \frac{1}{N_{g}N_{g_{2}}} & \frac{1}{N_{g}^{2}}
\end{array}\right],
\end{displaymath}
and finally $U\left|g_{33}\right\rangle =\sqrt{p_{g_{3}}}\left(\frac{1}{N_{g}}\left|w\right\rangle +0\left|v_{1}\right\rangle -\frac{1}{N'_{g_{2}}}\left|v_{2}\right\rangle \right)$
with 
$
N'_{g_{2}}=\frac{p_{g_{3}}}{p_{g_{1}}+p_{g_{2}}}
$,\begin{displaymath}
\text{to get  (iii)}=x_{g_{3}}p_{g_{3}}\left[\begin{array}{c|ccc}
& \left\langle v_{1}\right| & \left\langle v_{2}\right| & \left\langle w\right|\\
\hline \left|v_{1}\right\rangle \\
\left|v_{2}\right\rangle  &  & \frac{1}{N_{g_{2}}^{\prime2}} & -\frac{1}{N'_{g_{2}}N_{g}}\\
\left|w\right\rangle  &  & -\frac{1}{N_{g}N'_{g_{2}}} & \frac{1}{N_{g}^{2}}
\end{array}\right].
\end{displaymath}
Now we can combine all of these into a single matrix and try to obtain
some simpler constraints.

\noindent
{\footnotesize%
	\begin{minipage}[t]{1\columnwidth}%
		\begin{displaymath}
		M\overset{\text{def}}{=}\left[\begin{array}{c|ccccc}
		& \left\langle v\right| & \left\langle v_{1}\right| & \left\langle v_{2}\right| & \left\langle w\right| & \left\langle w_{1}\right|\\
		\hline \left|v\right\rangle  & x\\
		\left|v_{1}\right\rangle  &  & x-\frac{x_{g_{1}}p_{g_{1}}}{N_{8_{1}}^{2}}-\frac{x_{g_{2}}p_{g_{2}}}{N_{g_{1}}^{\prime2}} & -\frac{x_{g_{1}}p_{g_{1}}}{N_{g_{1}}N_{g_{2}}}+\frac{x_{g_{2}}p_{g_{2}}}{N_{g_{1}}'N_{g_{2}}} & -\frac{x_{g_{1}}p_{g_{1}}}{N_{g_{1}}N_{g}}+\frac{x_{g_{2}}p_{g_{2}}}{N_{g_{1}}'N_{g}}\\
		\left|v_{2}\right\rangle  &  & -\frac{x_{g_{1}}p_{g_{1}}}{N_{g_{2}}N_{g_{1}}}+\frac{x_{g_{2}}p_{g_{2}}}{N_{g_{2}}N'_{g_{1}}} & x-\frac{x_{g_{1}}p_{g_{1}}}{N_{g_{2}}^{2}}-\frac{x_{g_{2}}p_{g_{2}}}{N_{g_{2}}^{2}}-\frac{x_{g_{3}}p_{g_{3}}}{N_{g_{2}}^{\prime2}} & -\frac{x_{g_{1}}p_{g_{1}}}{N_{g_{2}}N_{g}}-\frac{x_{g_{2}}p_{g_{2}}}{N_{g_{2}}N_{g}}+\frac{x_{g_{3}}p_{g_{3}}}{N_{g_{2}}'N_{g}}\\
		\left|w\right\rangle  &  & -\frac{x_{g_{1}}p_{g_{1}}}{N_{g}N_{g_{1}}}+\frac{x_{g_{2}}p_{g_{2}}}{N_{g}N_{g_{1}}'} & -\frac{x_{g_{1}}p_{g_{1}}}{N_{g}N_{g_{2}}}-\frac{x_{g_{2}}p_{g_{2}}}{N_{g}N_{g_{2}}}+\frac{x_{g_{3}}p_{g_{3}}}{N_{g}N'_{g_{2}}} & \frac{p_{h_{1}}x_{h_{1}}+p_{h_{2}}x_{h_{2}}}{N_{h}^{2}}-\frac{1}{N_{g}^{2}}\sum_{i}x_{g_{i}}p_{g_{i}} & \frac{\sqrt{p_{h_{1}}p_{h_{2}}}}{N_{h}^{2}}(x_{h_{1}}-x_{h_{2}})\\
		\left|w_{1}\right\rangle  &  &  &  & \frac{\sqrt{p_{h_{1}}p_{h_{2}}}}{N_{h}^{2}}(x_{h_{1}}-x_{h_{2}}) & \frac{p_{h_{2}}x_{h_{1}}+p_{h_{1}}x_{h_{2}}}{N_{h}^{2}}
		\end{array}\right]\ge0.
		\end{displaymath}
	\end{minipage}
	
}

Despite this appearing to be a complicated expression, we can conclude
that it is always so that the larger $x$ is the looser is the
constraint. To show this and simplify the calculation, note that
$M$ can be split into a scalar condition, $x\ge0$ -- from the $\left|v\right\rangle \left\langle v\right|$
part -- and a sub-matrix which we choose to write as 
\begin{displaymath}
\begin{array}{c|c|c}
& \begin{array}{cc}
\left\langle v_{1}\right| & \left\langle v_{2}\right|\end{array} & \begin{array}{cc}
\left\langle w\right| & \left\langle w_{1}\right|\end{array}\\
\hline \begin{array}{c}
\left|v_{1}\right\rangle \\
\left|v_{2}\right\rangle 
\end{array} & C & B^{T}\\
\hline \begin{array}{c}
\left|w\right\rangle \\
\left|w_{1}\right\rangle 
\end{array} & B & A
\end{array}\ge0.
\end{displaymath}
We have that $\left[\begin{array}{cc}
C & B^{T}\\
B & A
\end{array}\right]\ge0\iff\left[\begin{array}{cc}
A & B\\
B^{T} & C
\end{array}\right]\ge0\iff C\ge0,\,A-BC^{-1}B^{T}\ge0,\,(\mathbb{I}-CC^{-1})B^{T}=0$, using Shur's Complement condition for positivity where $C^{-1}$
is the generalized inverse. We can take $x$ to be sufficiently large so that $C>0$ and thereby
make sure that $\mathbb{I}-CC^{-1}=0$. Then, the only condition
of interest is 
\begin{displaymath}
A-BC^{-1}B^{T}\ge0.
\end{displaymath}
Actually, we can do even better than this. Note that if $C>0$ then
$C^{-1}>0$ and that the second term is of the form 
\begin{displaymath}
\underbrace{\left[\begin{array}{cc}
	a & b\\
	0 & 0
	\end{array}\right]}_{B}\underbrace{\left[\begin{array}{cc}
	\alpha & \gamma\\
	\gamma & \beta
	\end{array}\right]}_{C^{-1}}\underbrace{\left[\begin{array}{cc}
	a & 0\\
	b & 0
	\end{array}\right]}_{B^{T}}=\left[\begin{array}{cc}
\left[\begin{array}{cc}
a & b\end{array}\right]\left[\begin{array}{cc}
\alpha & \gamma\\
\gamma & \beta
\end{array}\right]\left[\begin{array}{c}
a\\
b
\end{array}\right] & 0\\
0 & 0
\end{array}\right]\ge0,
\end{displaymath}
because $C^{-1}>0$. We can therefore write the constraint equation
as
$
A\ge BC^{-1}B^{T}\ge0
$
and note that $A\ge0$ is a necessary condition. This also becomes
a sufficient condition in the limit that $x\to\infty$ because $C^{-1}\to0$
in that case. Thus, we have reduced the analysis to simply checking
if 
\begin{displaymath}
\left[\begin{array}{cc}
\frac{p_{h_{1}}x_{h_{1}}+p_{h_{2}}x_{h_{2}}}{N_{h}^{2}}-\frac{1}{N_{g}^{2}}\sum_{i}x_{g_{i}}p_{g_{i}} & \frac{\sqrt{p_{h_{1}}p_{h_{2}}}}{N_{h}^{2}}(x_{h_{1}}-x_{h_{2}})\\
\frac{\sqrt{p_{h_{1}}p_{h_{2}}}}{N_{h}^{2}}(x_{h_{1}}-x_{h_{2}}) & \frac{p_{h_{2}}x_{h_{1}}+p_{h_{1}}x_{h_{2}}}{N_{h}^{2}}
\end{array}\right]\ge0.
\end{displaymath}
This is a $2\times2$ matrix and can be checked for positivity using the
trace and determinant method or we can use again Schur's
Complement conditions. Here, however, we intend to use a more
general technique.
Let us introduce 
\begin{displaymath}
\left\langle x_{g}\right\rangle \overset{\text{def}}{=}\frac{1}{N_{g}^{2}}\sum_{i}x_{g_{i}}p_{g_{i}},\,\left\langle \frac{1}{x_{h}}\right\rangle \overset{\text{def}}{=}\frac{1}{N_{h}^{2}}\sum_{i}\frac{p_{h_{i}}}{x_{h_{i}}}.
\end{displaymath}
Term (I) and one element from term (III) constitute a
matrix $A$ which can be written as \begin{small}
\begin{align*}
A  =x_{h_{1}}\left|h_{11}\right\rangle \left\langle h_{11}\right|+x_{h_{2}}\left|h_{22}\right\rangle \left\langle h_{22}\right|-\left\langle x_{g}\right\rangle \left|w\right\rangle \left\langle w\right| =\begin{array}{c|cc}
& \left\langle h_{11}\right| & \left\langle h_{22}\right|\\
\hline \left|h_{11}\right\rangle  & x_{h_{1}}\\
\left|h_{22}\right\rangle  &  & x_{h_{2}}
\end{array}-\left\langle x_{g}\right\rangle \left|w\right\rangle \left\langle w\right|.
\end{align*}
\end{small}
We use 
$F-M\ge0\iff\mathbb{I}-\sqrt{F}^{-1}M\sqrt{F}^{-1}\ge0$ for $F>0$, to obtain $\mathbb{I}\ge\left\langle x_{g}\right\rangle \left|w''\right\rangle \left\langle w''\right|$,
where $\left|w''\right\rangle =\frac{\sqrt{\frac{p_{h_{1}}}{x_{h_{1}}}}\left|h_{11}\right\rangle +\sqrt{\frac{p_{h_{2}}}{x_{h_{2}}}}\left|h_{22}\right\rangle }{N_{h}}$.
Normalizing this we get $\left|w'\right\rangle =\frac{\left|w''\right\rangle }{\sqrt{\left\langle \frac{1}{x_{h}}\right\rangle }}$
which entails $\mathbb{I}\ge\left\langle x_{g}\right\rangle \left\langle \frac{1}{x_{h}}\right\rangle \left|w'\right\rangle \left\langle w'\right|$
and that leads us to the final condition 
$
\frac{1}{\left\langle x_{g}\right\rangle }\ge\left\langle \frac{1}{x_{h}}\right\rangle .
$

In fact all the techniques used in reaching this result can be extended
to the $m\to n$ transition case as well and so the aforesaid result
holds in general.

\section{Approaching bias  \texorpdfstring{$\epsilon(k)=1/(4k+2)$}{epsilon(k)=1/(4k+2)}}
\label{sec:appendalgebraicanalytic}
\begin{lemma}
	\label{lem:spanningLemma}Consider an $n$-dimensional vector space, a diagonal matrix $X=\diag(x_{1},x_{2}\dots x_{n})$ and a vector
	$\left|c\right\rangle =(c_{1},c_{2}\dots,c_{n})$ where all the $x_{i}$s
	are distinct and all the $c_{i}$ are non-zero. Then, the vectors $\left|c\right\rangle ,X\left|c\right\rangle ,\dots X^{n-1}\left|c\right\rangle $
	span the vector space.
\end{lemma}

\begin{proof}
	We write the vectors as 
	\begin{displaymath}
	\left|\tilde{w}_{i}\right\rangle =X^{i-1}\left|c\right\rangle =\left[\begin{array}{c}
	x_{1}^{i-1}c_{1}\\
	x_{2}^{i-1}c_{2}\\
	\vdots\\
	x_{n}^{i-1}c_{n}
	\end{array}\right].
	\end{displaymath}
	We show that the set of vectors are linearly independent, which is
	equivalent to showing that the determinant of the matrix containing
	the vectors as rows (or equivalently as columns) is non-zero, i.e.
	\begin{displaymath}
	\det\left(\underbrace{\left[\begin{array}{ccccc}
		1 & 1 & \dots &  & 1\\
		x_{1} & x_{2} &  &  & x_{n}\\
		x_{1}^{2} & x_{2}^{2} &  &  & x_{n}^{2}\\
		\vdots &  & \ddots\\
		x_{1}^{n-1} & x_{2}^{n-1} & \dots &  & x_{n}^{n-1}
		\end{array}\right]}_{:=\tilde{X}}\left[\begin{array}{ccccc}
	c_{1}\\
	& c_{2}\\
	\\
	&  &  & \ddots\\
	&  &  &  & c_{n}
	\end{array}\right]\right)=c_{1}\cdot c_{2}\cdot\dots c_{n}\cdot\det\tilde{X}
	\end{displaymath}
	is non-zero. Notice that $\tilde{X}$ is the so-called
	Vandermonde matrix (restricted to being a square matrix) and its determinant,
	known as the Vandermonde determinant, is $\det(\tilde{X})=\prod_{1\le i\le j\le n}(x_{j}-x_{i})\neq0$
	as $x_{i}$s are distinct. As $c_{i}$s are all non-negative our
	proof is complete.
\end{proof}
\subsection{Proof of \ref{lem:expectationLemma}}\label{subsec:proofexpectationlemma}
In our proof we will need the following \ref{lem:fAssignmentLemma}, which gives a property of the $f-$assignments. 
\begin{lemma}
	\label{lem:fAssignmentLemma} $\sum_{i=1}^{n}\frac{f(x_{i})}{\prod_{j\neq i}(x_{j}-x_{i})}=0$
	where $f(x_{i})$ is a polynomial of order $k\le n-2$ where $x_{i}\in\mathbb{R}$
	are distinct.
\end{lemma}
The proof can be found in \cite{Mochon07,ACG+14}. 

\begin{proof}[Proof of \ref{lem:expectationLemma}]
	The equality $\left\langle x^{k}\right\rangle =0$ for $k\leq n-2$ is a direct consequence of \ref{lem:fAssignmentLemma}, and we proceed to  prove the inequality $\left\langle x^{n-1}\right\rangle >0$. 
	Suppose for now that (we prove it in the end)
	\begin{equation}
		\sum_{i=1}^n\frac{x_i^{n-1}}{\prod_{j\neq i}(x_j-x_i)} = (-1)^{n-1}. 
		\label{eq:theMinusOnePow}
	\end{equation}
	Define $p(x_{i})=\frac{-(-x_{i})^{m}}{\prod_{j\neq i}(x_{j}-x_{i})}$ so that $t=\sum_{i}p(x_{i})\left\llbracket x_{i}\right\rrbracket $.
	Observe that 
	\begin{align*}
	\left\langle x^{n-1}\right\rangle  & =\sum_{i}x_{i}^{n-m-1}p(x_{i})\\
	& =\sum_{i}(-1)^{m}x_{i}^{n-1}\frac{-1}{\prod_{j\neq i}(x_{j}-x_{i})}\\
	& =(-1)^{m}(-1)\sum_{i}\frac{x_{i}^{n-1}}{\prod_{j\neq i}(x_{j}-x_{i})}\\
	& =(-1)^{m}(-1)(-1)^{n-1}=(-1)^{m+n}
	\end{align*}
	where we used Equation~Equation~\eqref{eq:theMinusOnePow}.

	It remains to prove Equation~Equation~\eqref{eq:theMinusOnePow}. We show that $d(n)=\sum_{i=1}^n\frac{x_i^{n-1}}{\prod_{j\neq i}(x_j-x_i)}=(-1)^{n-1}$ by induction. %
	The base of  the induction gives us $d(2)=\frac{x_1}{x_2-x_1}+\frac{x_2}{x_1-x_2}=-1$. We continue by assuming that it holds for $d(n)$ and take
	\begin{align*}
		d(n+1)&=\sum_{i=1}^{n+1}\frac{x_i^n}{\prod_{j\neq i}(x_j-x_i)}=\sum_{i=1}^{n+1}\frac{-(x_{n+1}-x_i)x_i^{n-1}+x_{n+1}x_i^{n-1}}{\prod_{j\neq i}(x_j-x_i)}\\
	&=-\sum_{i=1}^{n+1}(x_{n+1}-x_i)\frac{x_i^{n-1}}{\prod_{j\neq i}(x_j-x_i)}+x_{n+1}\underbrace{\sum_{i=1}^{n+1}\frac{x_i^{n-1}}{\prod_{j\neq i}(x_j-x_i)}}_{\text{$=0$, from \ref{lem:fAssignmentLemma}}}\\
	&=-\sum_{i=1}^{n}\frac{x_{n+1}-x_i}{x_{n+1}-x_i}\frac{x_i^{n-1}}{\prod_{j\neq i,n+1}(x_j-x_i)}+(x_{n+1}-x_{n+1})\frac{x_{n+1}^{n-1}}{\prod_{j\neq n+1}(x_j-x_{n+1})}=-d(n).
	\end{align*}
	This completes the proof.
\end{proof}

\subsection{Restricted decomposition into \texorpdfstring{$f_{0}$}{f0}-assignments}\label{subsec:Restricted-decomposition-into-monomials}

The monomial decomposition we presented in \cref{subsec:fassignmentequivmonomial} is not unique. Here, we give another useful
decomposition that, however, only works in a restricted case; that is when the roots of $f$ are right roots, as described below.
\begin{lemma}[$f$ with right roots to $f_{0}$]
	Consider a set of real coordinates satisfying $0<x_{1}<x_{2}\dots<x_{n}$
	and let $f(x)=(r_{1}-x)(r_{2}-x)\dots(r_{k}-x)$ where $k\le n-2$
	and the roots $\{r_{i}\}_{i=1}^{k}$ of $f$ are right roots, i.e.
	they are such that for every root $r_{i}$ there exists a distinct
	coordinate $x_{j}<r_{i}$. Let $t=\sum_{i=1}^{n}p_{i}\left\llbracket x_{i}\right\rrbracket $
	be the corresponding $f$-assignment. Then, there exist $f_{0}$-assignments, $\{t_{0;j}\}$,
	on a subset of $(x_{1},x_{2}\dots x_{n})$, such that $t=\sum_{i=1}^{m}\alpha_{i}t_{0;i}$
	where $\alpha_{i}>0$ is a real number and $m>0$ is an integer. \label{lem:rightRoots}
\end{lemma}

\begin{proof}
	For simplicity, assume that $x_{i}<r_{i}, \ \forall i$, but the
	argument works in general. We can, then, write
	\begin{align*}
	t & =\sum_{i=1}^{n}\frac{-f(x_{i})}{\prod_{j\neq i}(x_{j}-x_{i})}\left\llbracket x_{i}\right\rrbracket \\
	& =\sum_{i=1}^{n}\left(\frac{-(r_{1}-x_{1})(r_{2}-x_{i})\dots(r_{k}-x_{i})}{\prod_{j\neq i}(x_{j}-x_{i})}+\frac{-(x_{1}-x_{i})(r_{2}-x_{i})\dots(r_{k}-x_{i})}{\prod_{j\neq i}(x_{j}-x_{i})}\right)\left\llbracket x_{i}\right\rrbracket \\
	& =(r_{1}-x_{1})\sum_{i=1}^{n}\frac{-(r_{2}-x_{i})\dots(r_{k}-x_{i})}{\prod_{j\neq i}(x_{j}-x_{i})}\left\llbracket x_{i}\right\rrbracket +\sum_{i=2}^{n}\frac{-(r_{2}-x_{i})\dots(r_{k}-x_{i})}{\prod_{j\neq i,1}(x_{j}-x_{i})}\left\llbracket x_{i}\right\rrbracket ,
	\end{align*}
	where the first term has the same form that we started with (except
	for a positive constant which is irrelevant for the EBM/ validity condition, see Proposition \ref{prop:ebmvalid}) but with the polynomial having one less
	degree. The second term also has the same form, except that the number
	of points involved has been reduced. Note how this process relies
	crucially on the fact that $r_{1}-x_{1}>0$; otherwise the term
	on the left would, by itself, not correspond to a valid move. This
	process can be repeated until we obtain a sum of $f_{0}$-assignments
	on various subsets of $(x_{1},x_{2}\dots x_{n})$.
\end{proof}

The advantage of this decomposition is that we can immediately apply
it to the $f$-assignment of the bias-$1/10$ game. This is relevant because constructing solutions to $f_{0}$-assignments
is relatively easy and so they, together with this result, allow us
to derive the $1/10$ bias protocol circumventing the perturbative
approach that we used in \cref{sec:TEF}.

\begin{figure}[h]
	\begin{centering}
		\includegraphics[scale=0.95]{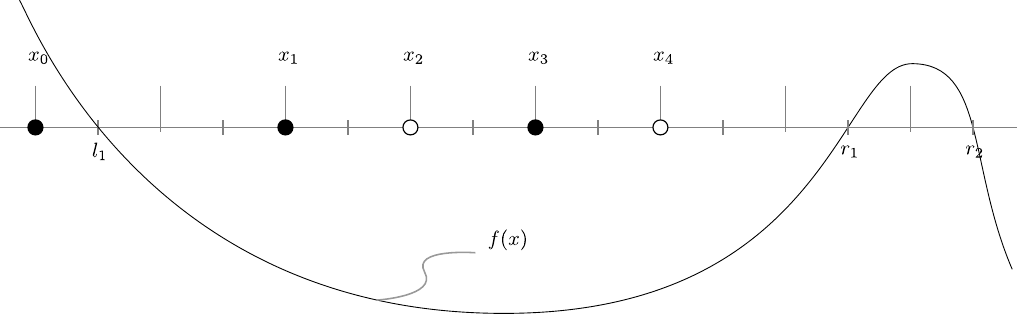}
		\par\end{centering}
	\caption{The main $1/10$ move involves $n=5$ points. $f$ has $k=3$ roots,
		all of which are right roots.}
\end{figure}

\begin{example}[The main $1/10$ move.]
	The key move in the $1/10$-bias point game has its coordinates given
	by $x_{0},x_{1},x_{2},x_{3},x_{4}$ and roots given by $l_{1},r_{1},r_{2}$
	which satisfy $x_{0}<l_{1}<x_{1}<x_{2}<x_{3}<x_{4}<r_{1}<r_{2}$.
	Each root is a right root here because $x_{0}<l_{1}$, $x_{3}<r_{1}$,
	$x_{4}<r_{2}$. Hence, from \ref{lem:rightRoots}, this assignment can be expressed
	as a combination of $f_{0}$-assignments defined over subsets of the
	initial set of coordinates and each $f_{0}$-assignment admits a simple
	solution given by \cref{prop:balancedf0algeb} and \cref{prop:unbalancedf0algeb} .\label{exa:1by10move}
\end{example}

Another simple example is the class of $f$-assignments describing
merge moves (see \cref{exa:merge}). We place the roots of $f$ in such a way that all points,
except one, have negative weights.
\begin{center}
	\begin{figure}[h]
		\begin{centering}
			\includegraphics[scale=1]{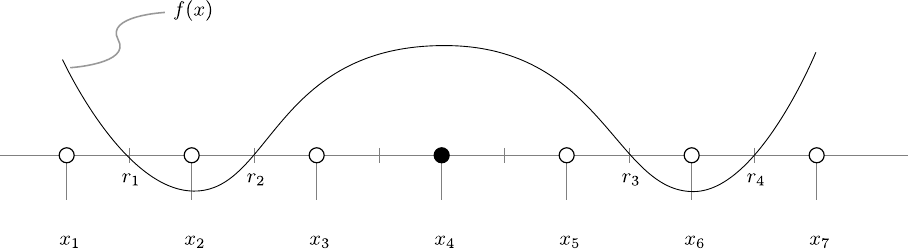}
			\par\end{centering}
		\caption{Merge involving $n=7$ points. $f$ has in total $k=n-3=4$ right
			roots. \label{fig:MergeFromF}}
	\end{figure}
	\par\end{center}

\begin{example}[Merge]
	For merges (see  \cref{fig:MergeFromF}) we only get right-roots and
	hence, we can write them as sums of $f_{0}$-assignments and obtain the solution using \cref{prop:balancedf0algeb} and \cref{prop:unbalancedf0algeb}.
	For $n$ points, the polynomial has degree $n-3$
	and so $\left\langle x\right\rangle =0$, just as expected for a
	merge.
\end{example}

This scheme fails for moves corresponding to lower bias games.
For instance, the main move of the bias $1/14$ game has its coordinates given by $x_{0},x_{1},x_{2},x_{3},x_{4},x_{5},x_{6}$
and the roots of $f$ are $l_{1},l_{2},r_{1},r_{2},r_{3}$ satisfying
$x_{0}<l_{1}<l_{2}<x_{1}<x_{2}\dots<x_{6}<r_{1}<r_{2}<r_{3}$. Here,
we can either consider $l_{1}$ to be a right root, in which case
$l_{2}$ is a left root (i.e. a root which is not a right root).
Or we can consider $l_{2}$ to be a right root in which case $l_{1}$
becomes a left root. Thus for games with bias $1/14$ and
less, we must revert to \ref{lem:generalMonomialDecomposition}, which
means we can not -- at least by this scheme -- avoid finding the solution
to all the monomial assignments.

Since we mentioned the merge move, for completeness let us consider also the split
move (see \cref{exa:split}). The situation (see  \cref{fig:SplitFromF})
is similar to that of merge but with one key distinction: the polynomial
has degree $n-2$; it has $n-3$ right roots and one left root. Thus, it can not be expressed as a
sum of $f_{0}$-assignments using \ref{lem:rightRoots}. Of course,
merges and splits by themselves are not of much interest in this discussion
because we already know that the Blinkered Unitary solves them both (see
\cref{subsec:BlinkeredUnitary}).
\begin{center}
	\begin{figure}[h]
		\begin{centering}
			\includegraphics[scale=1]{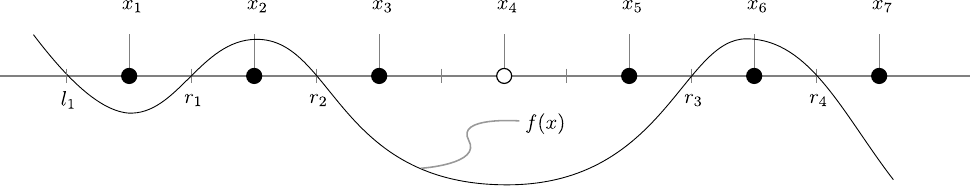}
			\par\end{centering}
		\caption{Split involving 7 points. $f$ has  $k=n-2=5$ roots; 4 right
			and one left.\label{fig:SplitFromF}}
	\end{figure}
	\par\end{center}

\end{document}